\documentclass[11pt,a4paper]{article}

\usepackage[a4paper]{geometry}
\usepackage{fullpage}

\usepackage{amssymb,latexsym,mathtools,amsfonts,amsthm}
\usepackage{amsbsy}
\usepackage{bm}
\usepackage{mathrsfs}
\usepackage{extarrows}
\usepackage{enumerate}
\usepackage{enumitem}
\usepackage{float}
\usepackage{multirow}

\usepackage{graphicx}
\usepackage{overpic}
\usepackage{tikz}
\usetikzlibrary{arrows.meta,calc,decorations,decorations.markings}

\usepackage{color}
\usepackage{cite}
\usepackage{hyperref}

\hypersetup{
    hidelinks,
    colorlinks=true,
    linkcolor=blue,
    citecolor=blue
}

\newtheorem{theorem}{Theorem}[section]
\newtheorem{lemma}[theorem]{Lemma}

\newtheorem{assume}[theorem]{Assumption}
\newtheorem{RHP}[theorem]{RH Problem}

\theoremstyle{remark}
\newtheorem{remark}[theorem]{Remark}

\numberwithin{equation}{section}
\allowdisplaybreaks
\newcommand{\Ca}{\mathcal{C}}
\newcommand{\R}{\mathbb{R}}
\newcommand{\C}{\mathbb{C}}

\newcommand{\ii}{\mathrm{i}}

\newcommand{\e}{\mathrm{e}}

\newcommand{\q}{\mathbf{q}}

\newcommand{\h}{\mathbf{h}}
\newcommand{\bu}{\mathbf{u}}
\newcommand{\bTheta}{\mathbf{\Theta}}
\newcommand{\bbu}{\boldsymbol{u}}
\newcommand{\G}{\mathbf{G}}
\newcommand{\bA}{\mathbf{A}}
\newcommand{\bQ}{\mathbf{Q}}
\newcommand{\ce}{\boldsymbol{e}}
\newcommand{\bJ}{\mathbf{J}}
\newcommand{\bbX}{\mathbf{X}}
\newcommand{\bPhi}{\mathbf{\Phi}}
\newcommand{\m}{\mathbf{m}}
\newcommand{\w}{\mathbf{w}}

\newcommand{\M}{\boldsymbol{M}}
\newcommand{\T}{\boldsymbol{T}}
\newcommand{\I}{\boldsymbol{I}}
\newcommand{\N}{\boldsymbol{N}}
\newcommand{\bH}{\boldsymbol{H}}
\newcommand{\Y}{\boldsymbol{Y}}
\newcommand{\E}{\boldsymbol{E}}
\newcommand{\V}{\boldsymbol{V}}
\newcommand{\bV}{\widetilde{\boldsymbol{V}}}
\newcommand{\bv}{\boldsymbol{v}}
\newcommand{\sig}{\boldsymbol{\sigma}}
\newcommand{\bw}{\boldsymbol{w}}
\newcommand{\bL}{\boldsymbol{L}}
\newcommand{\bX}{\boldsymbol{X}}

\newcommand{\A}{\boldsymbol{\mathcal{A}}}
\newcommand{\B}{\mathbf{B}}
\newcommand{\cD}{\mathcal D}

\newcommand{\X}{\mathcal{X}}
\newcommand{\cP}{\mathcal{P}}
\newcommand{\rP}{\mathrm{P}}
\newcommand{\rz}{\mathsf{z}}
\newcommand{\cE}{\boldsymbol{\mathcal{E}}}
\newcommand{\xcE}{\mathcal{E}}

\newcommand{\eps}{\epsilon}

\DeclareMathOperator{\diag}{diag}
\DeclareMathOperator*{\Res}{Res}
\DeclareMathOperator*{\Ai}{Ai}
\DeclareMathOperator*{\re}{Re}
\DeclareMathOperator*{\im}{Im}

\newcommand{\bPi}{\mathbf{\Pi}}

\def\be{\begin{equation}}
\def\ee{\end{equation}}
\def\bse{\begin{subequations}}
\def\ese{\end{subequations}}
\def\bpm{\begin{pmatrix}}
\def\epm{\end{pmatrix}}
\def\bi{\begin{itemize}}
\def\ei{\end{itemize}}

\title{\bf
Painlev\'e-type asymptotics for the defocusing Manakov system with nonzero boundary conditions}
\date{}

\author{\hspace{0.6 cm}{Haibing Zhang$^{a}$, Xianguo Geng$^{a,b}$\footnote{\footnotesize
 Corresponding author. {\sl Email address}: xggeng@zzu.edu.cn}, Ruomeng Li$^{a}$, Huan Liu$^{a}$}\\
\leftline{\hspace{0.6 cm}{\small{\sl $^{a}$ School of Mathematics and Statistics, Zhengzhou University, 100 Kexue Road, Zhengzhou, }}}\\
\leftline{\hspace{0.6 cm}{\small{\sl \quad Henan 450001, People's Republic of China}}}\\
\leftline{\hspace{0.6 cm}{\small{\sl $^{b}$ Institute of Mathematics, Henan Academy of Sciences, Zhengzhou, Henan 450046,}}}\\
\leftline{\hspace{0.6 cm}{\small{\sl \quad People's Republic of China}}}}

\begin{document}
\maketitle{}
\begin{abstract} The defocusing Manakov system is the two-component vector analogue of the defocusing nonlinear Schr\"odinger equation. In this paper, we study this system under parallel nonzero boundary conditions \[ \ii \q_t+\q_{xx}+2(q_0^2-\|\q\|^2)\q=0,\qquad \q(x,t)\to \q_\pm,\quad x\to\pm\infty, \] where \(\|\q_\pm\|=q_0\) and \(\q_+=\e^{\ii(\theta_+-\theta_-)}\q_-\), with \(\theta_\pm\in[0,2\pi)\). We analyze the long-time asymptotics of the solution in the transition region \[ \cP= \left\{(x,t)\in\R\times\R_+:\ \left|\frac{x}{2t}+q_0 \right|\le C t^{-2/3}\right\}, \] which separates the soliton region from the solitonless region. Under a genericity assumption at the branch point and a sufficiently rapid convergence assumption toward the nonzero background, we derive a uniform asymptotic formula for the solution in \(\cP\). Apart from the modulated background, the first correction is of order \(t^{-1/3}\) and is expressed in terms of the Hastings--McLeod solution of the Painlev\'e II equation, while the error term is of order \(\mathcal{O}(t^{-2/3}\log t)\). The proof is based on the \(3\times3\) inverse-scattering framework for the defocusing Manakov system with nonzero boundary conditions~\cite{BD2015-1} and on a Deift--Zhou steepest descent analysis of the associated \(3\times3\) Riemann--Hilbert problem.

A central feature of the present analysis is that the local parametrix near
\(-q_0\) is not given by the usual Painlev\'e II model alone.  Instead, it is a
coupled Painlev\'e--error-function model: a Painlev\'e II Riemann--Hilbert
problem and an error-function Riemann--Hilbert problem are combined in a single
local construction.  This type of coupled local construction was recently
encountered in the steepest descent analysis of the ``bad'' Boussinesq
equation~\cite{CL2024main} and arises here in the finite-density vector NLS
setting.

\end{abstract}

\vspace{0.2cm}

\vspace{0.2cm}

\tableofcontents

\section{Introduction}

The Manakov system~(see, e.g.,~\cite{M1974,APT2004}) is a fundamental example of an integrable vector nonlinear wave equation.  It arises as a two-component generalization of the nonlinear
Schr\"odinger~(NLS) equation and appears, for example, in nonlinear optics and in the
theory of multi-component Bose--Einstein condensates~\cite{AH1981,EG2003,CP1999,W1999,HCHE2011,YCHH2012}.  In this paper, we
consider the defocusing Manakov system with nonzero boundary conditions~(NZBCs)
\be \label{E:demanakovS}
        \ii \q_t+\q_{xx}+2(q_0^2-\|\q\|^2)\q=0,
        \qquad
        \q=(q_1,q_2)^{\top},
\ee
\be \label{E:bjtj}
        \lim_{x\to\pm\infty}\q(x,t)=\q_\pm,\qquad
        \|\q_\pm\|=q_0>0,\qquad
        \q_+=\q_- \e^{\ii(\theta_+-\theta_-)}.
\ee
Note that system~\eqref{E:demanakovS} here differs from the classical Manakov equation~\cite{M1974} by an additional term $2 q_0^2 \q$. One can remove this term by the simple rescaling $\q(x,t) \to \q(x,t) \e^{-2\ii q_0^2 t}$, and this term was added so that the boundary conditions~\eqref{E:bjtj} are independent of time.  The boundary condition~\eqref{E:bjtj} is usually referred to as the \emph{parallel} NZBCs.   If one component of \(\q\) vanishes identically, then
\eqref{E:demanakovS}--\eqref{E:bjtj} reduce to the scalar defocusing NLS
equation with finite-density boundary conditions; see for example~\cite{F1987,CuJe2016}.
Thus, \eqref{E:demanakovS}--\eqref{E:bjtj} may be viewed as a vector
finite-density initial-value problem.

The Manakov system is a completely integrable  system and admits a  Lax-pair representation~\cite{M1974}.  This makes it possible to study the Cauchy
problem by means of the inverse scattering transform~(IST), which may be viewed as a
nonlinear analogue of the Fourier transform.  For rapidly decaying initial data,
the IST for scalar and vector NLS equations is
classical~\cite{ZS1972,M1974,APT2004}. However, the finite-density case is much more delicate because the nonzero background changes the spectral structure,  and dark solitons may occur. The IST for the  scalar NLS equation with NZBCs  was  studied 
in~\cite{ZS1972,F1987,DPVV2012,BK2014}.   Extending the IST from the scalar to the vector case requires additional considerations. Unlike the scalar problem, not all Jost eigenfunctions are analytic in the required domains, and auxiliary eigenfunctions associated with the adjoint spectral problem are needed in order to build a complete analytic basis.  The
IST for the vector NLS system with NZBCs  was developed rigorously only
more recently; see~\cite{PAB2006,BD2015-1,BD2015-2,PBA-2011,BKP-2016}.

\begin{figure}
\centering
\begin{tikzpicture}[scale=0.91,>=latex,font=\small]

\draw[->, line width=0.9pt] (-7.2,0) -- (7.2,0) node[below] {\(x\)};
\draw[->, line width=0.9pt] (0,0) -- (0,3.8) node[left] {\(t\)};

\fill[blue!18] (0,0) -- (5.55,3.1) -- (-5.55,3.1) -- cycle;
\fill[green!18] (0,0) -- (-7,0) -- (-7,3.1)--(-6.65,3.1) -- cycle;
\fill[green!18] (0,0) -- (7,0) -- (7,3.1)--(6.65,3.1) -- cycle;

\draw[dashed, line width=0.9pt] (0,0) -- (5.55,3.1);
\draw[dashed, line width=0.9pt] (0,0) -- (6.65,3.1);
\draw[dashed, line width=0.9pt] (0,0) -- (-5.55,3.1);
\draw[dashed, line width=0.9pt] (0,0) -- (-6.65,3.1);

\draw[line width=1.2pt] (0,0) -- (6.2,3.1);
\draw[line width=1.2pt] (0,0) -- (-6.2,3.1);

\node[align=center] at (0,2.1)
{
soliton region: $|\xi| <q_0$
};

\node[align=center] at (4.85,1.0)
{
solitonless region\\
\(|\xi|>q_0\)
};

\node[align=center] at (-4.85,1.0)
{
solitonless region:\\
\(|\xi|>q_0\)
};

\node[align=center] at (6.55,3.75)
{
transition\\
\(\xi\approx q_0\)
};

\node[align=center] at (-6.55,3.75)
{
transition\\
\(\xi\approx -q_0\)
};

\node[below left] at (0,0) {\(0\)};

\end{tikzpicture}
\caption{Space-time regions for the defocusing Manakov system with \(q_0=1\).
The two solid rays correspond to \(\xi=\pm q_0\), where $\xi=x/(2t)$.  The shaded central
wedge is the soliton region, the outer domains are the
solitonless region, and the narrow bands indicated by dashed
lines represent the transition regions.}
\label{fig:regions}
\end{figure}

In the modern formulation of inverse scattering theory, the inverse problem
is usually recast as a Riemann--Hilbert (RH) problem. In the pioneering
work~\cite{pz93}, Deift and Zhou introduced the nonlinear steepest descent
method for oscillatory RH problems and thereby established the long-time
asymptotics of solutions to the modified KdV equation. Since then, the
nonlinear steepest descent method and its extensions have become fundamental
tools in the long-time asymptotic analysis of integrable systems. In particular, this approach has been successfully applied to the scalar focusing and defocusing NLS equations with vanishing boundary conditions \cite{PZ94,PZ2002,PZto94,DM2008,BJM2018}, as well as to scalar NLS equations with NZBCs \cite{BM2017,BLM2021,DP2019,DP2024,DLM2020,H2003-1,H2002,CuJe2016,WF2022,WF2023}. For comparison, let us recall the asymptotic picture for the defocusing scalar NLS equation with finite-density NZBCs. After normalizing the background amplitude to be \(q_0=1\), the solution admits the following long-time asymptotic descriptions in the
half-plane \(t\geq0\):
\begin{enumerate}
\item In the soliton region \(\left|x/(2t)\right| < 1\), one has
\[
q(x,t)=q_{\rm sol}(x,t)+\mathcal O(t^{-1}),\qquad t \to \infty,
\]
where \(q_{\rm sol}(x,t)\) denotes the modulated dark-soliton term; see
Ref.~\cite{CuJe2016}.
\item In the solitonless region \(\left|x/(2t)\right|>1\), one has
\[
q(x,t)=q_{\rm bg}+t^{-1/2}q_{\rm rad}(x,t)+\mathcal O(t^{-3/4}), \qquad t \to \infty,
\]
where $q_{\rm bg}$ denotes the modulated background, and \(q_{\rm rad}(x,t)\) denotes the modulated linear wave; see
Ref.~\cite{WF2022}.

\item In the transition regions \(\frac{x}{2t}\approx \pm1\), one has
\[
q(x,t)=q_{\rm bg}+t^{-1/3}q_{\rm p}(x,t)+\mathcal O(t^{-1/2}), \qquad t \to \infty,
\]
where \(q_{\rm bg}\) denotes the modulated background, and \(q_{\rm p}(x,t)\) can be expressed in terms of a solution of the Painlev\'e II equation; see Ref.~\cite{WF2023}. We emphasize that the \(\mathcal O(t^{-1/2})\) term is a pure error term: it is not produced by a parabolic-cylinder parametrix or by any other similar local parametrix. In fact, under sufficiently strong regularity assumptions and sufficiently rapid convergence of the initial data to the background, as imposed in the present paper, this error is expected to be improved to \(\mathcal O(t^{-2/3})\).
\end{enumerate}

Despite these advances in the scalar theory, the corresponding long-time asymptotic analysis for the defocusing Manakov system with NZBCs had remained largely open. This problem was also mentioned as an important open direction in Refs.~\cite{BD2015-1,BKP-2016,P2023}. In our recent work~\cite{GZW2025}, we made some progress in this direction.  We
showed that the solution of \eqref{E:demanakovS}--\eqref{E:bjtj} admits an asymptotic picture analogous to that in the scalar finite-density NLS case.  More precisely, the \((x,t)\)-half-plane can be divided into three
regions:
\[
\left|\frac{x}{2t}\right|<q_0,\qquad
\left|\frac{x}{2t}\right|>q_0,\qquad
\frac{x}{2t}\approx \pm q_0,
\]
which are referred to as the soliton region, the solitonless region, and the transition regions, respectively; see Fig.~\ref{fig:regions}.
However, the structure and features of the asymptotic formula in the vector case
are not merely a direct copy of its scalar counterpart.
 In~\cite{GZW2025}, we showed that the long-time behavior in the soliton region differs substantially from that in the scalar case. In particular, the soliton term propagates together with an \(\mathcal O(t^{-1/2})\) radiation term, whereas no such radiation term appears in the corresponding scalar asymptotics. More precisely, for \(|x/(2t)|<q_0\), we obtained
\[
\q(x,t)
=
\q_{\rm sol}(x,t)
+t^{-1/2}\q_{\rm rad}(x,t)
+\mathcal O(t^{-1}\log t),
\qquad t\to\infty .
\]
This radiation term can be interpreted as an intrinsic vectorial effect.
The method developed for the soliton region can be extended to the solitonless region \(|\xi|>q_0\) without introducing any essential new difficulty. Indeed, after introducing suitable conjugations and factorizations, as in the soliton-region analysis, one performs the lens-opening transformation and reduces the corresponding RH problem to a parabolic-cylinder parametrix similar to the one appearing in the soliton region. Consequently, the asymptotic formula in the solitonless region consists of a modulated background together with an \(\mathcal O(t^{-1/2})\) dispersive radiation term, in agreement with the scalar finite-density case.

In order to obtain a global description of the solution for the vector finite-density initial-value problem, it remains to analyze the asymptotic behavior in the two transition regions. This leads naturally to the following questions: is the long-time asymptotic behavior in the transition regions still analogous to that in the scalar case? More importantly, is the underlying Deift--Zhou steepest descent mechanism also analogous?

The main purpose of the present paper is to answer these two questions.  More
precisely, we analyze the long-time asymptotics of the solution in the
transition regions.  Our analysis leads to the following interesting observation. The final long-time asymptotic formula in the
transition regions have the same structural form as in the scalar finite-density NLS equation, but the mechanism behind it is essentially different. 
We explain this point as follows.
For the scalar finite-density NLS equation, the corresponding local parametrix
can be matched directly to the standard Painlev\'e II RH
problem~\cite{WF2023}. In the present vector problem, this is no longer
sufficient. The local construction near \(-q_0\) involves a coupled Painlev\'e--error-function model: one part is governed by the standard Painlev\'e II RH problem and produces the Hastings--McLeod function, while the other part is governed by an error-function RH problem. This is precisely why we say that the mechanism is different. In the large-\(z\)
expansion of the RH solution, as \(t\to\infty\), the
\(1/z\)-coefficient contains, besides the background and the
\(\mathcal O(t^{-1/3})\) Painlev\'e contribution, an additional
\(\mathcal O(t^{-1/2})\) term before the final error term appears. This is in
sharp contrast with the scalar case, where the expansion contains only the
background, the \(\mathcal O(t^{-1/3})\) Painlev\'e contribution, and the error
term. Surprisingly, however, the \(\mathcal O(t^{-1/2})\) contribution produced
by the error-function model does not contribute to the final asymptotic formula for the finite-density vector initial-value problem. 
As a consequence, the final result still has the same form as in the scalar
case.

Interestingly, such a coupled model was recently introduced by Charlier and
Lenells in their steepest descent analysis of the ``bad'' Boussinesq
equation~\cite{CL2024main}.  We show that the same type of local model also
arises in the defocusing Manakov system with NZBCs.  To the best of our
knowledge, after the work of Charlier and Lenells on the Boussinesq equation,
this provides another example in integrable systems where a Painlev\'e II model
and an error-function model are coupled within a single local parametrix.

 Since the left and right transition regions can be treated
in a similar way, we focus on the left transition region
\be \label{E:PL}
\cP=
\left\{(x,t)\in\R\times\R_+:\ |\xi+q_0|\le Ct^{-2/3}\right\},
\qquad
\xi=\frac{x}{2t},
\ee
where \(C>0\) is fixed. The scaling in~\eqref{E:PL} is consistent with the transition scale that appears in the scalar finite-density NLS problem~\cite{WF2023}.
Indeed, as \(\xi\) approaches \(-q_0\), the stationary points of the phase
functions approach the branch point \(-q_0\).  Hence the usual quadratic
approximation of the phase, which leads to parabolic-cylinder parametrices, is no longer sufficient.  Under the scaling
\[
z+q_0=\mathcal O(t^{-1/3}),\qquad
(\xi+q_0)t^{2/3}=\mathcal O(1),
\]
the cubic term in \(z+q_0\) and the perturbation caused by \(\xi+q_0\) are balanced, which gives a Painlev\'e-type local parametrix.

Finally, we mention that the Deift--Zhou steepest descent method for RH
problems has also been used to study the long-time behavior of many other
integrable systems; see, for instance,
\cite{pz93,CLPa2020,HZ2022,WF2023,WXF2025,BIS2010,XYZ2024,HC2020,
WangZhu2025BoussinesqPainleve,WangZhu2025SK}.  The present work follows this
general line, but the finite-density vector setting leads to a different local
analysis in the transition region.

\subsection{Main results} 
We now state the main result. The assumptions used below are formulated precisely in Section~\ref{S:reIST}. In brief, they require rapid decay of the perturbation from the nonzero background and exclude nongeneric degeneracies at the branch points. 

Let us first briefly recall the inverse-scattering framework for the defocusing Manakov system with NZBCs; see Section~\ref{S:reIST} for more details. Starting from initial data satisfying the above assumptions, one obtains a set of scattering data through the direct scattering transform. We denote the scattering data generated by the corresponding initial data by 
\be \label{E:shsjtlyy} 
\sigma_d= \left\{ r_1(z),\ r_2(z),\ r_3(z),\ \{ \zeta_j,\tau_j\}_{j=1}^{N_1},\ \{ z_j,\kappa_j\}_{j=1}^{N_2} \right\}.
 \ee
 The time dependence of these scattering data is determined by the time part of the Lax pair for the Manakov system. The recovery of the potential from the time-dependent scattering data is called the inverse scattering transform. In the modern formulation of the inverse scattering method, this recovery is usually carried out by encoding the scattering data into an appropriate RH problem and then reconstructing the potential from the solution of that RH problem. A main result of~\cite{BD2015-1} shows that the solution of \eqref{E:demanakovS}--\eqref{E:bjtj} can be represented in terms of the solution of a \(3\times3\) RH problem. The main purpose of the present paper is to develop a suitable Deift--Zhou steepest descent analysis for this RH problem in the transition region \(\cP\), and thereby to derive the transition asymptotics for the vector finite-density initial-value problem. 

The main theorem of this paper is stated as follows.

\begin{theorem}[Painlev\'e asymptotics in \(\cP\)]\label{Th:main}
Let \(\q(x,t)\) be a smooth solution of the defocusing Manakov system \eqref{E:demanakovS} satisfying the NZBCs~\eqref{E:bjtj}. Assume that the initial data satisfy Assumptions~\ref{As:gesa} and~\ref{As:1}, and let \(\sigma_d\), given by~\eqref{E:shsjtlyy}, be the corresponding scattering data. Then, as \(t\to+\infty\),
\be \label{E:asygs}
\q(x,t)
=
\e^{\ii \aleph}
\left(1+t^{-1/3}f_{\rm P}(x,t)\right)\q_+
+\mathcal{O}(t^{-2/3}\log t),
\ee
uniformly for \((x,t)\in\cP\).  Here the phase \(\aleph\) is given by
\[
\begin{aligned}
\aleph
=&\frac{1}{2\pi}\int_{0}^{\infty}
s^{-1}
\log\biggl(
1-\frac{1}{\gamma(s)}|r_1(s)|^2-|r_2(s)|^2
\biggr)\,\mathrm{d}s  \\
&+\frac{1}{2\pi}\int_{-\infty}^{-q_0}
s^{-1}
\log\left(
1+\frac{1}{\gamma(s)}|r_3(s)|^2
\right)\,\mathrm{d}s+\sum_{j=1}^{N_1}2\arg \zeta_j
+\sum_{j=1}^{N_2}2\arg z_j 
\end{aligned}
\]
with $\gamma(z)=1-\frac{q_0^2}{z^2}$,
and the bounded function \(f_{\rm P}(x,t)\) is defined by
\[
f_{\rm P}(x,t)
=
\frac{\ii}{2q_0 (\frac{3}{4q_0})^{1/3}}
\left[
\int_y^\infty u_{ HM}^2(s)\,\mathrm{d}s+u_{ HM}(y)
\right],
\qquad
y=2\left(\frac{4q_0}{3}\right)^{1/3}(\xi+q_0)t^{2/3}.
\]
Here \(u_{ HM}\) denotes the Hastings--McLeod solution of the Painlev\'e II
equation
\[
        u''(\tau)=\tau u(\tau)+2u^3(\tau),
\]
characterized by
\be \label{E:uHMbjtj}
u_{HM}(\tau)=
\begin{cases}
\Ai(\tau)(1+o(1)), & \tau\to+\infty,\\[1mm]
\sqrt{-\tau/2}\,(1+o(1)), & \tau\to-\infty.
\end{cases}
\ee
\end{theorem}

The proof of Theorem~\ref{Th:main}  is  presented  in Section~\ref{S:dza}.  The theorem
shows that, apart from the modulated background, the first correction is of order \(t^{-1/3}\), with a scalar coefficient governed by the Painlev\'e II transcendent.  This indicates that
the asymptotic structure of the solution in the vector problem is similar to
that in the scalar case.  Nevertheless, several features of our asymptotic formula are specific to the vector setting and differ from their scalar counterparts.

First, the modulation phase \(\aleph\) contains the term
$
        \sum_{j=1}^{N_2}2\arg z_j 
$,
which has no counterpart in the scalar modulation phase; see
\cite[Eq.~(1.6)]{WF2023}.  The reason is that the discrete spectral points
\(\{z_j\}_{j=1}^{N_2}\) correspond to dark--bright soliton components.  Such
dark--bright soliton solutions are specific to the vector NLS system and do
not occur in the scalar finite-density NLS equation.  Their presence produces
an additional phase modulation, which is precisely the contribution represented
by \(\sum_{j=1}^{N_2}2\arg z_j\).
Second, the error estimate also reveals a distinction between the scalar and
vector transition analyses. In the scalar finite-density NLS problem~\cite{WF2023}, the
transition asymptotics are governed by a pure Painlev\'e II model. 
Under sufficiently strong regularity assumptions and sufficiently rapid
convergence of the initial data to the background, as imposed in the present
paper, the error estimate in~\cite[Eq.~(1.5)]{WF2023} is expected to be
improved to order \(\mathcal O(t^{-2/3})\).   In contrast, the error term obtained in Theorem~\ref{Th:main} for the 
Manakov system is \(\mathcal O(t^{-2/3}\log t)\).  The additional logarithmic factor  comes from the coupled Painlev\'e--error-function
local parametrix near the branch point and reflects the higher-dimensional
\(3\times3\) structure of the Manakov RH problem.

We conclude this section with several remarks.

\begin{remark}
The function \(f_{\rm P}(x,t)\) is uniformly bounded in \(\cP\).  This follows
directly from the fact that the scaled variable $y$
remains bounded for \((x,t)\in\cP\), together with the smoothness of the
Hastings--McLeod solution \(u_{ HM}\).
\end{remark}

\begin{remark}
Theorem~\ref{Th:main} is formulated under a genericity assumption at the branch
point.  If this genericity condition fails, the local model is modified, while
the overall steepest descent strategy remains the same.  In that case, one
obtains an analogous transition formula in which the Hastings--McLeod solution
is replaced by an Ablowitz--Segur solution of the Painlev\'e II equation.
\end{remark}

\noindent{\bf Organization of the paper.}
Section~\ref{S:reIST} recalls the inverse scattering and RH
formulation for the defocusing Manakov system with NZBCs.
Section~\ref{S:mp} is devoted to the coupled Painlev\'e--error-function model
problem.  In Section~\ref{S:dza}, we perform the nonlinear steepest descent
analysis.  The proof of
Theorem~\ref{Th:main} is completed in Section~\ref{S:pthmain}.  Finally, some technical details are presented in  Appendices~\ref{App:AAAA1} to \ref{App:AAA2}.

\noindent{\bf Notation.}
The superscripts \(\top\) and \(\dagger\) denote transpose and Hermitian
conjugation, respectively.  The asterisk denotes complex conjugation.  The
symbols \(C>0\) and \(c>0\) denote generic constants whose values may change
from line to line.  Unless otherwise stated, $\log (z)$ always denotes
the principal branch of the logarithm. We write \(\R_+=(0,+\infty)\), \(\R_-=(-\infty,0)\), and
denote the upper and lower half-planes by \(\C_+\) and \(\C_-\).  For a
matrix-valued function on a contour, all \(L^p\)-norms are understood
entrywise.




\section{The IST for the defocusing Manakov system with NZBCs}\label{S:reIST}
In this section, we  review the IST of the defocusing Manakov system~\eqref{E:demanakovS} with NZBCs~\eqref{E:bjtj}. Since these results have been well established in Ref.~\cite{BD2015-1}, we will omit the proof. Throughout this paper, we assume that the potential $\q(x,t)$ is sufficiently smooth and decays rapidly to the non-zero background~\eqref{E:bjtj}. The potentials we consider constitute merely a subclass of those in~\cite{BD2015-1} for which the inverse scattering analysis can proceed smoothly. This is because our primary objective is to reveal the asymptotic characteristics of the solutions, and thus we do not take into account solutions with low regularity. 

As usual, the IST for an integrable system is based on its formulation in terms of a Lax pair. It is well-known that the defocusing Manakov system~\eqref{E:demanakovS}  possesses a $3\times 3$ matrix Lax pair~\cite{BD2015-1}:
\be \label{E:laxp1}
\bPhi_x=\widehat{\bbX} \bPhi, \qquad \bPhi_t=\widehat{\mathbf{T}} \bPhi,
\ee
where
\begin{align*}
&\widehat{\bbX}(x,t,k)=-\ii k \bJ +\bQ, \qquad \widehat{\mathbf{T}}(x,t,k)=2 \ii k^2 \bJ- \ii \bJ \left(\bQ_x-\bQ^2+q_0^2 \right)-2k \bQ,\\
&\bJ=\bpm 1& \mathbf{0}^{\top}\\
\mathbf{0} &-\I_{2 \times 2}
\epm,  \quad
\bQ=\bpm
0& \q^{\dagger }\\
\q &  \mathbf{0}_{2 \times 2}
\epm.
\end{align*}
Compared with the case of zero boundary conditions, the inverse-scattering
analysis with NZBCs is more involved because the spectral parameter \(k\) is
coupled with the function
\[
        \lambda(k)=\sqrt{k^2-q_0^2}.
\]
The branch points are located at \(k=\pm q_0\).  We take the branch cut to be
\((-\infty,-q_0]\cup[q_0,\infty)\) and choose the branch of \(\lambda\) by the
condition \(\lambda(0)=\ii q_0\).  Following~\cite{BD2015-1}, we introduce the
uniformization variable
\[
        z=k+\lambda.
\]
Then the inverse transformation is given by
\begin{equation}\label{E:intr}
        k=\frac{1}{2}\left(z+\frac{q_0^2}{z}\right),
        \qquad
        \lambda=\frac{1}{2}\left(z-\frac{q_0^2}{z}\right).
\end{equation}
In what follows, the inverse-scattering analysis will be carried out in the
\(z\)-plane.

%
%

Following the notation in Ref.~\cite{BD2015-1}, we denote the orthogonal vector of a two-component complex-valued vector $\bv=(v_1,v_2)$ as $\bv^{\perp}=\left(v_2,-v_1  \right)^{\dagger}$. We then introduce  three matrices
\be \label{E:trmatrix}
\E_{\pm}(z)=
\bpm
1&0&-\frac{\ii q_0}{z}\\
\ii \frac{\q_{\pm}}{z} & \frac{\q_{\pm}^{\perp}}{q_0}&\frac{\q_{\pm}}{q_0}
\epm, \quad
 \mathbf{\Lambda}(z)=\mathrm{diag} \left( -\lambda,  k,  \lambda  \right),\quad
\mathbf{\Omega}(z)=\mathrm{diag} \left( -2k \lambda,  k^2+\lambda^2,  2 k \lambda  \right),
\ee
which satisfy the relation
$$
\E_{\pm}^{-1} \widehat{\bbX}_{\pm} \E_{\pm}=\ii \mathbf{\Lambda},\qquad
\E_{\pm}^{-1} \widehat{\mathbf{T}}_{\pm} \E_{\pm}=-\ii \mathbf{\Omega},
$$
where  $\widehat{\bbX}_{\pm}= \lim_{x \to \pm \infty} \widehat{\bbX}$ and $\widehat{\mathbf{T}}_{\pm}=\lim_{x \to \pm \infty} \widehat{\mathbf{T}}$.
Let
$$
\Delta \bQ_{\pm}(x,t)=\bQ(x,t)-\bQ_{\pm}, \quad \bQ_{\pm}=\bpm
0& \q_{\pm}^{\dagger }\\
\q_{\pm} &  \mathbf{0}_{2 \times 2}
\epm.
$$
We define the Jost eigenfunctions as the unique solutions to the following integral equations:
	\begin{align*}
		&\boldsymbol{\mu}_-(x,t,z) = \E_-(z) + \int_{-\infty}^{x} \E_{-}(z) \e^{\ii   (x-y) \mathbf{\Lambda}(z)} \E^{-1}_{-}(z) \Delta \bQ_{-}(y,t) \boldsymbol{\mu}_-(y,t,z) \e^{-\ii   (x-y) \mathbf{\Lambda}(z)}  \mathrm{d}y, \\
		&\boldsymbol{\mu}_+(x,t,z) = \E_+(z) - \int_{x}^{+\infty} \E_{+}(z) e^{\ii   (x-y) \mathbf{\Lambda}(z)} \E^{-1}_{+}(z) \Delta \bQ_{+}(y,t) \boldsymbol{\mu}_+(y,t,z) \e^{-\ii   (x-y) \mathbf{\Lambda}(z)} \mathrm{d}y, 
	\end{align*}
It was shown in~\cite[Section~2.2]{BD2015-1} that, provided
\(\Delta \mathbf Q_\pm\) decay sufficiently rapidly as \(x\to\pm\infty\), the modified
eigenfunctions \(\boldsymbol{\mu}_\pm\) are well defined at least for
\(z\in\R\setminus\{0,\pm q_0\}\), and certain columns of \(\boldsymbol{\mu}_\pm\) admit
analytic continuations away from the real axis.

Let $ \bPhi_{\pm}(x,t,z)=\boldsymbol{\mu}_{\pm}(x,t,z) \e^{\ii \mathbf{\Lambda}(z)x - \ii \mathbf{\Omega}(z)t}$. Then, $\bPhi_{\pm}(x,t,z)$  are the fundamental solutions of  Lax pair~\eqref{E:laxp1} for $z \in \R \setminus \{0,\pm q_0 \}$. This is because
\be \label{E:degamma}
\mathrm{det}  \boldsymbol{\mu}_{\pm}(x,t,z)=\mathrm{det} \E_{\pm}(z) =1-\frac{q_0^2}{z^2}=:\gamma(z).
\ee
Therefore, there exists a matrix $\bA(z)$ independent of  $x$  and  $t$  such that
\be \label{E:deA}
\bPhi_-(x,t,z)=\bPhi_+(x,t,z) \bA(z), \quad  z \in \R \setminus \{0, \pm q_0 \}.
\ee
 Defining  $\B(z) = \bA^{-1}(z)$ ,  we let $a_{ij}$ and $b_{ij}$ denote the $(ij)$-th entries of the scattering matrices $\bA(z)$ and $\B(z)$, respectively. 

We need to impose some assumptions on the scattering coefficient $a_{11}(z)$ to ensure the subsequent analysis proceeds smoothly.

\begin{assume}\label{As:gesa}
 The spectral function $a_{11}(z)$  possesses the following properties:
\bi
\item[$\mathrm{(a)}$]
The zeros of $a_{11}(z)$ are all simple, finite in number, and none of them lie on the real axis.
\item[$\mathrm{(b)}$]  $a_{11}(z)$  exhibits general behavior as $z$ approaches the branch points $\pm q_0$, specifically
\be \label{E:a11qx}
\lim_{z \to \pm q_0} (z\mp q_0)  a_{11}(z) \ne 0.
\ee
\ei
\end{assume}
We provide some brief comments on the above assumption:
$\mathrm{(i)}.$ The assumption $\mathrm{(a)}$ actually achieves two goals: Firstly, it excludes the existence of the so-called spectral singularities; Secondly, it rules out  higher-order poles solutions of the defocusing Manakov system. Although double-pole  solutions for the system are considered in~\cite[section 4]{BD2015-1}, such cases fall outside the scope of the present work. $\mathrm{(ii)}.$ As pointed out in Ref.~\cite{BD2015-1}, all simple zeros of  $a_{11}(z)$  lie either on the upper semicircle
$
C_0 = \{ z \in \mathbb{C}: \  |z| = q_0,\  \im z > 0 \}
$
or strictly inside it. We denote these two sets of zeros by  $\{ \zeta_j \}_{j=1}^{N_1}$  on  $C_0$ and  $\{z_j \}_{j=1}^{N_2}$ inside $C_0$, respectively. $\mathrm{(iii)}.$ Assume that equation \eqref{E:a11qx} holds, which implies that  $a_{11}(z)$ possesses first-order singularities at the branch points $\pm q_0$.
As indicated in the inverse scattering analysis of scalar finite-density NLS problem, this represents the generic case~(see~\cite[Appendix C]{CuJe2016}).

Now let us define the reflection coefficients  $r_{1}(z)$,  $r_2(z)$ and $r_3(z)$ by
\be \label{E:fsxsr12}
r_1(z)=\frac{a_{21}(z)}{a_{11}(z)}, \quad r_2(z)=\frac{a_{31}(z)}{a_{11}(z)}, \quad
r_{3}(z)=\frac{a_{23}(z)}{a_{33}(z)}.
\ee
In fact, the reflection coefficients $\{r_j \}_{j=1}^3$ depend only on the initial data $\q(x,0)$, since we may set $t = 0$ in~\eqref{E:deA}.  Moreover, one can observe that  we have only two independent reflection coefficients. This is because   $r_1(z) =\frac{\ii q_0}{z} r_3(\hat{z})$, where $\hat{z}=\frac{q_0^2}{z}$.
To simplify the Deift–Zhou steepest-descent analysis in the subsequent sections, we require  that the reflection coefficients admit analytic continuation off the real axis.
This can be achieved if we assume that the initial value $\q(x,0)$ satisfy the following assumption. 

\begin{assume}\label{As:1}
Suppose that there exists a constant $\varepsilon >0$ such that
\be \label{E:aszssj}
\partial_x^j  \left( \q(x,0)-\q_{\pm} \right) \e^{\pm 2 \varepsilon x}   \in L^1(\R_{\pm}), \qquad  j=0,1,2,3.
\ee
\end{assume}

\begin{remark}
For more general initial data, the reflection coefficients in general do not
admit analytic continuations away from the real axis.  This difficulty can be
overcome by rational approximation techniques; see
\cite{pz93,Lenells2017}.  The exponential convergence assumption is therefore
imposed here mainly for technical convenience.
\end{remark}

Under Assumptions~\ref{As:1} and~\ref{As:gesa}, the reflection coefficients enjoy `` good" properties. Define  $S_{\varepsilon }=\{z \in \C \ :\   |\im z| < \varepsilon \}\setminus (B_1 \cup B_2)$ , where  $B_1$  and  $B_2$  denote the disks centered at  $\frac{\ii}{2}$  and  $-\frac{\ii}{2}$, respectively, each with radius  $\frac{1}{2}$.
Then we have the following lemma.
\begin{lemma}\label{L:fsxsxz}
Suppose the initial data $\q(x,0)$ satisfies Assumption~\ref{As:gesa} and~\ref{As:1}.
Then the associated reflection coefficients defined by~\eqref{E:fsxsr12} have the following properties:
\bi
\item
 Each  $r_j(z)$  is analytic for  $z \in S_{\varepsilon } \setminus \{ 0, \pm q_0 \}$, and as  $z \to \infty$ within  $S_{\varepsilon }$, we have
\be
r_{1}(z)=\mathcal{O}(\frac{1}{z}), \quad
r_3(z)=\mathcal{O}(\frac{1}{z}), \quad
r_{2}(z)=\mathcal{O}(\frac{1}{z}).
\ee

\item  
Then the functions $\{r_j(z) \}_{j=1}^3$ have well-defined limits as  $ S_{\varepsilon } \ni z \to 0$.  Furthermore, we have
\be \label{E:r123z0}
r_2(z) = \mathcal{O}(z^2), \quad r_3(z)=\mathcal{O}(z^2), \quad  S_{\varepsilon } \ni z \to 0.
\ee

\item The functions  $\{r_j(z) \}_{j=1}^3$ have well-defined limits at the branch points $\pm q_0$. Furthermore, we have
\be \label{E:asyratzd}
\lim_{z \to \pm q_0} r_2(z)= \mp \ii, \quad \lim_{z \to \pm q_0} r_1(z)=\lim_{z \to \pm q_0} r_3(z)=0.
\ee
\ei
\end{lemma}
\begin{remark}
The proof of this lemma follows the argument in~\cite[Appendix A.1 and proof of Lemma~2.2]{GZW2025}.
We point out, however, that the estimates for \(r_2(z)\) and \(r_3(z)\) as
\(z\to0\) are improved here to order \(\mathcal O(z^2)\).  This improvement is
a consequence of the stronger assumptions imposed in the present paper.  More
precisely, the third derivative of the initial perturbation is still assumed
to decay exponentially, which allows one to perform one additional integration
by parts in the estimates of the scattering coefficients.  This yields the
stronger small-\(z\) bounds for \(r_2(z)\) and \(r_3(z)\).

\end{remark}

Following the idea from Ref.~\cite{BD2015-1}, we define a piecewise meromorphic function $\M(x,t,z)$ as follows~(see~\cite[Eq.(3.1a) and Eq.(3.1b)]{BD2015-1}):
\be \label{E:exM}
\M(x,t,z)=\begin{cases}
\left(\frac{\boldsymbol{\mu}_{-1}}{a_{11}}, \frac{\m}{b_{33}}, \boldsymbol{\mu}_{+3}   \right), & \Im z >0,\\
\left(\boldsymbol{\mu}_{+1}, -\frac{\bar{\m} }{b_{11}},   \frac{\boldsymbol{\mu}_{-3}}{a_{33}}  \right), & \Im z <0,
\end{cases}
\ee
where $
\bar{\m}(x,t,z)=-\bJ[\bPhi_{-1}^*  \times \bPhi_{+3}^*] (x,t,z^*) / \gamma(z),\
\m(x,t,z)=-\bJ[\bPhi_{-3}^*  \times \bPhi_{+1}^*] (x,t,z^*) / \gamma(z)$. Here $``\times"$ denotes the usual cross product.
Then, according to~\cite[Lemma~3.3]{BD2015-1}, the function
\(\M(x,t,z)\) defined by~\eqref{E:exM} satisfies the following residue
conditions at the two types of discrete eigenvalues
\(\{\zeta_j\}_{j=1}^{N_1}\) and \(\{z_j\}_{j=1}^{N_2}\).

\begin{enumerate}
\item[\(\mathrm{(i)}\)] For each \(\zeta_j\), \(j=1,\ldots,N_1\),
\begin{equation}\label{E:mlstjzc}
\mathop{\mathrm{Res}}\limits_{z=\zeta_j}\M(x,t,z)
=
\lim_{z\to\zeta_j}\M(x,t,z)
\begin{pmatrix}
0&0&0\\
0&0&0\\
\tau_j\e^{\theta_{31}(x,t,\zeta_j)}&0&0
\end{pmatrix},
\end{equation}
where \(\tau_j/\zeta_j\in\R\).

\item[\(\mathrm{(ii)}\)] For each \(z_j\), \(j=1,\ldots,N_2\),
\begin{equation}\label{E:mlstjzc1}
\begin{aligned}
\mathop{\mathrm{Res}}\limits_{z=z_j}\M(x,t,z)
&=
\lim_{z\to z_j}\M(x,t,z)
\begin{pmatrix}
0&0&0\\
\kappa_j\e^{\theta_{21}(x,t,z_j)}&0&0\\
0&0&0
\end{pmatrix},\\
\mathop{\mathrm{Res}}\limits_{z=z_j^*}\M(x,t,z)
&=
\lim_{z\to z_j^*}\M(x,t,z)
\begin{pmatrix}
0&\dfrac{\kappa_j^*}{\gamma(z_j^*)}
\e^{\theta_{12}(x,t,z_j^*)}&0\\
0&0&0\\
0&0&0
\end{pmatrix},
\end{aligned}
\end{equation}
where \(\kappa_j\) is a complex constant.
\end{enumerate}

In~\eqref{E:mlstjzc} and~\eqref{E:mlstjzc1}, we use the notation
\[
        \theta_{i\ell}(x,t,z)
        =
        \theta_i(x,t,z)-\theta_\ell(x,t,z),
        \qquad 1\le i,\ell\le 3,
\]
where \(\{\theta_i\}_{i=1}^3\) are given by~\eqref{E:theta123}.  Thus the
scattering data associated with the initial data \(\q(x,0)\) are written as
\[
\sigma_d=
\left\{
r_1(z),\ r_2(z),\ r_3(z),\
\{\zeta_j,\tau_j\}_{j=1}^{N_1},\
\{z_j,\kappa_j\}_{j=1}^{N_2}
\right\}.
\]

 The main result of Ref.~\cite{BD2015-1} provides a RH characterization for the solution of the defocusing Manakov system~\eqref{E:demanakovS} with NZBCs~\eqref{E:bjtj}.  Given the scattering data  $\sigma_d$, one can define the jump matrix  $\bV$  as follows:
\be \label{E:Vex}
\bV(x,t,z)=\e^{\bTheta(x,t,z)}
\bpm
1-\frac{1}{\gamma(z)}|r_1(z)|^2-|r_2(z)|^2& \frac{1}{\gamma(z)}(-r_1(z)+r_2(z)r_3(z))^*& -r^*_2(z) \\
r_1(z)-r_2(z)r_3(z) &1+\frac{1}{\gamma(z)}|r_3(z)|^2& -r_3(z)\\
r_2(z) & -\frac{1}{\gamma(z)}r_3^*(z)& 1
\epm
\e^{-\bTheta(x,t,z)},
\ee
where $\bTheta(x,t,z)=\mathrm{diag}\left(\theta_1(x,t,z),\theta_2(x,t,z),\theta_3(x,t,z)   \right)$ with
\be \label{E:theta123}
\begin{aligned}
\theta_1(x,t,z)&=- \ii \lambda(z)x+2 \ii k(z) \lambda(z) t,\\
\theta_2(x,t,z)&=\ii k(z)x-\ii (k^2(z)+\lambda^2(z))t,\\
\theta_3(x,t,z)&=\ii \lambda(z)x-2 \ii k(z) \lambda(z) t.
\end{aligned}
\ee
Using this time-dependent jump matrix, we formulate the following RH problem.

\begin{RHP}\label{RHP:sig}

Find a $3 \times 3$ matrix-valued function $\M(x,t,z)$ with the following properties:
\bi
\item $\M(x,t,\cdot) : \mathbb{C}\setminus  (  \R \cup  \mathcal{Z})  \to \mathbb{C}^{3 \times 3}$ is analytic, where
\be \label{E:defZ}
\mathcal{Z}=\mathcal{Z}_1 \cup \mathcal{Z}_2, \   \text{with} \ \ \mathcal{Z}_1=\cup_{j=1}^{N_1}  \left(\zeta_j \cup  \zeta_j^* \right) , \quad \mathcal{Z}_2= \cup_{j=1}^{N_2}\left(z_j \cup z_j^* \cup \hat{z}_j \cup \hat{z}_j^*  \right).
\ee
Across $\R$, $\M(x,t,z)$ satisfies the jump condition:
\be \label{E:Jump}
\M_+(x,t,z)=\M_-(x,t,z) \bV(x,t,z), \quad z \in \R \setminus \{0 \}.
\ee
\item  $\M(x,t,z)$ admits the asymptotic behavior:
\be \label{E:astj}
 \M(x,t,z)=\I+\mathcal{O}(\frac{1}{z}), \quad z \to \infty; \quad \M(x,t,z) =\frac{1}{z} \sig_1+ \mathcal{O}(1), \quad z \to 0,
\ee
where
\be \label{E:M0infty}
\sig_1=\bpm
0&0&-\ii q_0\\
0&0&0\\
\ii q_0&0&0
\epm.
\ee
\item  $\M(x,t,z)$ satisfies the growth conditions near the branch points $\pm q_0$:
\be \label{E:gcc}
\begin{cases}
\M_1(x,t,z)=\mathcal{O}(z \mp q_0), & z \in \C_+ \to \pm q_0,\\
\M_3(x,t,z)=\mathcal{O}(z \mp q_0), & z \in \C_- \to \pm q_0.
\end{cases}
\ee
\item $\M(x,t,z)$ satisfies the symmetries
\begin{equation}\label{E:RHP11}
\M(x,t,z)=\M(x,t,\hat{z}) \bPi(z), \quad
(\M^{-1})^{\top}(x,t,z)=-\frac{1}{\gamma(z)}\bJ \M^*(x,t,z^*) \mathbf{\Gamma}(z),
\end{equation}
where $\hat{z}=\frac{q_0^2}{z}$ and
\be \label{E:pi}
\bPi(z)= \bpm
0&0&-\ii \frac{q_0}{z}\\
0&1&0\\
\ii \frac{q_0}{z}&0&0
\epm, \quad
\mathbf{\Gamma}(z)=\bpm
-1 & 0&0\\
0& \gamma(z)&0\\
0&0&1
\epm.
\ee
\item
$\M(x,t,z)$ satisfies the residue conditions~\eqref{E:mlstjzc}and~\eqref{E:mlstjzc1} at the discrete spectral points  $\zeta_j$,  $z_j$, and  $z_j^*$, respectively.
Moreover, the residue conditions at the discrete spectral points $\zeta_j^*$, $\hat{z}_j$ and $\hat{z}_j^*$ can be derived from the $z \to \hat{z}$ symmetry in~\eqref{E:RHP11}.
\ei
\end{RHP}

\begin{remark}
The RH problem formulated here differs slightly from the one in~\cite{BD2015-1}.
In the present paper, we normalize the solution to the identity matrix \(\I\) as
\(z\to\infty\).  
\end{remark}

\begin{theorem}[RH characterization of the solution]\label{Th:RHch}
If $\q(x,t)$ is a solution to the defocusing Manakov system~\eqref{E:demanakovS} that decays rapidly to the nonzero background~\eqref{E:bjtj},  and its initial data satisfy Assumptions~\ref{As:gesa} and~\ref{As:1}, then the RH problem~\ref{RHP:sig} admits a unique solution $\M(x,t,z)$. Moreover, the solution $\q(x,t)$ can be recovered as follows:
\be \label{E:cggs}
\q(x,t)= \widetilde{\M}_{\infty}   \left( -\ii \lim_{z \to \infty}z  \bpm  \M_{21}(x,t,z)\\
 \M_{31}(x,t,z)  \epm    \right),
\ee
where  $\M_{ij}$ denotes the $(ij)$-entry of the matrix-valued function $\M$, and $\widetilde{\M}_{\infty}$ is defined by
$$
\widetilde{\M}_{\infty}=\begin{pmatrix}
\frac{q_{2,+}^*}{q_0}  & \frac{q_{1,+}}{q_0}\\
 -\frac{q_{1,+}^*}{q_0} &\frac{q_{2,+}}{q_0}
\end{pmatrix}.
$$
\end{theorem}
\begin{proof}
Let $\M(x,t,z)$ be the solution of RH problem~\ref{RHP:sig}, then the function
$$
\widehat{\M}(x,t,z)=\M_{\infty} \M(x,t,z), \quad \M_{\infty}=\bpm
1&0&0\\
\mathbf{0}& \q_+^{\perp}/q_0& \q_+/q_0
\epm,
$$
satisfies  the RH problem in~\cite[section 3.1 ]{BD2015-1}. Therefore, the theorem is a direct consequence of~\cite[Theorem 3.8]{BD2015-1}.
\end{proof}
\section{ Analysis of the model RH problem} \label{S:mp}
Before carrying out the Deift-Zhou steepest descent analysis, we first investigate the model RH problem used in this paper.  The model problem  is a coupling of the Painlev\'e II model problem and an error function model problem. We should mention that a similar model problem first appeared in the work of Charlier and Lenells on  ``bad" Boussinesq equation~\cite{CL2024main}.

\begin{figure}
\centering
\begin{tikzpicture}[
    >=Stealth,
    line width=0.7mm,
    line cap=round,
    line join=round,
    font=\small
]

\pgfmathsetmacro{\a}{3}
\pgfmathsetmacro{\b}{3}
\pgfmathsetmacro{\m}{tan(30)}
\pgfmathsetmacro{\ha}{\a*\m}
\pgfmathsetmacro{\hb}{\b*\m}

\pgfmathsetmacro{\yT}{1.2}
\pgfmathsetmacro{\yB}{-0.7}
\pgfmathsetmacro{\yM}{0.25}
\pgfmathsetmacro{\xM}{0.3}

\coordinate (T)  at (0,\yT);
\coordinate (TL) at (-\a,{\yT+\ha});
\coordinate (TR) at ( \a,{\yT+\ha});

\coordinate (B0) at (-\xM,\yM);
\coordinate (B1) at ( \xM,\yM);
\coordinate (L)  at ({-\xM-\b},{\yM-\hb});
\coordinate (R)  at ({ \xM+\b},{\yM+\hb});

\coordinate (B3) at (0,\yB);
\coordinate (BL) at (-\a,{\yB-\ha});
\coordinate (BR) at ( \a,{\yB-\ha});

\draw (TL) -- (T) -- (TR);
\draw (L) -- (B0) -- (B1) -- (R);
\draw (BL) -- (B3) -- (BR);

\fill[blue]  (T)  circle (2.5pt);
\fill[blue]  (B3) circle (2.5pt);
\fill[red] (B0) circle (2.5pt);
\fill[red] (B1) circle (2.5pt);

%

\draw[-{Stealth[length=3mm,width=2.2mm]}]
    ($(T)!0.28!(TL)$) -- ($(T)!0.58!(TL)$);
\draw[-{Stealth[length=3mm,width=2.2mm]}]
    ($(T)!0.28!(TR)$) -- ($(T)!0.58!(TR)$);

\draw[-{Stealth[length=3mm,width=2.2mm]}]
    ($(L)!0.22!(B0)$) -- ($(L)!0.56!(B0)$);
\draw[-{Stealth[length=3mm,width=2.2mm]}]
    ($(B1)!0.22!(R)$) -- ($(B1)!0.56!(R)$);

\draw[{Stealth[length=3mm,width=2.2mm]}-]
    ($(B3)!0.28!(BL)$) -- ($(B3)!0.58!(BL)$);
\draw[{Stealth[length=3mm,width=2.2mm]}-]
    ($(B3)!0.28!(BR)$) -- ($(B3)!0.58!(BR)$);

\node at ($(TL)+(-0.48, 0.20)$) {$2$};
\node at ($(TR)+( 0.48, 0.20)$) {$1$};
\node at ($(BL)+(-0.48,-0.20)$) {$3$};
\node at ($(BR)+( 0.48,-0.20)$) {$4$};

\node at ($(L)!0.55!(B0)+(0,-0.30)$) {$5$};
\node at ($(B0)!0.50!(B1)+(0,-0.28)$) {$6$};
\node at ($(B1)!0.55!(R)+(0,-0.30)$) {$7$};

\end{tikzpicture}
\caption{The jump contour $\mathcal X$ for the model problem $\N^{\mathcal X}$. The blue points represent $\beta_2$ and $\beta_3$ from top to bottom, while the red points on the horizontal segment represent $\beta_0$ and $\beta_1$ from left to right.}
\label{ModelFig}
\end{figure}
\subsection{Definition of the model problem}
Below we provide the definition of the  model RH problem in the case $(x,t) \in \cP_+$, where $\cP_+=\cP \cap \{\xi \geq -q_0 \}$. 
Let the jump contour $\X=\cup_{j=1}^7 \X_j$ be defined as follows~(see Figure~\ref{ModelFig}):
\begin{align} \nonumber
&\X_1 = \bigl\{\beta_2+ r \e^{\frac{\ii\pi}{6}}\, \big| \, 0 \leq r < \infty\bigr\}, && \X_2 = \bigl\{ \beta_2 + r \e^{\frac{5\ii\pi}{6}}\, \big| \, 0 \leq r < \infty\bigr\},
	\\ \nonumber
&\X_3 = \bigl\{\beta_3 + r \e^{-\frac{5 \ii \pi}{6}}\, \big| \, 0 \leq r < \infty\bigr\}, && \X_4 = \bigl\{\beta_3 + r \e^{-\frac{\ii \pi}{6}}\, \big| \, 0 \leq r < \infty\bigr\},
	\\\nonumber
&\X_5 = \bigl\{ \beta_0+ r \e^{-\frac{ 5 \ii \pi}{6}}\, \big| \, 0 \leq r < \infty\bigr\}, && \X_7 = \bigl\{\beta_1 + r \e^{\frac{ \ii\pi}{6}}\, \big| \, 0 \leq r < \infty\bigr\},
	\quad
\X_6 = [\beta_0, \beta_1].
\end{align}
Here,  $\beta_j$  is the image of the  critical point $\rz_j$  under the conformal mapping  $z \to \beta$~(see~\eqref{E:betaj}).
Then we define a $3 \times 3$ matrix-valued function $\V^{\X}(y,s,t;\beta)$ on $\X$, which serves as the jump matrix for the model problem.  The parameters $y$ and $s$  will be specified in~\eqref{E:ybeta} and~\eqref{E:deffs}, respectively.

The exact expression of $\V^{\X}$ is given as follows:
\begin{align}
&\V^{\X}_1=\begin{pmatrix}
 1 &0  &0 \\
0  & 1 & 0\\
\ii \e^{2\ii(y \beta+\frac{4}{3} \beta^3)}  &0  &0
\end{pmatrix}, \quad
\V^{\X}_2=\begin{pmatrix}
 1 &0  &0 \\
0  & 1 & 0\\
-\ii \e^{2\ii(y \beta+\frac{4}{3} \beta^3)}  & 0 &0
\end{pmatrix}, \quad
\V^{\X}_3=\begin{pmatrix}
 1 &0  &\ii \e^{-2\ii(y \beta+\frac{4}{3} \beta^3)} \\
0  & 1 & 0\\
0  & 0 &0
\end{pmatrix}, \nonumber \\
&\V^{\X}_4=\begin{pmatrix}
 1 &0  &-\ii \e^{-2\ii(y \beta+\frac{4}{3} \beta^3)} \\
0  & 1 & 0\\
0  & 0 &0
\end{pmatrix},\quad
\V^{\X}_{5 }=\V^{\X}_{6}=\V^{\X}_{7}=\begin{pmatrix}
 1 & \ii s \e^{-\ii y \beta+\ii (\frac{4}{3}q_0)^{\frac{2}{3}}t^{\frac{1}{3}} \beta^2  } &0 \\
  0& 1 &0 \\
  0&  -s \e^{\ii y \beta+\ii (\frac{4}{3}q_0)^{\frac{2}{3}}t^{\frac{1}{3}} \beta^2  }  &1
\end{pmatrix}, \label{E:exVjX}
\end{align}
where $\V^{\X}_j$ denotes the restriction of $\V^{\X}$ to the subcontour $\X_j$.

\begin{RHP}[Model RH problem] \label{RHP:model}
$\N^{\X}(y,s,t;\beta)$ is a $3 \times 3$ matrix-valued function on $\C$ except for $\X$.
\bi
\item $\N^{\X}(y,s,t;\beta)$ satisfies the following jump conditions:
$$
\N^{\X}_+(y,s,t;\beta)=\N^{\X}_-(y,s,y;\beta) \V^{\X}(y,s,t;\beta), \quad \beta \in \X.
$$

\item $\N^{\X}(y,s,t;\beta)$ has the following boundary condition: as  $\beta \to \infty$, $  \N^{\X}(y,s,t;\beta)=\I+\mathcal{O}(\frac{1}{\beta}).$

\item $\N^{\X}(y,s,t;\beta)$ has the following boundary condition: as  $\beta \to 0$, $\N^{\X}(y,s,t;\beta)=\mathcal{O}(1).$
\ei

\end{RHP}
The following theorem further elaborates on the  properties of this model RH problem, which is essential for subsequent calculations.
\begin{theorem}\label{Th:asXi}
There exists a $T>0$ such that for $t>T$ RH problem~\ref{RHP:model} admits a unique solution $\N^{\X}(y,s,t;\beta)$  whenever $(x,t) \in \cP_+$. Moreover, as $\beta \to \infty$, the following asymptotic expansion holds uniformly for $\mathrm{arg} \beta \in [0,2\pi]$:
\be\label{E:asyPhi}
\N^{\X}(y,s,t;\beta)=\I+ \frac{1}{\beta}\left(\M^P_1(y)+ \cE^{\X}_1(y,s,t) \right) +\mathcal{O}(\frac{1}{\beta^2}),
\ee
where $\cE^{\X}_1(y,s,t)$ satisfies
\be \label{E:estcEX1}
\cE^{\X}_1(y,s,t)=\frac{\M^{P}(y,0) \M^{E}_1(s) (\M^P(y,0))^{-1}  }{(\frac{4}{3} q_0)^{1/3} t^{1/6}}   +\mathcal{O}(t^{-1/3}), \quad t \to \infty.
\ee
In the above expression, the functions $\M^P_1(y)$, $\M^P(y,0)$ and $\M^E_1(s)$ are defined by~\eqref{m1Pexpression},~\eqref{mPat0expression}, and~\eqref{M1Eexpression}, respectively.
\end{theorem}

Since the proof of this theorem is quite long, we have devoted a separate subsection to it.
\subsection{Proof of Theorem~\ref{Th:asXi} } \label{S:prof1}
Our proof is based on the analysis of a similar model problem in~\cite[Appendix A]{CL2024main}. The main idea of the proof is to subtract the contribution of the outer model away from the origin, and to subtract the contribution of the local model near the origin, thereby reducing the problem to a small norm RH problem. By analyzing this small norm RH problem, one can  obtain the desired result.  For brevity,  we may omit the dependence on parameters, for instance by letting $\N^{\X}(\beta):=\N^{\X}(y,s,t;\beta)$.

We first note that by performing a simple contour deformation, we can (and will) henceforth assume that $\beta_2 = \ii$ and $\beta_3=-\ii$. On the other hand, when $(x, t) \in \cP_+$ and $t \to \infty$, $\beta_1$ and $\beta_0$ will decay to $0$ at a rate of $\mathcal{O}(t^{-1/3})$. Therefore, in the vicinity  of the origin, it is necessary to introduce a suitable scaling transformation so that the contribution near the origin can be approximated by an exactly solvable model RH problem.  Outside a certain neighborhood of the origin, we show that the remaining contribution can be approximated by the solution of the Painlev\'e II model RH problem. To this end, we define $\cD_{\eps}(0)$ as a small disk centered at the origin with radius $\eps$,  and its boundary is oriented counterclockwise. We choose $\eps$ sufficiently small so that $\pm \ii$  are not in the disk. Inside the disk, we introduce the transformation $\beta \to \alpha$ as
$$
\alpha=(\frac{4}{3} q_0)^{1/3} t^{1/6} \beta.
$$
Our next lemma shows that inside $\cD_{\eps}(0)$, $\N^{\X}(\beta)$ is well approximated by $\M^E(\alpha(\beta))$ as $t \to \infty$,  where $\M^E(\alpha)$  denotes the solution to the error function RH problem~\ref{RHME}. Let $E^{\eps}:=\cup_{j=5}^7 E^{\eps}_j$, where  $E^{\eps}_j:=\cD_{\eps}(0) \cap \X_j$.
\begin{lemma}
The function $\widetilde{\M}^E(\beta):=\M^E(\alpha(\beta))$ satisfies the  jump condition:
\be
\widetilde{\M}^E_+(\beta)=\widetilde{\M}^E_-(\beta) \widetilde{\V}^E(\beta), \quad  \beta \in E^{\eps},
\ee
where
$$
\widetilde{\V}^E(\beta)=
\V^E(\alpha(\beta)), \quad \beta \in E^{\eps}.
$$
For large $t$, the jump matrix $\widetilde{\V}^E(\beta)$ satisfies the estimates
\begin{align}\label{E:estVEX}
\begin{cases}
\|\V^{\X}- \widetilde{\V}^E \|_{L^{\infty} (E^{\eps})} \leq C t^{-1/6},\\
\|\V^{\X}- \widetilde{\V}^E \|_{L^{1} (E^{\eps})} \leq C t^{-1/3}.
\end{cases}
\end{align}
Furthermore, as $t \to \infty$,
\begin{align}
&\| (\widetilde{\M}^E)^{-1}-\I \|_{L^{\infty}(\partial \cD_{\eps}(0) )}=\mathcal{O}(t^{-1/6}),\label{E:tilM1}\\
&(\widetilde{\M}^E)^{-1}(\beta)-\I=-\frac{ \M^E_1 }{(\frac{4q_0}{3})^{1/3} t^{1/6} \beta}+\mathcal{O}(t^{-1/2}),  \quad \beta \in \partial \cD_{\eps}(0),\label{E:tilM2}
\end{align}
where  $\M^E_1$ is given by~\eqref{M1Eexpression}.

\end{lemma}
\begin{proof}
For $\beta \in E^{\eps}$, we have
$$
\V^{\X}- \widetilde{\V}^E=
\bpm
0& \ii s \e^{\ii (\frac{4}{3}q_0)^{\frac{2}{3}}t^{\frac{1}{3}} \beta^2} \left(\e^{-\ii y \beta} -1  \right) &0\\
0&0&0\\
0& -s \e^{\ii (\frac{4}{3}q_0)^{\frac{2}{3}}t^{\frac{1}{3} } \beta^2}\left(\e^{\ii y \beta} -1  \right)&0
\epm.
$$
On \(E^\eps_5\cup E^\eps_7\),
$$
\left|\e^{\ii(\frac43q_0)^{2/3}t^{1/3}\beta^2}\right|
\leq
\e^{-ct^{1/3}r^2},
\qquad r=|\beta-\beta_j|,\quad j=0,1.
$$
Since \(y\) is bounded, we also have
$$
\left|\e^{\pm\ii y\beta}-1\right|
\leq C|\beta|
\leq C(|\beta_j|+r)
\leq C(t^{-1/3}+r),
\qquad j=0,1.
$$
Consequently,
$$
\left|(\V^{\X}-\widetilde{\V}^E)_{12}\right|
+
\left|(\V^{\X}-\widetilde{\V}^E)_{32}\right|
\leq
C(t^{-1/3}+r)\e^{-ct^{1/3}r^2}, \qquad   \beta \in E^\eps_5\cup E^\eps_7.
$$
Taking the supremum in \(r\geq0\), we obtain
$$
\|\V^{\X}-\widetilde{\V}^E\|_{L^\infty(E^\eps_5\cup E^\eps_7)}
\leq Ct^{-1/6}.
$$
Moreover,
$$
\|\V^{\X}-\widetilde{\V}^E\|_{L^1(E^\eps_5\cup E^\eps_7)}
\leq
C\int_0^\infty (t^{-1/3}+r)\e^{-ct^{1/3}r^2}\,\mathrm d r
\leq Ct^{-1/3}.
$$

It remains to consider \(E^\eps_6=[\beta_0,\beta_1]\). On this segment \(\beta\) is real, and hence
$
\left|\e^{\ii(\frac43q_0)^{2/3}t^{1/3}\beta^2}\right|=1
$.
Since \(y\) is bounded,
$$
\left|\e^{\pm\ii y\beta}-1\right|\leq C|\beta| \leq C \left( |\beta_0|+|\beta_1| \right)\leq Ct^{-1/3}, \qquad \beta \in E^\eps_6.
$$
Thus
$$
\|\V^{\X}-\widetilde{\V}^E\|_{L^\infty(E^\eps_6)}
\leq Ct^{-1/3},
$$
and, since \(|E^\eps_6|\leq Ct^{-1/3}\),
$
\|\V^{\X}-\widetilde{\V}^E\|_{L^1(E^\eps_6)}
\leq Ct^{-2/3}
$.
Combining the estimates on \(E^\eps_5\), \(E^\eps_6\), and \(E^\eps_7\), we obtain~\eqref{E:estVEX}.

We next prove \eqref{E:tilM1} and \eqref{E:tilM2}. If \(\beta\in\partial\cD_\eps(0)\), then
\[
|\alpha|=\left(\frac{4q_0}{3}\right)^{1/3}t^{1/6}|\beta|\to\infty, \qquad t \to \infty.
\]
Thus~\eqref{MEasymptotics}  yields
\be \label{E:MEaqm}
\M^E(x,t, \alpha(\beta))=\I+\frac{\M^E_1  }{(\frac{4q_0}{3})^{1/3} t^{1/6} \beta}+\mathcal{O}(t^{-1/2}), \quad t \to \infty,
\ee
uniformly for  $\beta \in \cD_{\eps}(0)$ and $(x,t) \in \cP_+$.
Since $\M^E_1$ is bounded, we immediately obtain~\eqref{E:tilM1} and  \eqref{E:tilM2}.
\end{proof}
Based on the previous analysis, we  proceed to introduce a transformation to obtain a small-norm RH problem and conduct error analysis. Let us define $\cE^{\X}(\beta)$ as
\be \label{E:Ebeta}
\cE^{\X}(\beta)=\begin{cases}
\N^{\X}(\beta) \left( \M^P (\beta) \right)^{-1}, & \beta \in \C \setminus \cD_{\eps}(0),\\
\N^{\X}(\beta) \left(\widetilde{\M}^E (\beta) \right)^{-1} \left( \M^P (\beta) \right)^{-1}, &  \beta \in \cD_{\eps}(0),
\end{cases}
\ee
where $\M^P (\beta) $ is the solution to the Painlev\'e model problem~\ref{RHmP}.
 Define
\[
\hat E=\left(\bigcup_{j=5}^7\X_j\right)\cup\partial\cD_{\eps}(0),
\qquad
E'=\hat E\setminus\overline{\cD_{\eps}(0)}.
\]
 Then it can be directly verified that $\cE^{\X}(\beta)$ has no jump on $\cup_{j=1}^4 \X_j$; it only has a jump on  $\hat{E}$. The jump matrix $\hat{\bv}$ for $\cE^{\X}(\beta)$ is given by the following expression:
\be \label{E:exbv}
\hat{\bv}=
\begin{cases}
\M^P \V^{\X} (\M^P)^{-1}, & \beta \in E',\\
\M^P (\widetilde{\M}^E)^{-1} (\M^P)^{-1}, & \beta \in \partial \cD_{\eps}(0),\\
\M^P \widetilde{\M}^E_- \V^{\X} (\widetilde{\M}^E_+)^{-1}(\M^P)^{-1}, & \beta \in E^{\eps}.
\end{cases}
\ee
\begin{lemma}\label{L:estofw}
Let $\hat{\bw}=\hat{\bv}-\I$. The following estimates hold uniformly for large $t$ and $(x,t) \in \cP_+$:
\begin{align}
&\|\hat{\bw} \|_{(L^1 \cap  L^{\infty})(E')} \leq C \e^{-ct^{1/3}}, \label{E:estofw1}\\
& \|\hat{\bw}  \|_{L^1 \cap L^{\infty}(\partial \cD_{\eps})} \leq C t^{-1/6},\label{E:estofw2} \\
& \|\hat{\bw}  \|_{L^1 (E^{\eps})} \leq C t^{-1/3},\label{E:estofw3}\\
& \|\hat{\bw}  \|_{L^{\infty} (E^{\eps})} \leq C t^{-1/6}. \label{E:estofw4}
\end{align}
\end{lemma}
\begin{proof}
When $z \in E'$, both $\M^P$  and $(\M^P)^{-1}$ are uniformly bounded, and $\V^{\X} -\I$ decays exponentially as $t \to \infty$; hence,~\eqref{E:estofw1} holds.~\eqref{E:estofw2} follows directly from~\eqref{E:tilM1} and~\eqref{E:exbv}. When $z \in E^{\eps}$, we have
$$
\hat{\bw}=\M^P \widetilde{\M}^E_-  \left(  \V^{\X} -\widetilde{\V}^{E}  \right)(\widetilde{\M}^E_+)^{-1}(\M^P)^{-1}.
$$
Combining the above equation with~\eqref{E:estVEX}, we immediately obtain~\eqref{E:estofw3} and~\eqref{E:estofw4}.
\end{proof}

The estimates in Lemma~\ref{L:estofw} show that
\begin{align*}
\|\hat{\bw} \|_{(L^1 \cap  L^{\infty})(\hat{E} )} \leq C t^{-1/6},
\quad (x,t) \in \cP_+.
\end{align*}
Thus by employing  the general inequality $\|f \|_{L^p} \leq \| f\|_{L^1}^{\frac{1}{p}} \|f \|_{L^{\infty}}^{\frac{p-1}{p}}$, we  immediately get
\begin{align}\label{E:Lpw}
\|\hat{\bw} \|_{L^p(\hat{E})} \leq C t^{-1/6}, \quad (x,t) \in \cP_+.
\end{align}
This indicates that $\cE^{\X}(\beta)$ satisfies a small-norm RH problem.

 Let \(\Ca\) be the Cauchy transform on \(\hat E\),
\[
\Ca h(\beta)=\frac{1}{2\pi\ii}\int_{\hat E}\frac{h(\zeta)}{\zeta-\beta}\,d\zeta,
\]
and let \(\Ca_\pm\) denote its boundary values. Define
$
\Ca_{\hat{\bw}}h=\Ca_-(h\hat{\bw})
$.
Since \(\|\hat{\bw}\|_{L^\infty(\hat E)}\to0\), the operator \(I-\Ca_{\hat{\bw}}\) is invertible on \(L^2(\hat E)\) for all sufficiently large \(t\), uniformly for \((x,t)\in\cP_+\). We set
\be\label{E:uhat}
\hat{\bbu}
=
\I+(I-\Ca_{\hat{\bw}})^{-1}\Ca_{\hat{\bw}}\I
\in \I+L^2(\hat E).
\ee
Then the solution of the small-norm RH problem is
\be\label{E:cEXc}
\cE^{\X}(\beta)
=
\I+\frac{1}{2\pi\ii}\int_{\hat E}
\frac{\hat{\bbu}(\zeta)\hat{\bw}(\zeta)}{\zeta-\beta}\,d\zeta,
\qquad \beta\in\C\setminus\hat E.
\ee
Consequently, for each \(N\geq1\),
\be\label{E:gbb1}
\cE^{\X}(\beta)
=
\I+\sum_{j=1}^{N}\frac{\cE^{\X}_j}{\beta^j}
+\mathcal{O}(\beta^{-N-1}),
\qquad \beta\to\infty,
\ee
where
\be\label{E:cEXex}
\cE^{\X}_1(y,s,t)
=
\lim_{\beta\to\infty}^{\angle}\beta(\cE^{\X}(\beta)-\I)
=
-\frac{1}{2\pi\ii}\int_{\hat E}\hat{\bbu}(\zeta)\hat{\bw}(\zeta)\,d\zeta.
\ee
\begin{lemma}\label{L:estofcEX}
As $t \to \infty$,
\begin{align}\label{E:estofCEX}
\cE^{\X}_1(s,y,t)=-\frac{1}{2\pi \ii}\int_{\partial \mathcal{D}_{\eps}(0)} \hat{\bw}(\zeta) \mathrm{d}\zeta+\mathcal{O}(t^{-1/3}).
\end{align}
\end{lemma}

\begin{proof}
The function $\cE^{\X}_1(s,y,t)$ can be rewritten as
$$
\cE^{\X}_1(s,y,t) = -\frac{1}{2\pi \ii}\int_{\partial \mathcal{D}_{\eps}(0)} \hat{\bw}(\zeta) d\zeta + \mathbf{Q}_1(x,t) + \mathbf{Q}_2(x,t),
$$
where
\begin{align*}
\mathbf{Q}_1(x,t) = -\frac{1}{2\pi \ii}\int_{E'} \hat{\bw} (\zeta) \mathrm{d}\zeta, \qquad
 \mathbf{Q}_2(x,t) = -\frac{1}{2\pi \ii}\int_{\hat{E}} (\hat{\bbu}(\zeta)-\I) \hat{\bw}(\zeta) \mathrm{d} \zeta.
\end{align*}
Let's estimate $\|\hat{\bbu}-\I \|_{L^2(\hat{E})}$. A direct calculation yields
\begin{equation}\label{E:xljc}
\begin{aligned}
\|\hat{\bbu} - \I\|_{L^2(\hat{E})}&\leq \|(I-\Ca_{\hat{\bw}})^{-1}\Ca_{\hat{\bw}}\I \|_{L^2(\hat{E})}\leq \sum_{j=0}^{\infty}\| \Ca_{\hat{\bw}}\|^j_{\mathcal{B}(L^2(\hat{E}))}\|\Ca_{\hat{\bw}} \I \|_{L^2(\hat{E})}\\
&\leq  \frac{\|\Ca_- \|_{\mathcal{B}(L^2(\hat{E}))} \| \hat{\bw}\|_{L^2(\hat{E})}}{1-\|\Ca_- \|_{\mathcal{B}(L^2(\hat{E}))} \|\hat{\bw} \|_{L^{\infty}(\hat{E})}} \leq Ct^{-1/6}, \qquad t>T.
\end{aligned}
\end{equation}
Then the lemma follows from Lemma~\ref{L:estofw} and Eq.~\eqref{E:xljc} and straightforward estimates.
\end{proof}
It remains to compute the contribution from \(\partial\cD_{\eps}(0)\). For \(\beta\in\partial\cD_{\eps}(0)\),
\[
\hat{\bw}(\beta)
=
\M^P(\beta)\left[(\widetilde{\M}^E(\beta))^{-1}-\I\right](\M^P(\beta))^{-1}.
\]
Using~\eqref{E:tilM2} and Cauchy's formula, we obtain
\[
-\frac{1}{2\pi\ii}\int_{\partial\cD_{\eps}(0)}
\hat{\bw}(\zeta)\,d\zeta
=
\frac{\M^P(y,0)\M^E_1(s)(\M^P(y,0))^{-1}}
     {(\frac43q_0)^{1/3}t^{1/6}}
+\mathcal{O}(t^{-1/2}), \qquad t \to \infty.
\]
Combining this with~\eqref{E:estofCEX} proves~\eqref{E:estcEX1}.

Finally, by~\eqref{E:Ebeta}, for sufficiently large \(\beta\),
\be\label{E:gbl}
\N^{\X}(\beta)=\cE^{\X}(\beta)\M^P(\beta).
\ee
Therefore, using~\eqref{E:gbb1}, \eqref{mPasymptotics}, and~\eqref{E:gbl}, we obtain the expansion~\eqref{E:asyPhi}. This completes the proof of Theorem~\ref{Th:asXi}.


\begin{figure}
\begin{center}
\begin{tikzpicture}
\node at (0,0) {\includegraphics[width=4.4cm]{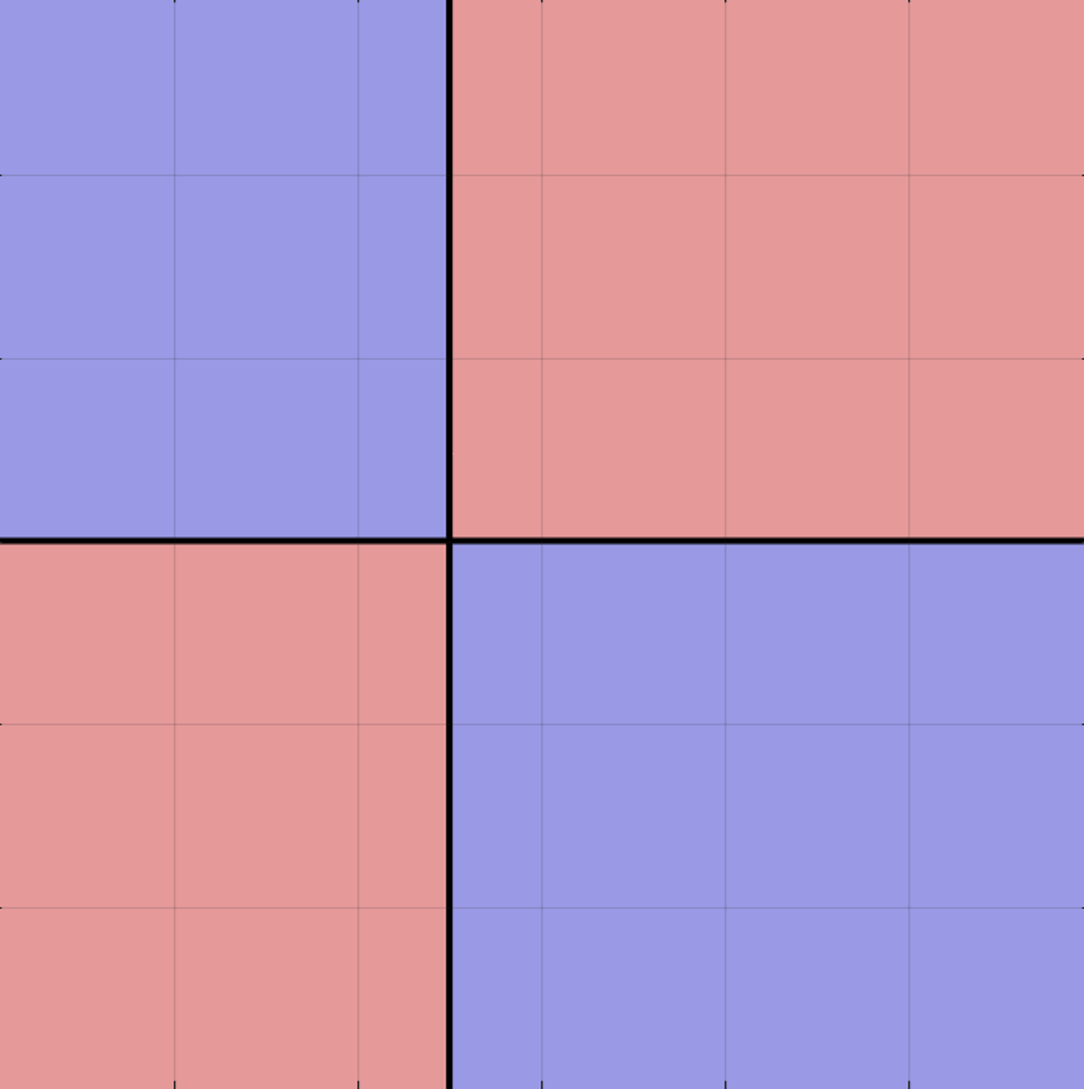} };
\draw[dashed,very thick](0.08,2.25)--(0.08,-2.25);
\node at (0.25,-0.2) {\small $0$};
\filldraw[black](-0.40,0.01)circle(1pt);
\node at (-0.8,-0.2) {\small $-q_0$};
\end{tikzpicture} \hspace{0.1cm} \
\begin{tikzpicture}
\node at (0,0) {\includegraphics[width=4.4cm]{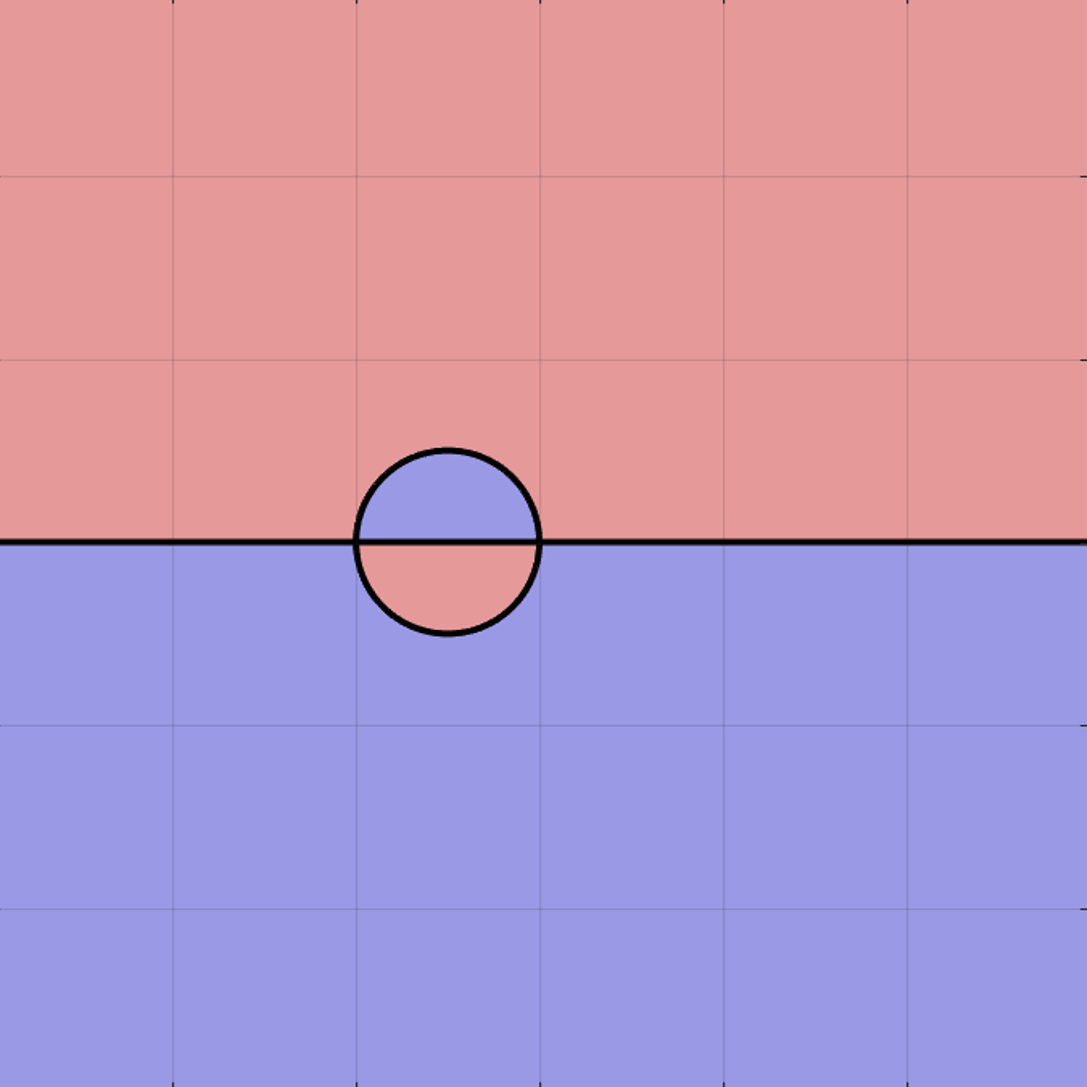}};
\draw[dashed,very thick](0,2.25)--(0,-2.25);
\node at (-1.2,-0.2) {\small $-q_0$};
\node at (0.2,-0.2) {\small $0$};
\filldraw[black](-0.79,-0)circle(1pt);
\end{tikzpicture} \ \ \ \
\begin{tikzpicture}
\node at (0,0) {\includegraphics[width=4.4cm]{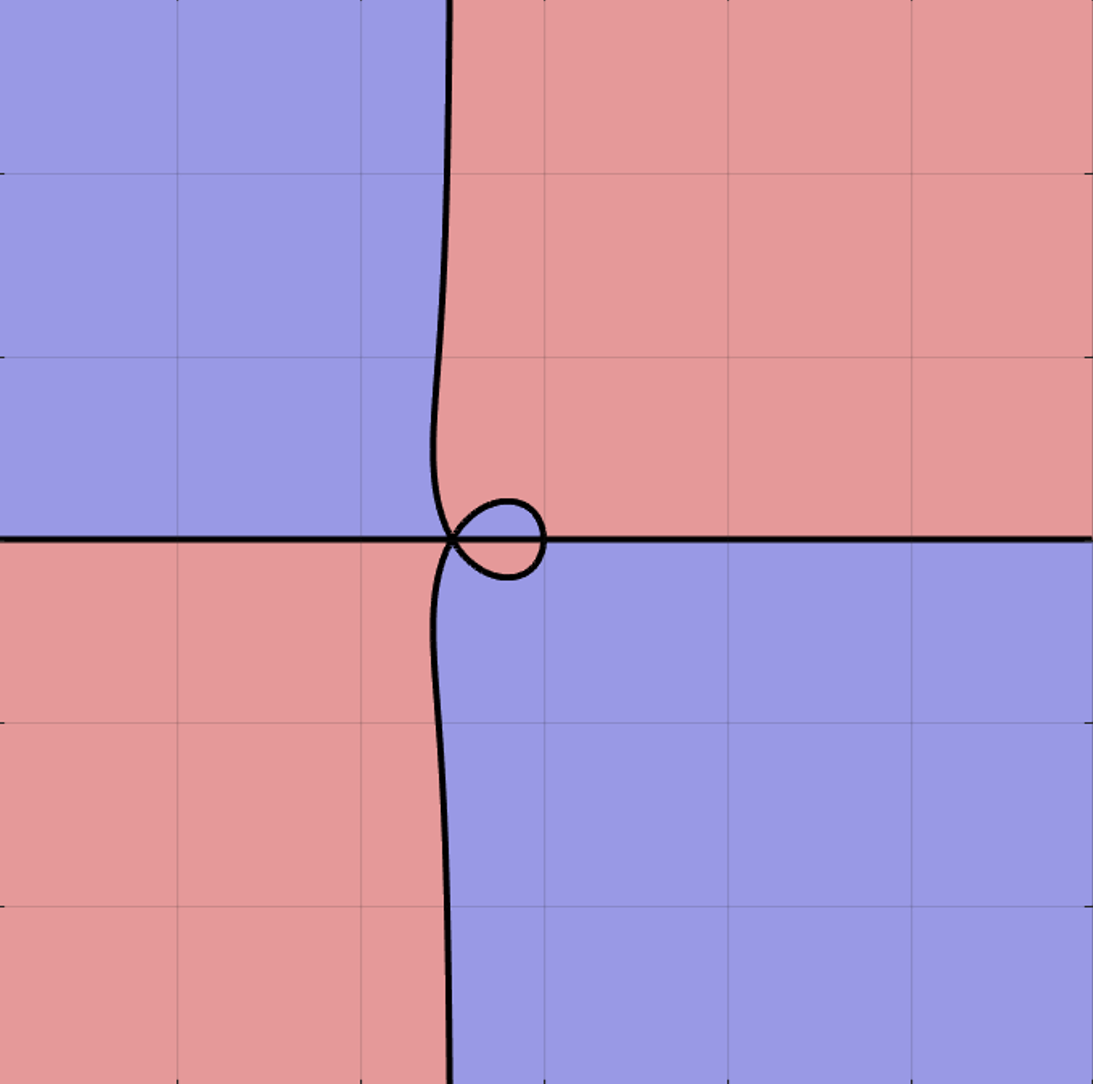}};

\node at (0.2,-0.2) {\small $0$};
\draw[dashed,very thick](0,2.25)--(0,-2.25);
\filldraw[black](-0.4,-0)circle(1pt);
\node at (-0.8,-0.2) {\small $-q_0$};
\end{tikzpicture}
\end{center}
\caption{ From left to right: The signature tables for $\Phi_{21}$, $\Phi_{32}$ and $\Phi_{31}$ for $\xi=-1$ and $q_0=1$. The purple regions correspond to $\{z: \Re \Phi_{ij}<0 \}$ and the pink regions to $\{z: \Re \Phi_{ij}>0 \}$.}
\label{fig:signature table-1}
\end{figure}

\section{Long time asymptotics} \label{S:dza}
In this section, we study the long-time behavior of the RH problem~\ref{RHP:sig} in the Painlev\'e region $\cP$ by using Deift-Zhou nonlinear steepest descent method. Recall that the  region $\cP$ is defined as follows:
\be \label{E:depan}
\cP=\{ (x,t) \in \R \times \R_+ \ | \  |\xi+q_0|  \leq C t^{-2/3} \},
\ee
where $\xi =x/(2t)$ and $C>0$  is a constant.  For the purpose of the proof, it is necessary to split this region into two subregions and address them separately. Denote $\cP = \cP_+ \cup \cP_-$, where
\be \label{E:Ppm}
\cP_+:=\cP \cap \{\xi \geq -q_0 \},  \qquad
\cP_-:=\cP \cap \{\xi\leq -q_0 \}.
\ee
Since the analysis for regions $\cP_+$ and $\cP_-$ is quite similar, we present only the proof for region $\cP_+$.

The long-time asymptotics of RH problem~\ref{RHP:sig} is affected by the growth and decay of the oscillatory terms $\e^{\theta_{21}}$,  $\e^{\theta_{31}}$ and $\e^{\theta_{32}}$ in the jump matrix  $\bV(x,t,z)$. Let $\theta_{ij}(x,t,z)=t \Phi_{ij}(\xi,z)$.  Therefore, we need to analyze the properties of the three phase functions $\Phi_{21}(\xi,z)$, $\Phi_{31}(\xi,z)$ and $\Phi_{32}(\xi,z)$. The expressions for these three phase functions are as follows:
\be
\Phi_{32}(\xi,z)=-2 \ii \xi \frac{q_0^2}{z}+\ii \frac{q_0^4}{z^2}, \quad
\Phi_{21}(\xi,z)=2 \ii \xi z-\ii z^2, \quad
\Phi_{31}(\xi,z)=4 \ii \xi \lambda(z)-4 \ii k(z)\lambda(z).
\ee
 We  need the sign  tables of these  functions~(see Fig~\ref{fig:signature table-1}), which are crucial in the Deift-Zhou analysis.
We also need information about the critical points of these phase functions. A straightforward computation yields the following result:
\bi
\item The function $\Phi_{21}(\xi,z)$ has only one stationary point $\rz_1=\xi$. When $\xi \geq -q_0$ and $|\xi+q_0|  \leq C t^{-2/3} $, we have $|\rz_1+q_0| \leq  C t^{-2/3} $.
\item The function $\Phi_{32}(\xi,z)$ has only one stationary point $\rz_0=\frac{q_0^2}{\xi}$. When $\xi \geq -q_0$ and $|\xi+q_0|  \leq C t^{-2/3} $, we have $|\rz_0+q_0| \leq  C t^{-2/3} $.
\item The function $\Phi_{31}(\xi,z)$ has a stationary point in each of the second and third quadrants, denoted as
\be \label{E:z2z3}
\rz_2=\frac{\eta + \ii \sqrt{4q_0^2-\eta^2} }{2}, \quad \rz_3=\rz_2^*= \frac{\eta - \ii \sqrt{4q_0^2-\eta^2} }{2}, \quad
\eta=\frac{\xi -\sqrt{\xi^2 +8 q_0^2}}{2}.
\ee
When $\xi \geq -q_0$ and $|\xi+q_0|  \leq C t^{-2/3} $, we have
$$
|\rz_2+q_0| \leq  C t^{-1/3}, \quad
 |\rz_3+q_0| \leq  C t^{-1/3}.
$$
\ei
The above analysis shows that these four stationary points will converge to the same point $-q_0$ as $t \to \infty$. This represents a completely different structure compared to the standard Painlev\'e asymptotic analysis of integrable systems, such as those in Refs.~\cite{CLPa2020} and~\cite{WF2023}.

We next  collect a useful property of the phase function \(\Phi_{31}(\xi,z)\). 
The following lemma shows that the local extension of the steepest descent contour from the stationary point \(\rz_2\) in the direction \(\e^{\ii\pi/6}\) can be chosen uniformly for \((x,t)\in\mathcal P_+\). 
The same argument also applies to the extension from \(\rz_2\) in the direction \(\e^{5\ii\pi/6}\). 
Similarly, the corresponding local extensions near \(\rz_3\) can be constructed in the required directions. 
Therefore, the contour deformations near the stationary points \(z_2\) and \(z_3\), as shown in Fig.~\ref{Sigma1Fig}, are always admissible.

\begin{lemma}\label{L:uniform-extension-near-minus-q0}
There exist constants \(\varepsilon>0\) and \(t_0>0\) such that, for all
\(t\geq t_0\) and uniformly for \((x,t)\in\mathcal P_+\), one has
\[
\re \Phi_{31}\bigl(\xi,\rz_2+\tau\e^{\ii\pi/6}\bigr)<0,
\qquad 0<\tau\leq\varepsilon .
\]
\end{lemma}

\begin{proof}
Since \(|\rz_2|=q_0\) and \(\rz_2 \to -q_0\) as $t \to \infty$, we may write
$\rz_2=-q_0\e^{-\ii\theta}$, $\theta>0$,
where \(\theta\to0^+\) as \(t\to+\infty\). The stationary point \(\rz_2\) satisfies
\[
z^4-\xi z^3-\xi q_0^2z+q_0^4=0.
\]
Substituting \(\rz_2=-q_0\e^{-\ii\theta}\) into this equation gives
\be \label{E:jdbsofxi}
\xi
=
q_0\left(-2\cos\theta+\frac{1}{\cos\theta}\right)
=
-q_0+\frac32 q_0\theta^2+\mathcal O(\theta^4),
\qquad \theta\to0^+.
\ee
We now consider the ray
$
z=\rz_2+\tau\e^{\ii\pi/6}$, $ \tau\geq0
$.
We use the  identity
\[
\Phi_{31}(\xi,z)
=
-\ii\frac{z-q_0}{z^2}(z+q_0)^3
+
2\ii(\xi+q_0)\left(z-\frac{q_0^2}{z}\right).
\]
Let $
u=z+q_0$ and $v=\xi+q_0$.
Near \(z=-q_0\), the above identity implies
\be \label{E:mkycyy-1}
\Phi_{31}(\xi,z)
=
\frac{2\ii}{q_0}u^3
+
4\ii v u
+
\mathcal O(u^4+v u^2).
\ee
By~\eqref{E:jdbsofxi}, along the ray under consideration we have
\[
u
=
q_0(1-\e^{-\ii\theta})+\tau\e^{\ii\pi/6}
=
\ii q_0\theta+\tau\e^{\ii\pi/6}+ 
\mathcal O(\theta^2), \qquad
v
=
\frac32 q_0\theta^2+\mathcal O(\theta^4).
\]
Set
$
u_1=\ii q_0\theta+\tau\e^{\ii\pi/6}
$.
Since \(u=u_1+\mathcal O(\theta^2)\) and \(v=\mathcal O(\theta^2)\), it follows that
\be \label{E:mkycyy-2}
v u=\frac32 q_0\theta^2u_1+O((\tau+\theta)^4),
\qquad
u^3=u_1^3+\mathcal O((\tau+\theta)^4),
\qquad
v u^2=\mathcal O((\tau+\theta)^4).
\ee
Therefore, by~\eqref{E:mkycyy-1} and~\eqref{E:mkycyy-2}, we have
\[
\Phi_{31}\bigl(\xi,\rz_2+\tau\e^{\ii\pi/6}\bigr)
=
\frac{2\ii}{q_0}u_1^3
+
6\ii q_0\theta^2u_1
+
\mathcal O((\tau+\theta)^4).
\]
Writing
$
u_1
=
\frac{\sqrt3}{2}\tau
+
\ii\left(q_0\theta+\frac{\tau}{2}\right)
$,
a direct computation gives
\[
\re \Phi_{31}\bigl(\xi,\rz_2+\tau\e^{\ii\pi/6}\bigr)
=
-\frac1{q_0}
\left(\tau+2q_0\theta\right)
\left(2\tau^2-q_0\tau\theta+2q_0^2\theta^2\right)
+
\mathcal O((\tau+\theta)^4).
\]
Since \(q_0>0\) is fixed, there exists a constant \(c_0>0\) such that, for all sufficiently small \(\tau,\theta\geq0\),
\[
\frac1{q_0}
\left(\tau+2q_0\theta\right)
\left(2\tau^2-q_0\tau\theta+2q_0^2\theta^2\right)
\geq
c_0(\tau+\theta)^3.
\]
Consequently,
\[
\re \Phi_{31}\bigl(\xi,\rz_2+\tau\e^{\ii\pi/6}\bigr)
\leq
-c_0(\tau+\theta)^3+\mathcal O((\tau+\theta)^4).
\]
Choosing \(\varepsilon>0\) and \(\theta_0>0\) sufficiently small, we obtain
\[
\re \Phi_{31}\bigl(\xi,\rz_2+\tau\e^{\ii\pi/6}\bigr)<0,
\qquad
0<\tau\leq\varepsilon,\quad 0<\theta<\theta_0.
\]
Finally, since \(\theta\to0^+\) as \(t\to+\infty\) uniformly for \((x,t)\in\mathcal P_+\), there exists \(t_0>0\) such that \(0<\theta<\theta_0\) for all \(t\geq t_0\) and \((x,t)\in\mathcal P_+\). This proves the lemma.
\end{proof}

\subsection{The first transformation}
We first introduce a diagonal conjugation which converts the jump on \(\R_+\) into a factorization compatible with the signature table in Fig.~\ref{fig:signature table-1}.
Define
 \be \label{E:fiesttrsanformation}
\M^{(1)}(x,t,z)=\mathbf{\Delta}_{\infty}^{-1}  \M(x,t,z) \mathbf{\Delta}(z), \qquad z\in \C,
\ee
where
\be
\mathbf{\Delta}(z)=\bpm
\delta_1(z) & & \\
 &\frac{1}{\delta_1(z) \delta_1(\hat{z})} & \\
 & & \delta_1(\hat{z})
\epm,
\ee
and
\begin{align}
\delta_1(z)&=\mathrm{exp}\bigg \{\frac{-1}{2 \pi \ii} \int_{\R_+} \frac{\log \left(1-\frac{1}{\gamma(s)}|r_1(s)|^2-|r_2(s)|^2\right)}{s-z} \mathrm{d}s \bigg\} \label{E:th1},\\
\mathbf{\Delta}_{\infty}&=\diag \left( 1, 1/\delta_1(0) , \delta_1(0)   \right). \label{E:th1infty}
\end{align}
Here and below, the logarithm is taken on the principal branch. Under Assumptions~\ref{As:gesa}--\ref{As:1}, the integrand in~\eqref{E:th1} is locally integrable on \(\R_+\). In particular,
$$
 \delta_1(0)=  \mathrm{exp}\bigg \{-\frac{1}{2 \pi \ii} \int_{\R_+} \frac{\log \left(1-\frac{1}{\gamma(s)}|r_1(s)|^2-|r_2(s)|^2\right) }{s} \mathrm{d}s \bigg\}.
$$

\begin{remark}
Although $\mathbf \Delta(z)$ may be singular near \( q_0\), the singularity
are compensated by the prescribed growth behavior of \(\M(x,t,z)\) at the
branch point. This issue was discussed in detail in the asymptotic analysis of the soliton region; see~\cite[Lemma~3.4]{GZW2025}. This is also the reason why the factor
\(\mathbf \Delta(z)\) is introduced before the gauge transformation in
Lemma~\ref{Th:rel}.  Indeed, after this transformation, the new unknown
\(\N(x,t,z)\) may no longer satisfy the corresponding growth conditions at the branch point.
\end{remark}

We now present the jump matrix of $\M^{(1)}$ on the real axis. By a direct calculation, we have
\be \label{E:jumpM1ex}
\begin{aligned}
\bV^{(1)}=
\begin{cases}
\begin{pmatrix}
1  & \scriptstyle  -\frac{1}{\gamma} r_1^*  \frac{1}{\delta_1^2(z) \delta_1(\hat{z})}  \e^{\theta_{12}} &  \scriptstyle -r_2^*  \frac{\delta_1(z)}{\delta_1(\hat{z})}  \e^{\theta_{13}}\\
 0 & 1 & \scriptstyle  -r_3  \delta_1^2(\hat{z}) \delta_1(z)\e^{\theta_{23}}\\
 0 & 0 &1
\end{pmatrix}
 \begin{pmatrix}
 1 & 0 & 0\\
  \scriptstyle  r_1\delta_1^2(z) \delta_1(\hat{z})  \e^{\theta_{21}}&1  &0 \\
  \scriptstyle  r_2   \frac {\delta_1(\hat{z})}{\delta_1(z)}\e^{\theta_{31}} &  \scriptstyle -\frac{1}{\gamma}r_3^*  \frac{1}{\delta_1^2(\hat{z}) \delta_1(z) }\e^{\theta_{32}} &1
\end{pmatrix}, &z \in \R_-,\\
\begin{pmatrix}
 1 & 0 &0 \\
 -\tilde{r}_1 \e^{\theta_{21}} &1  &0 \\
 \tilde{r}_2 \e^{\theta_{31}} & \frac{1}{\gamma}\tilde{r}_3^* \e^{\theta_{32}} &1
\end{pmatrix}
\begin{pmatrix}
1  &\frac{1}{\gamma}\tilde{r}_1^* \e^{\theta_{12}} & -\tilde{r}_2^* \e^{\theta_{13}}\\
0  & 1 & \tilde{r}_3 \e^{\theta_{23}}\\
  0&0  &1
\end{pmatrix}, & z \in \R_+,
\end{cases}
\end{aligned}
\ee
where
\begin{align*}
&\tilde{r}_1(z)=\rho^2(z)\rho(\hat{z})\hat{r}_1(z), \qquad
\tilde{r}_2(z)=\frac{\rho(z)}{\rho(\hat{z})}\hat{r}_2(z), \qquad
\tilde{r}_3(z)=\rho(z)\rho^2(\hat{z})\hat{r}_3(z),\\
&\hat{r}_1(z)=\frac{b_{21}(z)}{b_{11}(z)}, \qquad   \hat{r}_2(z)=\frac{a_{31}(z)}{a_{33}(z)}, \qquad
\hat{r}_3(z)=\frac{b_{23}(z)}{b_{33}(z)},\\
&\rho(z)=\prod_{j=1}^{N_1}\frac{z-\zeta_j^*}{z-\zeta_j} \prod_{j=1}^{N_2}\frac{z-z_j^*}{z-z_j} \mathrm{exp}\bigg\{\frac{1}{2\pi \ii}  \int_{\R_-} \frac{\log \left(1-\frac{1}{\gamma(s)}|r_1(s)|^2-|r_2(s)|^2\right) }{\zeta-z}\mathrm{d}\zeta \bigg\}.
\end{align*}
The derivation is the same as in~\cite[Section 3.1]{GZW2025}, except that the Blaschke product in \(\rho\) now contains both types of discrete spectrum. This modification does not change the algebraic factorization.

It follows that \(\M^{(1)}\) satisfies the following RH problem.
\begin{RHP}\label{RHP:M1}

Find a $3 \times 3$ matrix-valued function $\M^{(1)}(x,t,z)$ with the following properties:
\bi
\item $\M^{(1)}(x,t,\cdot) : \mathbb{C}\setminus  (  \R \cup  \mathcal{Z})  \to \mathbb{C}^{3 \times 3}$ is analytic. $\M^{(1)}(x,t,z)$ satisfies the jump condition:
\be \label{E:JumpM1}
\M^{(1)}_+(x,t,z)=\M^{(1)}_-(x,t,z) \bV^{(1)}(x,t,z), \quad z \in \R \setminus \{0 \}.
\ee
\item  $\M^{(1)}(x,t,z)$ admits the asymptotic behavior:
\be \label{E:M1astj}
 \M^{(1)}(x,t,z)=\I+\mathcal{O}(\frac{1}{z}), \quad z \to \infty; \quad \M^{(1)}(x,t,z) =\frac{1}{z} \sig_1+ \mathcal{O}(1), \quad z \to 0.
\ee
\item  $\M^{(1)}(x,t,z)$ satisfies the growth conditions near the branch points $\pm q_0$:
\be \label{E:M1gcc}
\begin{cases}
\M_1^{(1)}=\mathcal{O}(z +q_0), & z \in \C_+ \to -q_0,\\
\M_3^{(1)}=\mathcal{O}(z + q_0), & z \in \C_- \to - q_0,
\end{cases} \quad
\begin{cases}
\M_3^{(1)}=\mathcal{O}(z -q_0), & z \in \C_+ \to  q_0,\\
\M_1^{(1)}=\mathcal{O}(z - q_0), & z \in \C_- \to  q_0.
\end{cases}
\ee
\item $\M^{(1)}(x,t,z)$ satisfies the symmetries
\begin{equation}\label{E:M1RHP11}
\begin{aligned}
\M^{(1)}(x,t,z)&=\M^{(1)}(x,t,\hat{z}) \bPi(z), \\
[(\M^{(1)}(x,t,z))^{-1}]^{\top}&=-\frac{1}{\gamma(z)}\bJ \left(\M^{(1)}(x,t,z^*)\right)^{*} \mathbf{\Gamma}(z).
\end{aligned}
\end{equation}
\item
The residue conditions satisfied by $\M^{(1)}$ take the same form as those for $\M$, but with the residue constant $\tau_j$ related to $\zeta_j$ replaced by $\tilde{\tau}_j:=\tau_j |\delta_1(\zeta_j)|^2$, and the residue constant $\kappa_j$ related to $z_j$ replaced by $\tilde{\kappa}_j:=\kappa_j \delta_1^2(z_j) \delta_1(\hat{z}_j)$.
\ei
\end{RHP}
\begin{remark} \label{R:syV1}
Note that $\mathbf{\Delta}(z)$ satisfies the symmetries $\mathbf{\Delta}(z)=\bPi^{-1}(z) \mathbf{\Delta}(\hat{z}) \bPi(z)$ and  $\mathbf{\Delta}^{-1}(z)=\mathbf{\Delta}^*(z^*)$. Then, through a direct calculation, one can verify that the jump matrix $\bV^{(1)}$ satisfies the following symmetries for $z \in \R$:
\begin{align*}
\bV^{(1)}(x,t,z)=\bPi^{-1}(z) \left( \bV^{(1)}(x,t,\hat{z}) \right)^{-1}\bPi(z), \quad
\mathbf{D}(z) \bV^{(1)}(x,t,z)\mathbf{D}^{-1}(z)=\left( \bV^{(1)}(x,t,z) \right)^{\dagger},
\end{align*}
where $\mathbf{D}(z)=-\gamma(z)\mathbf{\Gamma}^{-1}(z)$.
\end{remark}

The RH problem~\ref{RHP:M1} is singular at the origin because of the normalization
\(\M^{(1)}(z)=z^{-1}\sigma_1+\mathcal{O}(1)\). We remove this singularity by introducing the following auxiliary RH problem.
\begin{RHP}\label{RHP:re}
Find a $3 \times 3$ matrix-valued function $\N(x,t,z)$ with the following properties:
\bi
\item $\N(x,t,\cdot) : \mathbb{C}\setminus (\R \cup \mathcal{Z} )  \to \mathbb{C}^{3 \times 3}$ is analytic. Across $\R$, $\N(x,t,z)$ satisfies the jump condition:
\be \label{E:JumpN}
\N_+(x,t,z)=\N_-(x,t,z) \bV^{(1)}(x,t,z), \quad z \in \R.
\ee
\item  $\N(x,t,z)$ admits the asymptotic behavior:
$$  \N(x,t,z)=\I+\mathcal{O}(\frac{1}{z}), \quad z \to \infty; \quad \N(x,t,z) = \mathcal{O}(1), \quad z \to 0. $$
 \item $\N(x,t,z)$ satisfies the same residue conditions as $\M^{(1)}(x,t,z)$ in the RH problem~\ref{RHP:M1}.
\ei
\end{RHP}

The two RH problems are related by a gauge transformation.

\begin{lemma}[Gauge Transformation]~\label{Th:rel}
Suppose that $\M^{(1)}(x,t,z)$ is the unique solution to RH problem~\ref{RHP:M1} and  $\N(x,t,z)$ is  the unique solution to RH problem~\ref{RHP:re}. Then there exists a function $\A(x,t)$ such that
\be \label{E:gfbh}
\M^{(1)}(x,t,z)=\left( \I + \frac{\A(x,t)}{z}\right) \N(x,t,z),
\ee
where $\A(x,t)=\sig_1 \N_+^{-1}(x,t,0)= \sig_1 \N_-^{-1}(x,t,0)$ with
$$
\N_{\pm}(x,t,0):=\lim \limits_{\epsilon \to 0^+} \N(x,t,0 \pm \ii \epsilon).
$$
\end{lemma}

\begin{proof}
See Appendix~\ref{App:AAAA1}.
\end{proof}

\subsection{The second transformation}
In this subsection, we need to introduce a transformation such that the transformed jump matrix, except near the point $-q_0$, decays uniformly as $t \to \infty$. 

We first introduce the  functions
\begin{align*}
\mathbf{\widetilde{\Delta} }(z)=\bpm
\delta(\hat{z})& & \\
 & \frac{1}{ \delta(z) \delta(\hat{z})}& \\
 & & \delta(z)
\epm,\qquad
\mathbf{P}(z)= \begin{pmatrix}
  \rP_1(z)&  & \\
  &  \frac{1}{\rP_1(z)\rP_1(\hat{z}) }& \\
  &  &\rP_1(\hat{z})
\end{pmatrix},
\end{align*}
where
\begin{align}
\delta(z)&=\mathrm{exp} \bigg\{ \frac{1}{2 \pi \ii}  \int_{-\infty}^{-q_0} \frac{\log (1+\frac{1}{\gamma(s)}|r_3(s)|^2)}{s-z} \mathrm{d}s\bigg\}, \label{E:th2}\\
\rP_1(z)&=\prod_{j=1}^{N_1}\frac{z-\zeta_j}{z-\zeta_j^*} \prod_{j=1}^{N_2}\frac{z-z_j}{z-z_j^*}. \label{E:th3}
\end{align}

\begin{lemma}\label{L:dxz}
The function \(\delta(z)\) is well defined for
\(z\in\C\setminus(-\infty,-q_0]\). Moreover, as \(z\to -q_0\) nontangentially
from \(\C\setminus(-\infty,-q_0]\), we have
\be \label{E:ddxw}
\left|\delta(z)-\delta(-q_0)\right|
\le C |z+q_0|\,|\log|z+q_0|| .
\ee
\end{lemma}
\begin{proof}
 The function $\log (1+\frac{1}{\gamma(z)}|r_3(z)|^2)$ is smooth on $(-\infty,-q_0)$ and vanishes at $z=-q_0$. Moreover, as $z \to -\infty$,  $\log (1+\frac{1}{\gamma(z)}|r_3(z)|^2)= \mathcal{O}(1/z^2)$. 
Hence the Cauchy integral in \eqref{E:th2} is well defined. The estimate follows
from the standard endpoint behavior of Cauchy transforms; see, for example,
\cite[Lemma~4.7]{CL2024main} or \cite[Lemma~2.11]{TS2016}.
\end{proof}

Let the contour \(\Sigma^{(1)}\) and the regions \(D_j\), \(1\le j\le 9\), be
as shown in Fig.~\ref{Sigma1Fig}. We set
\begin{align*}
&\T(z)= \mathbf{\widetilde \Delta}(z)\mathbf{P}(z)=: \diag \left(T_1(z), T_2(z),T_3(z)  \right),\\
&\widetilde{\T}(z)=\mathbf{\Delta}(z) \mathbf{\widetilde{\Delta} }(z)\mathbf{P}(z)=: \diag \left(\widetilde T_1(z), \widetilde  T_2(z),\widetilde  T_3(z)  \right).
\end{align*}
For convenience, we also introduce
\[
\begin{aligned}
R_1(z)&=
\frac{r_1(z)}
{1-\frac{1}{\gamma(z)}r_1(z)r_1^*(z^*)},
\qquad
R_3(z)=
\frac{r_3(z)}
{1+\frac{1}{\gamma(z)}r_3(z)r_3^*(z^*)},\\
R_2(z)&=r_2(z)+\frac{1}{\gamma(z)}r_1(z)r_3^*(z^*).
\end{aligned}
\]
We now define the triangular factors \(\G_j(x,t,z)\) for \(z\in D_j\),
\(j=1,\ldots,9\), as follows:
\begin{align}
\G_1(z)
&=
\begin{pmatrix}
1 & 0 & 0 \\
-r_1(z)\delta_1^2(z)\delta_1(\hat z)\e^{\theta_{21}}
& 1
& R_3(z)\delta_1(z)\delta_1^2(\hat z) \e^{\theta_{23}} \\
-R_2(z)\frac{\delta_1(z)}{\delta_1(\hat z)}\e^{\theta_{31}}
& 0 & 1
\end{pmatrix},
\notag\\[2mm]
\G_5(z)
&=
\begin{pmatrix}
1
&
-\frac{1}{\gamma(z)}
\bigl[
r_1^*(z^*)-R_2^*(z^*)R_3^*(z^*)
\bigr]
\frac{\e^{\theta_{12}}}{\delta_1^2(z)\delta_1(\hat z)}
&
-R_2^*(z^*)
\frac{\delta_1(\hat z)}{\delta_1(z)}
\e^{\theta_{13}}
\\
0 & 1 & 0 \\
0
&
-\frac{1}{\gamma(z)}R_3^*(z^*)
\frac{\e^{\theta_{32}}}{\delta_1(z)\delta_1^2(\hat z)}
& 1
\end{pmatrix},
\notag\\[2mm]
\G_3(z)
&=
\begin{pmatrix}
1
&
\frac{1}{\gamma(z)}R_1^*(z^*)
\frac{\e^{\theta_{12}}}{\delta_1^2(z)\delta_1(\hat z)}
& 0
\\
0 & 1 & 0 \\
-r_2(z)\frac{\delta_1(z)}{\delta_1(\hat z)}\e^{\theta_{31}}
&
\frac{1}{\gamma(z)}
\bigl[
r_3^*(z^*)-R_1^*(z^*)r_2(z)
\bigr]
\frac{\e^{\theta_{32}}}{\delta_1(z)\delta_1^2(\hat z)}
& 1
\end{pmatrix},
\notag\\[2mm]
\G_7(z)
&=
\begin{pmatrix}
1 & 0
& -r_2^*(z^*)\frac{\delta_1(\hat z)}{\delta_1(z)}
\e^{\theta_{13}}
\\
R_1(z)\delta_1^2(z)\delta_1(\hat z)\e^{\theta_{21}}
& 1
& -r_3(z)\delta_1^2(\hat z)\delta_1(z)\e^{\theta_{23}}
\\
0 & 0 & 1
\end{pmatrix},
\notag\\[2mm]
\G_2(z)
&=
\begin{pmatrix}
1 & 0 & 0 \\
0 & 1 & 0 \\
-r_2(z)\frac{\delta_1(z)}{\delta_1(\hat z)}\e^{\theta_{31}}
& 0 & 1
\end{pmatrix},
\quad
\G_6(z)
=
\begin{pmatrix}
1 & 0
& -r_2^*(z^*)\frac{\delta_1(\hat z)}{\delta_1(z)}
\e^{\theta_{13}}
\\
0 & 1 & 0 \\
0 & 0 & 1
\end{pmatrix},
\notag\\[2mm]
\G_8(z)
&=
\begin{pmatrix}
1 & 0 & 0 \\
R_1(z)\delta_1^2(z)\delta_1(\hat z) \e^{\theta_{21}}
& 1 & 0 \\
0 & 0 & 1
\end{pmatrix}, \quad 
\G_9(z)
=
\begin{pmatrix}
1 & 0 & 0 \\
-\tilde r_1(z) \e^{\theta_{21}} & 1 & 0 \\
\tilde r_2(z) \e^{\theta_{31}}
& \frac{1}{\gamma(z)}\tilde r_3^*(z^*) \e^{\theta_{32}}
& 1
\end{pmatrix},
\notag\\[2mm]
\G_4(z)
&=
\begin{pmatrix}
1
& -\frac{1}{\gamma(z)}\tilde r_1^*(z^*)e^{\theta_{12}}
&
\left[
\tilde r_2^*(z^*)
+\frac{1}{\gamma(z)}
\tilde r_1^*(z^*)\tilde r_3(z)
\right]\e^{\theta_{13}}
\\
0 & 1 & -\tilde r_3(z) \e^{\theta_{23}} \\
0 & 0 & 1
\end{pmatrix}.
\end{align}

\begin{figure}
\centering
\begin{tikzpicture}[
    scale=0.75,
    transform shape,
    line cap=round,
    line join=round,
    >=latex,
    contour/.style={line width=1.6pt, draw=black},
    arrowcontour/.style={
        contour,
        postaction={decorate},
        decoration={markings,mark=at position 0.56 with {\arrow{>}}}
    },
    axisarrow/.style={->, >=latex, line width=0.9pt, draw=black},
    every node/.style={font=\small, inner sep=1pt}
]

\def\H{2.60}           
\def\hV{0.82}          
\def\xC{-1.40}         
\def\hInner{1.05}      
\def\sep{1.30}         

\pgfmathsetmacro{\mThirty}{tan(30)}

\pgfmathsetmacro{\dxOuter}{(\H-\hV)/\mThirty}
\pgfmathsetmacro{\xA}{\xC-\dxOuter}
\pgfmathsetmacro{\xT}{\xC+\dxOuter}

\pgfmathsetmacro{\xLeftZero}{\xC+\hV/\mThirty}
\pgfmathsetmacro{\xRightZero}{\xC-\hV/\mThirty}

\pgfmathsetmacro{\xP}{\xLeftZero-2*\sep}
\pgfmathsetmacro{\xQ}{\xRightZero+2*\sep}

\pgfmathsetmacro{\xAone}{\xP-\hInner/\mThirty}
\pgfmathsetmacro{\xRone}{\xQ+\hInner/\mThirty}

\pgfmathsetmacro{\Rright}{\H}
\pgfmathsetmacro{\xZ}{\xT+\Rright}
\pgfmathsetmacro{\xE}{\xZ+\Rright}
\pgfmathsetmacro{\xRray}{\xE+1.45}

\pgfmathsetmacro{\yLabelOuter}{0.57*\H+0.43*\hV}
\pgfmathsetmacro{\yLabelInner}{0.58*\hInner+0.20}

\coordinate (L1) at (-7.20,  \H);
\coordinate (A)  at (\xA,    \H);
\coordinate (A1) at (\xAone, \hInner);

\coordinate (L2) at (-7.20, -\H);
\coordinate (B)  at (\xA,   -\H);
\coordinate (B1) at (\xAone,-\hInner);

\coordinate (P)  at (\xP, 0);
\coordinate (U)  at (\xC, \hV);
\coordinate (M)  at (\xC, 0);
\coordinate (Q)  at (\xQ, 0);
\coordinate (D)  at (\xC,-\hV);

\coordinate (R1) at (\xRone, \hInner);
\coordinate (T)  at (\xT,    \H);

\coordinate (R2) at (\xRone,-\hInner);
\coordinate (S)  at (\xT,   -\H);

\coordinate (Z0) at (\xZ, 0);
\coordinate (E)  at (\xE,  \H);
\coordinate (F)  at (\xE, -\H);
\coordinate (R3) at (\xRray,  \H);
\coordinate (R4) at (\xRray, -\H);

\draw[dashed, line width=1.5pt] (-7.20,0) -- (\xRray,0);

\draw[axisarrow] (-4.90,0) -- (-4.45,0);
\draw[axisarrow] ({\xT+0.28},0) -- ({\xT+0.70},0);
\draw[axisarrow] ({\xRray-0.48},0) -- ({\xRray-0.06},0);

\draw[arrowcontour] (L1) -- (A);

\draw[arrowcontour] (U) -- (A);

\draw[arrowcontour] (P) -- (A1) -- (A);

\draw[arrowcontour] (L2) -- (B);
\draw[arrowcontour] (B) -- (D);
\draw[arrowcontour] (B) -- (B1) -- (P);

\draw[arrowcontour] (U) -- (T);
\draw[arrowcontour] (Q) -- (R1) -- (T);

\draw[arrowcontour] (S) -- (D);

\draw[arrowcontour] (S) -- (R2) -- (Q);

\draw[arrowcontour] (P) -- (Q);

\draw[arrowcontour] (T)  arc[start angle=180,end angle=270,radius=\Rright];
\draw[arrowcontour] (Z0) arc[start angle=270,end angle=360,radius=\Rright];

\draw[arrowcontour] (S)  arc[start angle=180,end angle=90,radius=\Rright];
\draw[arrowcontour] (Z0) arc[start angle=90,end angle=0,radius=\Rright];

\draw[arrowcontour] (S) -- (F);

\draw[arrowcontour] (E) -- (R3);
\draw[arrowcontour] (F) -- (R4);

\fill (P)  circle (2.2pt);
\fill (U)  circle (2.2pt);
\fill (M)  circle (2.2pt);
\fill (Q)  circle (2.2pt);
\fill (D)  circle (2.2pt);
\fill (Z0) circle (2.4pt);

\node at ({\xC+0.47*\dxOuter},  \yLabelOuter+0.2) {$1$};
\node at ({\xC-0.47*\dxOuter},  \yLabelOuter+0.2) {$2$};
\node at ({\xC-0.47*\dxOuter}, -\yLabelOuter-0.2) {$3$};
\node at ({\xC+0.47*\dxOuter}, -\yLabelOuter-0.2) {$4$};

\node at ({0.54*\xQ+0.46*\xRone},  \yLabelInner) {$5$};

\node at ({0.54*\xP+0.46*\xAone},  \yLabelInner) {$6$};

\node at ({0.54*\xP+0.46*\xAone}, -\yLabelInner) {$7$};

\node at ({0.54*\xQ+0.46*\xRone}, -\yLabelInner) {$8$};

\node at (\xC, 0.28) {$9$};

\node at ({\xZ+0.58*\Rright},  1.10) {$10$};
\node at ({\xT+1.02},          1.10) {$11$};
\node at ({\xT+1.02},         -1.10) {$12$};
\node at ({\xZ+0.58*\Rright}, -1.10) {$13$};

\node[font=\large] at (-6.25,  1.25) {$D_1$};
\node[font=\large] at ({0.45*\xA+0.52*\xC},  1.08) {$D_2$};
\node[font=\large] at ({0.78*\xC+0.52*\xT},  1.00) {$D_2$};
\node[font=\large] at ({\xT+0.26},  0.62) {$D_3$};
\node[font=\large] at ({\xRray-0.62},  1.18) {$D_4$};

\node[font=\large] at (-6.25, -1.25) {$D_5$};
\node[font=\large] at ({0.45*\xA+0.52*\xC}, -1.08) {$D_6$};
\node[font=\large] at ({0.78*\xC+0.52*\xT}, -1.00) {$D_6$};
\node[font=\large] at ({\xT+0.26}, -0.62) {$D_7$};
\node[font=\large] at ({\xZ+0.05}, -1.55) {$D_8$};
\node[font=\large] at ({\xRray-0.62}, -1.18) {$D_9$};
\end{tikzpicture}
\caption{{\protect\small
    The  contour $\Sigma^{(1)}$  and the regions $D_j$, $j=1,...9$. }}
	\label{Sigma1Fig}
\end{figure}

The second transformation is defined by
\[
\N^{(1)}(x,t,z)
=
\T_{\infty}^{-1}\N(x,t,z)\G(x,t,z)\T(z),
\]
where
\[
\G(x,t,z)=
\begin{cases}
\G_j(x,t,z), & z\in D_j,\quad j=1,\ldots,9,\\
\I, & \text{elsewhere,}
\end{cases}
\]
and
\be \label{E:Tsharpinf}
\T_{\infty}
:=
\lim_{z\to\infty}\mathbf{\widetilde  \Delta}(z)\mathbf P(z)
=
\diag\left(
 \delta (0),
\frac{1}{\delta(0)\rP_1(0)},
\rP_1(0)
\right).
\ee

The jump contour for \(\N^{(1)}\) is denoted by \(\Sigma^{(1)}\) and is shown in
Fig.~\ref{Sigma1Fig}. We denote the corresponding jump matrix by
\(\V^{(1)}\). On the contours
\(\Sigma^{(1)}_j\), we have the following expressions:
\begin{align*}
&\V^{(1)}_1(z)=\e^{\bTheta}
\begin{pmatrix}
 1 & 0 &  0\\
  0&1  &0 \\
 r_2(z) \frac{\widetilde T_1}{\widetilde  T_3}(z)& 0 &1
\end{pmatrix}
\e^{-\bTheta}, \quad
\V^{(1)}_2(z)=\e^{\bTheta}
\begin{pmatrix}
 1 & 0 &  0\\
  0&1  &0 \\
 -r_2(z) \frac{\widetilde T_1}{ \widetilde T_3}(z)& 0 &1
\end{pmatrix}
\e^{-\bTheta},\\
&\V^{(1)}_3(z)=\e^{\bTheta}
\begin{pmatrix}
 1 & 0 &  -r_2^{*}(z^*) \frac{\widetilde T_3}{\widetilde T_1}(z)\\
  0&1  &0 \\
 0& 0 &1
\end{pmatrix}
\e^{-\bTheta}, \quad
\V^{(1)}_4(z)=\e^{\bTheta}
\begin{pmatrix}
 1 & 0 &  r_2^{*}(z^*) \frac{\widetilde T_3}{\widetilde  T_1}(z)\\
  0&1  &0 \\
 0& 0 &1
\end{pmatrix}
\e^{-\bTheta},\\
&\V^{(1)}_5(z)=
\bpm
1& -\frac{R_1^{*}(z^*)}{\gamma(z)}\frac{\widetilde T_2}{\widetilde T_1}(z)\e^{\theta_{12}} &0\\
0&1&0\\
0&-\frac{r_3^*(z^*)}{\gamma(z)} \frac{\widetilde T_2}{\widetilde T_3}(z) \e^{\theta_{32}}&1
\epm, \quad
\V^{(1)}_8(z)=
\bpm
1&0&0\\
-R_{1}(z) \frac{\widetilde T_1}{\widetilde  T_2}(z)\e^{\theta_{21}} &1&r_3(z)\frac{\widetilde T_3}{\widetilde T_2}(z) \e^{\theta_{23}}\\
0&0&1
\epm,
\\
&\V^{(1)}_6(z)=\e^{\bTheta}
\bpm
1& 0 &0\\
-r_1(z) \frac{\widetilde T_1}{\widetilde T_2}(z)&1&R_3(z)\frac{\widetilde T_3}{ \widetilde T_2}(z)\\
 (r_2(z)-R_2(z)) \frac{\widetilde T_1}{\widetilde  T_3}(z) &0&1
\epm\e^{-\bTheta},\\
&\V^{(1)}_7(z)=\e^{\bTheta} \widetilde \T^{-1}
\bpm
1& -\frac{(r_1(z^*))^*}{\gamma(z)}+\frac{r_1^*(z^*) r_3(z) R_3^{*}(z^*)}{\gamma^2(z)} &(r_2(z^*)-R_2(z^*))^* \\
0&1&0\\
0&-\frac{1}{\gamma(z)} R_3^{*}(z^*)&1
\epm\e^{-\bTheta} \widetilde \T,\\
&\V^{(1)}_9(z)=\e^{\bTheta} \widetilde \T_-^{-1}
\begin{pmatrix}
1-\frac{1}{\gamma(z)}r_1(z) r^*_1(z^*)  & -\frac{r_1^{*}(z^*)}{\gamma(z)}  &0 \\
 r_1(z)& 1+\frac{1}{\gamma(z)}r_3(z) r^*_3(z^*) & -r_3(z)\\
 0 & -\frac{r_3^{*}(z^*)}{\gamma(z)}  &1
\end{pmatrix}
\e^{-\bTheta} \widetilde \T_+,\\
&\V^{(1)}_{13}(z)=\e^{\bTheta}
\bpm
1&0&0\\
-\left(R_1(z)+\tilde{r}_1(z) \frac{1}{\delta_1^2(z) \delta_1(\hat{z})} \right)\frac{\widetilde T_1}{\widetilde T_2}(z) & 1 &0\\
\tilde{r}_2(z) \frac{T_1}{T_3}(z) & \frac{\tilde{r}_3^{*}(z^*)}{\gamma(z) }  \frac{T_2}{T_3}(z)&1
\epm \e^{-\bTheta},\\
&\V^{(1)}_{12}(z)=\e^{\bTheta}
\bpm
1&0&-r^{*}_2(z^*) \frac{\widetilde T_3}{\widetilde T_1}(z)\\
0 & 1 & \left[-r_3(z) +r^{*}_2(z^*) R_1(z)    \right]\frac{\widetilde T_3}{\widetilde T_2}(z)\\
0&0&1
\epm \e^{-\bTheta},\\
&\V^{(1)}_{10}(z)=\T^{-1}\G_4^{-1} \T, \quad
 \V^{(1)}_{11}(z)=\T^{-1}\G_3 \T.
\end{align*}
We omit the expressions of the jump matrices on the remaining unlabeled contours for brevity.

 The following lemma provides an estimate for the jump matrix near the origin.

\begin{lemma}\label{L:estoutq0}
For $z \in \Sigma^{(1)}_{10,11,12,13 }$, the estimate
\be \label{E:estinear0}
\left|\left( \V^{(1)}(x,t,z)  -\I \right)_{ij}  \right| \leq C |z| \e^{ t \re \Phi_{ij}(\xi,z)}
\ee
holds.
Moreover, if $(i,j) \ne (2,1)$, then
\be \label{E:estinear011}
\left|\left( \V^{(1)}(x,t,z)  -\I \right)_{ij}  \right| \leq C |z|^2 \e^{t \re \Phi_{ij}(\xi,z)}.
\ee

\end{lemma}

\begin{proof}
We prove the estimates only on \(\Sigma^{(1)}_{13}\); the arguments on
\(\Sigma^{(1)}_{10}\), \(\Sigma^{(1)}_{11}\), and \(\Sigma^{(1)}_{12}\) are the
same. A direct calculation gives
$$
\V^{(1)}_{13}(z)-\I=
\bpm
0&0&0\\
-\left(R_1(z)+\tilde{r}_1(z) \frac{1}{\delta_1^2(z) \delta_1(\hat{z})} \right)\frac{\widetilde T_1}{\widetilde T_2}(z)  \e^{\theta_{21}} & 0 &0\\
\tilde{r}_2(z)  \frac{T_1}{T_3}(z)  \e^{\theta_{31}} & \frac{\tilde{r}_3^{*}(z^*)}{\gamma(z)} \frac{T_2}{T_3}(z)  \e^{\theta_{32}}&0
\epm .
$$
By the local behavior of the scattering coefficients near \(z=0\) and their
symmetries, we have, as \(S_\varepsilon\ni z\to0\),
\be \label{E:xybakns}
r_2(z)=\mathcal O(z^2),\qquad r_3(z)=\mathcal O(z^2),\qquad
\tilde r_2(z)=\mathcal O(z^2),\qquad \tilde r_3(z)=\mathcal O(z^2).
\ee
Together with \(\gamma^{-1}(z)=\mathcal O(z^2)\) and the local behavior of the
scalar factors \(\widetilde T_j\) and \(\delta_1\), these estimates immediately give
\eqref{E:estinear0} and~\eqref{E:estinear011} for all nonzero entries of
\(\V^{(1)}_{13}-\I\), except possibly the \((2,1)\)-entry. It remains to estimate
this entry.

Set
$$
\mathcal{Q}(z)=R_1(z)+\tilde{r}_1(z) \frac{1}{\delta_1^2(z) \delta_1(\hat{z})} = R_1(z)+\hat{r}_1(z) \frac{\rho^2(z) \rho(\hat{z})}{\delta_1^2(z) \delta_1(\hat{z})}.
$$
Using the so-called trace formula~(see~\cite[section 3.3]{BD2015-1}), it is straightforward to verify that $\delta_1(z)=s_1(z) \rho(z)$, where
$$
s_1(z)=\begin{cases}
a_{11}(z), & z \in \C_+,\\
\frac{1}{b_{11}(z)}, & z \in \C_-.
\end{cases}
$$
Based on this relationship, we obtain
$
\mathcal{Q}(z)= R_1(z)+\hat{r}_1(z)b_{11}^2(z) b_{33}(z)
$.
A straightforward calculation gives us
\be \label{E:wmbb1}
|\mathcal{Q}(z)| \leq |R_1(z) -r_1(z)|+\left|
r_1(z)+\hat{r}_1(z) b_{11}^2(z) b_{33}(z)
\right|.
\ee
It's easy to verify that
\be \label{E:xyes1}
\begin{aligned}
\hat{r}_1(z)b_{11}^2(z) b_{33}(z)&=\frac{b_{12}(z)}{b_{11}(z)}b_{11}^2(z)b_{33}(z)\\
&=-r_1(z) a_{33}(z) a_{11}(z)b_{11}(z)b_{33}(z)+a_{23}(z)a_{31}(z)b_{11}(z)b_{33}(z).
\end{aligned}
\ee
Applying the above equation along with the estimates
\be \label{E:xygj2}
|a_{33}(z) a_{11}(z) b_{11}(z)b_{33}(z)-1| \leq C |z|,  \quad |a_{23}(z)a_{31}(z)b_{11}(z) b_{33}(z)| \leq  C |z|,
\ee
 we obtain
\be  \label{E:wmbb2}
\begin{aligned}
\left|
r_1(z)+\hat{r}_1(z) b_{11}^2(z) b_{33}(z)
\right|   &\leq C |a_{33}(z) a_{11}(z) b_{11}(z)b_{33}(z)-1|  +    |a_{23}(z)a_{31}(z)b_{11}(z) b_{33}(z)| \\
&\leq C |z|.
\end{aligned}
\ee
Combining~\eqref{E:wmbb1} with~\eqref{E:wmbb2}, we conclude that $\mathcal{Q}(z)$ satisfies
$
|\mathcal{Q}(z)| 
  \leq C |z| 
$.
So we immediately obtain
$$
\left|\left( \V^{(1)}_{13}(x,t,z)  -\I \right)_{21}  \right| \leq C |z| \e^{ t\re \Phi_{21}}.
$$
This proves the desired estimates on \(\Sigma^{(1)}_{13}\).
\end{proof}

Let \(\eps>0\) be sufficiently small, and let \(\cD_{\eps}\) denote the open disk
centered at \(-q_0\) with radius \(\eps\). Then we have the following lemma.
\begin{lemma}\label{L:yzx}
The jump matrix \(\V^{(1)}\) converges to the identity matrix as
\(t\to+\infty\), uniformly for \((x,t)\in\cP_+\) and
\(z\in\Sigma^{(1)}\setminus\cD_{\eps}\). More precisely,
\be \label{E:gjV1}
\begin{aligned}
&\|\V^{(1)}-\I \|_{L^1 (\Sigma^{(1)} \setminus \mathcal{D}_{\eps})}
\leq C t^{-1},\\
&\|\V^{(1)}-\I \|_{ L^{\infty}(\Sigma^{(1)} \setminus \mathcal{D}_{\eps})}
\leq C t^{-1/2}.
\end{aligned}
\ee
\end{lemma}

\begin{proof}
On
\[
\Sigma^{(1)}\setminus
\left(\cD_{\eps}\cup\Sigma^{(1)}_{10,11,12,13}\right),
\]
the jump matrix \(\V^{(1)}-\I\) is exponentially small as \(t\to+\infty\),
uniformly for \((x,t)\in\cP_+\). It remains to estimate the jumps on
\(\Sigma^{(1)}_{10,11,12,13}\). By direct inspection of the explicit formulas for the jump matrices, together
with the uniform lower bounds for \(|\re\Phi_{23}|\) and \(|\re\Phi_{13}|\) on
these four contours, all entries of \(\V^{(1)}-\I\), except possibly the
\((2,1)\)- and \((1,2)\)-entries, are exponentially small uniformly for
\((x,t)\in\cP_+\). It remains only to estimate these two entries.

We first estimate the \((2,1)\)-entry of \(\V^{(1)}_{13}-\I\). On \(\Sigma^{(1)}_{13}\), we have
$
|z|^2=\frac{q_0^2}{\varepsilon}|\im z|
$.
Combining this relation with Lemma~\ref{L:estoutq0} and the sign distribution
of the phase, we obtain
\[
\left|
\left(\V^{(1)}_{13}(x,t,z)-\I\right)_{21}
\right|
\le
C|z|\,e^{t\re\Phi_{21}(\xi,z)}
\le
C|z|e^{-ct|\im z|}
\le
C|z|e^{-ct|z|^2}, \quad z \in  \Sigma^{(1)}_{13}, \quad (x,t) \in \cP_+.
\]
Hence
\[
\left\|
\left(\V^{(1)}_{13}-\I\right)_{21}
\right\|_{L^\infty(\Sigma^{(1)}_{13})}
\le
C\sup_{u\ge0}u e^{-ctu^2}
\le Ct^{-1/2},
\]
and
\[
\left\|
\left(\V^{(1)}_{13}-\I\right)_{21}
\right\|_{L^1(\Sigma^{(1)}_{13})}
\le
C\int_0^\infty u e^{-ctu^2}\, \mathrm{d} u
\le Ct^{-1}.
\]
The remaining cases  can be
treated in the same way. This proves~\eqref{E:gjV1}.
\end{proof}

\begin{remark}\label{R:zsxy}
For later use, we record a simple consequence of the preceding estimates.  For
\((i,j)\ne(2,1)\), Lemma~\ref{L:estoutq0} gives
\[
\left|
\frac{1}{z}\left(\V^{(1)}(x,t,z)-\I\right)_{ij}
\right|
\le C|z|\,\e^{t\re\Phi_{ij}(\xi,z)},
\qquad
z\in\Sigma^{(1)}_{10,11,12,13} .
\]
Thus, by the same argument as in the proof of Lemma~\ref{L:yzx}, we have
\be \label{E:zsxy}
\left\|
\frac{1}{z}\left(\V^{(1)}(x,t,\cdot)-\I\right)_{ij}
\right\|_{L^1(\Sigma^{(1)}\setminus\cD_{\eps})}
\le Ct^{-1},
\qquad (i,j)\ne(2,1).
\ee
\end{remark}

We note that this transformation affects the residue conditions; below, we present the residue conditions satisfied by  $\N^{(1)}(x,t,z)$.

\begin{lemma}\label{L:mlstj}
The transformed matrix \(\N^{(1)}(x,t,z)\) satisfies the following residue
conditions.
\begin{itemize}
\item[{\rm (i)}]

For $1 \leq j \leq N_1$, we have
\begin{equation}\label{E:lstjN1+}
\mathrm{Res}_{z=\zeta_j}\N^{(1)}(x,t,z)=\lim_{z \to \zeta_j}\N^{(1)}(x,t,z) \bpm 0&0&\widetilde C_j\\ 0&0&0 \\ 0&0&0 \epm,
\end{equation}
where
$$
  \widetilde C_j=\frac{\tau_j^{-1}\e^{\theta_{13}(x,t,\zeta_j)}}{[\widetilde T_3^{-1}]'(\zeta_j)\widetilde T_1'(\zeta_j)}.
$$

\item[{\rm (ii)}]
 For $1\leq j \leq N_2$, we have
\begin{equation}\label{E:lstjN2+}
\begin{aligned}
\mathrm{Res}_{z=z_j}\N^{(1)}(x,t,z)=\lim_{z \to z_j}\N^{(1)}(x,t,z) \bpm 0&\widetilde D_j&0\\ 0&0&0 \\ 0&0&0 \epm, \\
 \mathrm{Res}_{z=z_j^*}\N^{(1)}(x,t,z)=\lim_{z \to z_j^*}\N^{(1)}(x,t,z) \bpm 0&0&0\\ \widehat{D}_j&0&0 \\ 0&0&0 \epm ,
\end{aligned}
\end{equation}
where
$$
 \widetilde D_j=\frac{\kappa_j^{-1}\e^{\theta_{12}(x,t,z_j)}}{[ \widetilde T_2^{-1}]'(z_j) \widetilde T_1'(z_j)}, \quad
\widehat{D}_j=\frac{\gamma(z_j^*)(\kappa^{*}_j)^{-1}\e^{\theta_{21}(x,t,z_j^*)}}{[\widetilde T_1^{-1}]'(z_j^*) \widetilde T_2'(z_j^*)}.
$$
\end{itemize}
\end{lemma}
\begin{proof}
Let \(\mathbf X_k\) denote the \(k\)-th column of a matrix \(\mathbf X\).  We
prove the first residue condition; the other two are obtained in the same way.
By the residue condition for \(\N\), together with the definition of the
conjugating factors, the third column of \(\N^{(1)}\) has a simple pole at
\(z=\zeta_j\), and
\begin{align*}
\Res_{z=\zeta_j}\N^{(1)}_3(x,t,z)
&=
\left(\lim_{z\to\zeta_j}(z-\zeta_j)T_3(z)\right)
\T_\infty^{-1}\N_3(x,t,\zeta_j)\\
&=
\tau_j^{-1}\e^{\theta_{13}(x,t,\zeta_j)}
\left(\lim_{z\to\zeta_j}(z-\zeta_j)\widetilde T_3(z)\right)
\left(\lim_{z\to\zeta_j}\frac{z-\zeta_j}{\widetilde T_1(z)}\right)
\N^{(1)}_1(x,t,\zeta_j).
\end{align*}
Since
\[
\lim_{z\to\zeta_j}(z-\zeta_j)\widetilde T_3(z)
=
\frac{1}{[\widetilde T_3^{-1}]'(\zeta_j)},
\qquad
\lim_{z\to\zeta_j}\frac{z-\zeta_j}{\widetilde T_1(z)}
=
\frac{1}{T_1'(\zeta_j)},
\]
we obtain
\[
\Res_{z=\zeta_j}\N^{(1)}_3(x,t,z)
=
\widetilde C_j \N^{(1)}_1(x,t,\zeta_j),
\]
which is equivalent to~\eqref{E:lstjN1+}. 
\end{proof}

\subsection{The third transformation: removal of poles}
For convenience, we extend the notation for the discrete spectrum by setting
\[
\zeta_{N_1+\ell}:=z_\ell,\qquad 1\le \ell\le N_2.
\]
For each \(\zeta_j\), \(1\le j\le N_1+N_2\), let \(\cD_j\) be a sufficiently
small open disk centered at \(\zeta_j\) with radius \(\eps>0\). We denote by
\(\widehat{\cD}_j\), \(\cD_j^*\), and \(\widehat{\cD}_j^*\) the images of
\(\cD_j\) under the maps \(z\mapsto \hat z\), \(z\mapsto z^*\), and
\(z\mapsto \hat z^*\), respectively. If \(1\le j\le N_1\), then
\[
\widehat{\cD}_j=\cD_j^*,\qquad
\widehat{\cD}_j^*=\cD_j .
\]
We choose \(\eps\) small enough so that all distinct disks are mutually disjoint
and do not intersect the original jump contour. Set
\[
\cD_{\rm sol}
=
\bigcup_{j=1}^{N_1+N_2}
\left(\cD_j\cup \widehat{\cD}_j\cup \cD_j^*\cup \widehat{\cD}_j^*\right),
\]
where repeated disks are counted only once.

We next define a piecewise constant matrix \(\bH(x,t,z)\) on \(\cD_{\rm sol}\).
Motivated by the residue conditions for \(\N^{(1)}\), see
Lemma~\ref{L:mlstj}, we first set
\[
\bH(x,t,z)=
\begin{cases}
\begin{pmatrix}
0&0&\widetilde C_j\\
0&0&0\\
0&0&0
\end{pmatrix},
& z\in \cD_j,\quad 1\le j\le N_1,\\[4mm]
\begin{pmatrix}
0&\widetilde D_{j-N_1}&0\\
0&0&0\\
0&0&0
\end{pmatrix},
& z\in \cD_j,\quad N_1+1\le j\le N_1+N_2,\\[4mm]
\begin{pmatrix}
0&0&0\\
\widehat D_{j-N_1}&0&0\\
0&0&0
\end{pmatrix},
& z\in \cD_j^*,\quad N_1+1\le j\le N_1+N_2 .
\end{cases}
\]
On the remaining image disks, \(\bH\) is defined by the symmetry relation
\[
\bH(x,t,z)
=
-\frac{\hat\zeta_j}{\zeta_j}
\bPi^{-1}(\hat\zeta_j)\,
\bH(x,t,\hat z)\,
\bPi(\hat\zeta_j),
\]
whenever \(z\in\widehat{\cD}_j\) or \(z\in\widehat{\cD}_j^*\) and the right-hand
side is already defined. Since the disks are disjoint, this definition is
unambiguous.

The third transformation is then defined by
\[
\N^{(2)}(x,t,z)=\N^{(1)}(x,t,z)\,\widetilde{\bH}(x,t,z),
\]
where
\[
\widetilde{\bH}(x,t,z)=
\begin{cases}
\I-\dfrac{\bH(x,t,z)}{z-\zeta_j},
& z\in \cD_j,\quad 1\le j\le N_1+N_2,\\[4mm]
\I-\dfrac{\bH(x,t,z)}{z-\zeta_j^*},
& z\in \cD_j^*,\quad 1\le j\le N_1+N_2,\\[4mm]
\I-\dfrac{\bH(x,t,z)}{z-\hat\zeta_j},
& z\in \widehat{\cD}_j,\quad N_1+1\le j\le N_1+N_2,\\[4mm]
\I-\dfrac{\bH(x,t,z)}{z-\hat\zeta_j^*},
& z\in \widehat{\cD}_j^*,\quad N_1+1\le j\le N_1+N_2,\\[4mm]
\I,
& z\in \C\setminus\cD_{\rm sol}.
\end{cases}
\]
Here, in the case \(1\le j\le N_1\), the disks
\(\widehat{\cD}_j\) and \(\widehat{\cD}_j^*\) coincide with
\(\cD_j^*\) and \(\cD_j\), respectively, so no additional cases are needed.

We orient \(\partial\cD_j\) and \(\partial\widehat{\cD}_j^*\) counterclockwise, and orient \(\partial\cD_j^*\) and \(\partial\widehat{\cD}_j\) clockwise. Define \[ \partial\cD_{\rm sol} = \bigcup_{j=1}^{N_1+N_2} \left( \partial\cD_j \cup \partial\widehat{\cD}_j \cup \partial\cD_j^* \cup \partial\widehat{\cD}_j^* \right), \] again with repeated circles counted only once.

 \begin{lemma}\label{L:third-transformation} The function \(\N^{(2)}(x,t,z)\) is analytic at all points of the discrete spectrum. Moreover, it satisfies the jump condition \[ \N^{(2)}_+(x,t,z)=\N^{(2)}_-(x,t,z)\V^{(2)}(x,t,z), \qquad z\in \Sigma^{(2)}:=\Sigma^{(1)}\cup\partial\cD_{\rm sol}, \] where \[ \V^{(2)}(x,t,z)= \begin{cases} \V^{(1)}(x,t,z), & z\in \Sigma^{(1)},\\[2mm] \widetilde{\bH}(x,t,z), & z\in \partial\cD_j\cup\partial\widehat{\cD}_j^*, \quad 1\le j\le N_1+N_2,\\[2mm] \widetilde{\bH}^{-1}(x,t,z), & z\in \partial\cD_j^*\cup\partial\widehat{\cD}_j, \quad 1\le j\le N_1+N_2 . \end{cases} \] Furthermore, there exists a constant \(c>0\) such that, as \(t\to+\infty\), \[ \|\V^{(2)}(x,t,\cdot)-\I\|_{(L^1\cap L^\infty)(\partial\cD_{\rm sol})} = \mathcal O(e^{-ct}), \qquad (x,t) \in\cP_+ . \] 

\end{lemma}

\begin{proof}
We only outline the argument, as the verification is standard.  The factor
\(\widetilde{\bH}\) is constructed from the residue conditions for
\(\N^{(1)}\) so that the principal parts at all discrete spectral points are
cancelled; thus \(\N^{(2)}\) is analytic at these points.  The jumps on
\(\partial\cD_{\rm sol}\) are then obtained directly from the definition
\(\N^{(2)}=\N^{(1)}\widetilde{\bH}\), where the orientation of each circle
determines whether the jump is \(\widetilde{\bH}\) or
\(\widetilde{\bH}^{-1}\).  Finally, the sign distribution of the phase in
\(\cP_+\) gives exponential decay of all residue constants appearing in
\(\bH\).  Since the circles have fixed radii,
we obtain, for some \(c>0\),
\[
\|\V^{(2)}-\I\|_{(L^1\cap L^\infty)(\partial\cD_{\rm sol})}
=\mathcal O(e^{-ct}), \qquad t \to \infty, 
\]
uniformly for \((x,t) \in\cP_+\).
\end{proof}

\subsection{The local parametrix }\label{sub:lG}
 In the previous subsections, we have shown that the jump matrix
\(\V^{(2)}\) converges uniformly to \(\I\), except in a neighborhood of
\(-q_0\). In this subsection, we construct a local parametrix near \(-q_0\)
in terms of the solution of the model RH problem~\ref{RHP:model}.

The first step involves constructing approximations for the phase function $\theta_{13}(x,t,z)$, which is defined by
$$
\theta_{13}(x,t,z)=\ii (- z+\frac{ q_0^2}{z})x+\ii (z^2-\frac{q_0^4}{z^2})t,
$$
 and appears in the exponential of the $(1,3)$ entry of the jump matrix $\V^{(2)}$. The behavior of this function is essential for matching the  Painlev\'e II  model problem.
For this purpose, by performing a Taylor expansion of this function at $z=-q_0$, we obtain
$$
\theta_{13}(x,t,z)=-2\ii \left[
(x+2t q_0)(z+q_0)+(\frac{x}{2q_0}+t)(z+q_0)^2+\frac{x+4 tq_0}{2q_0^2}(z+q_0)^3
\right]+\tilde{S}_1(x,t,z),
$$
where the remainder term $\tilde{S}_1(x,t,z)$ is given  by
$$
\tilde{S}_1(x,t,z)=t
\frac{\partial_k^4 \Phi_{13}(\xi,z_{\star})}{4 !}(z+q_0)^4.
$$
Here $z_{\star}$ denotes a point  lies on the line segment between $z$ and $-q_0$.

Introduce the variables $y$ and $\beta$ by
\be\label{E:ybeta}
y=2\left(\frac{4}{3}q_0 \right)^{1/3}(\xi+q_0)t^{2/3}, \quad
\beta=\left(\frac{3t}{4q_0}\right)^{1/3}(z+q_0),
\ee
and define function $S_1(x,t,z)$ by
\be \label{E:S1}
S_1(x,t,z)=-2\ii (\frac{x}{2q_0}+t)(z+q_0)^2-\ii \left(\frac{x+2 tq_0}{q_0^2}\right)(z+q_0)^3+\tilde{S}_1(x,t).
\ee
Then the function $\theta_{13}(x,t,z)$ can be expressed as follows :
\be \label{E:betabs}
\theta_{13}(x,t,z)=-2 \ii \left( y \beta+\frac{4}{3}\beta^3 \right)+S_1(x,t,z).
\ee

We also need the corresponding expansions of \(\theta_{12}\) and
\(\theta_{32}\) near \(-q_0\). A Taylor expansion gives
\begin{align*}
\theta_{32}(x,t,z)&=\theta_{32}(x,t,-q_0)+ \ii (x+2q_0 t)(z+q_0)+\left[\ii t +\ii \frac{2t}{q_0}(\xi+q_0)\right](z+q_0)^2\\
&\quad +\left[ \frac{\ii}{q_0^2} (x+2q_0  t)+\frac{2 \ii t}{q_0} \right]
(z+q_0)^3+\tilde{S}_2(x,t,z),\\
\theta_{12}(x,t,z)&=\theta_{12}(x,t,-q_0)-\ii(x+2t q_0)(z+q_0)+\ii t(z+q_0)^2,
\end{align*}
where the remainder term is defined as
$$
\tilde{S}_2(x,t,z)=t
\frac{\partial_k^4 \Phi_{32}(\xi,z_{\star})}{4 !}(z+q_0)^4,
$$
with \(z_\star\) lying on the line segment joining \(z\) and \(-q_0\). 
Let
\be \label{E:S2}
S_2(x,t,z)=\ii \frac{2t}{q_0}(\xi+q_0)(z+q_0)^2+\left[ \frac{\ii}{q_0^2} (x+2q_0  t)+\frac{2 \ii t}{q_0} \right]
(z+q_0)^3+\tilde{S}_2(x,t,z).
\ee
Then the functions $\theta_{12}(x,t,z)$ and $\theta_{32}(x,t,z)$ can be expressed as
\begin{align*}
\theta_{32}(x,t,z)=\theta_{32}(x,t,-q_0)+\ii y \beta(z)+\ii \left(\frac{4}{3}q_0\right)^{2/3}t^{1/3}\beta^2(z)+S_2(x,t,z)
\end{align*}
and
\begin{align*}
\theta_{12}(x,t,z)=\theta_{12}(x,t,-q_0)-\ii y \beta(z)+\ii \left(\frac{4}{3}q_0\right)^{2/3}t^{1/3}\beta^2(z),
\end{align*}
respectively.
Recall that $\cD_{\eps}$ denote the open disk of radius $\eps$ centered at the
point $-q_0$.  Then $z \to \beta$ is a biholomorphism from  $\cD_{\eps}$ onto the open disk of radius  $\left(3t/(4q_0)\right)^{1/3} \eps $  centered at the origin.  In particular,  the images of the critical points $\{ \rz_j\}_{j=0}^3$ of the phase functions $\Phi_{21}(\xi,z)$, $\Phi_{32}(\xi,z)$ and $\Phi_{31}(\xi,z)$ in the $\beta$-plane are $\{ \beta_j\}_{j=0}^3$, and their expressions are as follows:
\be \label{E:betaj}
\beta_j=\left(\frac{3t}{4q_0}\right)^{1/3}(\rz_j+q_0).
\ee
Let
$\X^{\eps}=\Sigma^{(2)} \cap \cD_{\eps}$ be as shown in Fig~\ref{XepsFig}.  Each contour
\(\X_j^\eps\), \(1\le j\le 8\), forms an angle either \(\pi/6\) or
\(5\pi/6\) with the real axis.

\begin{figure}
\centering
\begin{tikzpicture}[
    scale=0.6,
    line cap=round,
    line join=round,
    >=latex,
    contour/.style={line width=2pt, draw=black},
    arrowcontour/.style={
        contour,
        postaction={decorate},
        decoration={markings,mark=at position 0.58 with {\arrow{>}}}
    },
    circstyle/.style={line width=1.8pt, draw=black},
    every node/.style={font=\normalsize, inner sep=1pt}
]

\def\R{3.55}              
\def\h{0.72}              
\def\xp{-0.62}            
\def\xq{0.62}             
\def\xLO{-4.55}           
\def\xLI{-4.15}           
\def\xRI{4.15}            
\def\xRO{4.55}            

\pgfmathsetmacro{\t}{tan(30)}

\coordinate (O)  at (0,0);
\coordinate (U)  at (0,\h);
\coordinate (P)  at (\xp,0);
\coordinate (M)  at (0,0);
\coordinate (Q)  at (\xq,0);
\coordinate (D)  at (0,-\h);

\pgfmathsetmacro{\yATwo}{\h - \t*\xLO}
\coordinate (A2) at (\xLO,\yATwo);

\pgfmathsetmacro{\yASix}{-\t*(\xLI-\xp)}
\coordinate (A6) at (\xLI,\yASix);

\pgfmathsetmacro{\yBSeven}{\t*(\xLI-\xp)}
\coordinate (B7) at (\xLI,\yBSeven);

\pgfmathsetmacro{\yBThree}{-\h + \t*\xLO}
\coordinate (B3) at (\xLO,\yBThree);

\pgfmathsetmacro{\yTOne}{\h + \t*\xRO}
\coordinate (T1) at (\xRO,\yTOne);

\pgfmathsetmacro{\yTFive}{\t*(\xRI-\xq)}
\coordinate (T5) at (\xRI,\yTFive);

\pgfmathsetmacro{\ySEight}{-\t*(\xRI-\xq)}
\coordinate (S8) at (\xRI,\ySEight);

\pgfmathsetmacro{\ySFour}{-\h - \t*\xRO}
\coordinate (S4) at (\xRO,\ySFour);

\draw[circstyle] (O) circle (\R);

\draw[circstyle, ->] (0.15,\R) arc[start angle=88,end angle=104,radius=\R];

\draw[arrowcontour] (U) -- (A2);   
\draw[arrowcontour] (P) -- (A6);   

\draw[arrowcontour] (U) -- (T1);   
\draw[arrowcontour] (Q) -- (T5);   

\draw[arrowcontour] (B7) -- (P);   
\draw[arrowcontour] (B3) -- (D);   

\draw[arrowcontour] (S4) -- (D);   
\draw[arrowcontour] (S8) -- (Q);   

\draw[arrowcontour] (P) -- (Q);    

\fill (U) circle (2.2pt);
\fill (P) circle (2.2pt);
\fill (M) circle (2.2pt);
\fill (Q) circle (2.2pt);
\fill (D) circle (2.2pt);

\node at ( 1.95,  2.28) {$1$};
\node at (-1.95,   2.28) {$2$};

\node at (-2.00, - 2.28) {$3$};
\node at ( 1.95, - 2.28) {$4$};

\node at ( 2.35,  0.46) {$5$};
\node at (-2.35,  0.46) {$6$};

\node at (-2.35, -0.46) {$7$};
\node at ( 2.35, -0.46) {$8$};

\node at (0,-0.30) {$9$};

\end{tikzpicture}
\caption{{\protect\small
The solid contour inside the circle is $\X^{\eps}$, and the disk enclosed by the circle is $\cD_{\eps}$. }}
\label{XepsFig}
\end{figure}

To demonstrate how to construct model RH problem~\ref{RHP:model} we need to study the behavior of the reflection coefficients $\{ r_j\}_{j=1}^3$ at the branch point $-q_0$. First, it follows directly from~\eqref{E:asyratzd} that  $\lim\limits_{ z \to -q_0} r_2(z)=\ii$.
Now, it  remains to analyze the other two reflection coefficients.
One can  observe that $r^*_3(z)$ vanishes at the branch point $-q_0$, while $\gamma(z)$ has a simple zero at $-q_0$; thus, their ratio is well-defined there. We define $\tilde{s}$ to be this ratio, i.e.,
\be
\tilde{s}:=\lim_{ z \to -q_0} \frac{r^*_3(z^*)}{\gamma(z)}.
\ee
Then, by symmetry, we immediately obtain
$$
\lim_{ z \to -q_0} \frac{r^*_1(z^*)}{\gamma(z)}=-\ii \tilde{s}.
$$
We also need the limiting values of the scalar factors
\(\widetilde T_2/\widetilde T_1\) and \(\widetilde T_2/\widetilde T_3\) at \(-q_0\). By Lemma~\ref{L:dxz}, these limits are well defined and satisfy
$$
\frac{\widetilde T_2}{\widetilde T_1}(-q_0)=\frac{\widetilde T_2}{\widetilde T_3}(-q_0)=\frac{1}{\left( \delta(-q_0) \delta_1(-q_0) \rP_1(-q_0) \right)^3}.
$$
\begin{lemma}\label{L:Tsharpq0}
As \(z\to -q_0\) along any direction not tangent to the real axis, the following estimates hold:
\begin{subequations}
\begin{align}
\left| \frac{\widetilde T_1}{\widetilde T_3}(z) -1 \right| \leq C |\log |z+q_0||  |z+q_0|,\label{E:gjT3}\\
\left| \frac{\widetilde T_2}{\widetilde T_1}(z) -\frac{\widetilde T_2}{\widetilde T_1}(-q_0) \right| \leq C |\log |z+q_0|| |z+q_0|, \label{E:gjT2} \\
\left| \frac{\widetilde T_2}{\widetilde T_3}(z) -\frac{\widetilde T_2}{\widetilde T_3}(-q_0) \right| \leq C |\log |z+q_0||  |z+q_0|.\label{E:gjT}
\end{align}
\end{subequations}
\end{lemma}
\begin{proof}
The estimates follow directly from  Lemma~\ref{L:dxz}.
\end{proof}
\begin{remark}\label{R:Tsharpboundary} 
The estimates in Lemma~\ref{L:Tsharpq0} also hold for the boundary values of \(\widetilde \T\) on the jump contour. Indeed, by the Sokhotski--Plemelj formula and the standard endpoint estimate for Cauchy transforms whose density vanishes at the endpoint \(-q_0\), the non-tangential estimates for \(\widetilde \T\) extend to its boundary values. Hence, for \(z\) on the cut and sufficiently close to \(-q_0\), all estimates in Lemma~\ref{L:Tsharpq0} remain valid with
\(\widetilde T\) replaced by its boundary values \(\widetilde T_{j \pm}\).
\end{remark}

We now define
\be \label{E:deffs}
s=\tilde{s} \frac{\widetilde T_2}{\widetilde T_3}(-q_0)=\frac{\tilde{s} }{\left( \delta(-q_0) \delta_1(-q_0) \rP_1(-q_0) \right)^3},
\ee which is the parameter that appears in model RH problem~\ref{RHP:model}.
In order to relate $\N^{(2)}$ to $\N^{\X}$, we define
\be
\Y(x,t)
=
\diag\left(
\e^{\theta_1(x,t,-q_0)},
\e^{\theta_2(x,t,-q_0)},
\e^{\theta_3(x,t,-q_0)}
\right)
=
\diag\left(
1,
\e^{-\ii q_0 x-\ii q_0^2t},
1
\right).
\ee
Let
$$
\widetilde{\N}(x,t,z)=\N^{(2)}(x,t,z) \Y(x,t), \qquad z \in \cD_{\eps}.
$$
Then the jump matrix of \(\widetilde{\N}\) is $\widetilde{\V}=\Y^{-1} \V^{(2)} \Y$. Let $\widetilde{\V}|_{\X^{\eps}}$ denotes the restriction of $\widetilde{\V}$ to $\X^{\eps}$. For a fixed $\beta$,  as $t \to \infty$~(this leads to $z \to -q_0$), we have
$$
 r_{2}(z) \frac{\widetilde T_1}{\widetilde T_3}(z) \to  \ii,\qquad
 \frac{r^{*}_{3}(z^*)}{\gamma(z)} \frac{\widetilde T_2}{\widetilde T_3}(z) \to  s, \quad \frac{r^{*}_{1}(z^*) }{\gamma(z)} \frac{\widetilde T_2}{\widetilde T_1}(z)  \to   -\ii s.
$$
This suggests that $\widetilde{\V}|_{\X^{\eps}}(x,t,z(\beta))$ tends to the jump matrix $\V^{\X}(x,t,\beta)$ defined in~\eqref{E:exVjX} for large $t$. The above analysis suggests that we can approximate $\N^{(2)}(x,t,z)$ in $\cD_{\eps}$ by the $3 \times 3$ matrix-valued function $\N^{loc}(x,t,z)$ defined by
\be \label{E:Nloc}
\N^{loc}(x,t,z)= \Y  \N^{\X}(y,s,t;\beta(z)) \Y^{-1}, \quad z \in \cD_{\eps},
\ee
where $\N^{\X}(y,s,t;\beta)$ is the solution of the model RH problem~\ref{RHP:model}.

In order to justify the local approximation of \(\N^{(2)}\) by
\(\N^{loc}\) in \(\cD_\eps\), we first record several estimates which will be
used below.   We recall that the scaled
variable \(y\) remains bounded for \((x,t)\in\cP_+\).
For notational convenience, define
\[
\begin{alignedat}{2}
\Theta_{13}^{\rm mod}(z)&:=-2\ii\Omega_o(z),
\qquad \qquad  \qquad &
\Theta_{31}^{\rm mod}(z)&:=2\ii\Omega_o(z),\\
\Theta_{32}^{\rm mod}(z)&:=\ii y\beta(z)+\ii\Lambda_o(z),
&
\Theta_{12}^{\rm mod}(z)&:=-\ii y\beta(z)+\ii\Lambda_o(z),\\
\Theta_{23}^{\rm mod}(z)&:=-\ii y\beta(z)-\ii\Lambda_o(z),
&
\Theta_{21}^{\rm mod}(z)&:=\ii y\beta(z)-\ii\Lambda_o(z).
\end{alignedat}
\]
Finally, set
\[
\begin{alignedat}{2}
F_{32}(x,t,z)
&:=
\e^{\widehat\theta_{32}(x,t,z)}
-\e^{\Theta_{32}^{\rm mod}(z)},
\qquad&
F_{23}(x,t,z)
&:=
\e^{\widehat\theta_{23}(x,t,z)}
-\e^{\Theta_{23}^{\rm mod}(z)},\\
F_{13}(x,t,z)
&:=
\e^{\theta_{13}(x,t,z)}
-\e^{\Theta_{13}^{\rm mod}(z)},
\qquad&
F_{31}(x,t,z)
&:=
\e^{\theta_{31}(x,t,z)}
-\e^{\Theta_{31}^{\rm mod}(z)},\\
\widehat\theta_{32}(x,t,z)
&:=
\theta_{32}(x,t,z)-\theta_{32}(x,t,-q_0),
\qquad&
\widehat\theta_{23}(x,t,z)
&:=
\theta_{23}(x,t,z)-\theta_{23}(x,t,-q_0).
\end{alignedat}
\]

\begin{lemma}\label{L:mmqm}
For \((x,t)\in\cP_+\) and sufficiently large \(t\), the following estimates
hold:
\begin{subequations}
\begin{align}
\left|\e^{\Theta_{13}^{\rm mod}(z)}\right|
&\le C\e^{-ct|z+q_0|^3},
&& z\in\X_3^\eps\cup\X_4^\eps \cup \X_7^\eps, \label{mm1}\\
\left|\e^{\Theta_{31}^{\rm mod}(z)}\right|
&\le C\e^{-ct|z+q_0|^3},
&& z\in\X_1^\eps\cup\X_2^\eps \cup \X_6^\eps, \label{mm2}\\
\left|\e^{\Theta_{32}^{\rm mod}(z)}\right|
&\le C\e^{-ct|z+q_0|^2},
&& z\in\X_5^\eps\cup\X_7^\eps, \label{mm3}\\
\left|\e^{\Theta_{12}^{\rm mod}(z)}\right|
&\le C\e^{-ct|z+q_0|^2},
&& z\in\X_5^\eps\cup\X_7^\eps, \label{mm3333}\\
\left|\e^{\Theta_{23}^{\rm mod}(z)}\right|
&\le C\e^{-ct|z+q_0|^2},
&& z\in\X_6^\eps\cup\X_8^\eps, \label{mm4}\\
\left|\e^{\Theta_{21}^{\rm mod}(z)}\right|
&\le C\e^{-ct|z+q_0|^2},
&& z\in\X_6^\eps\cup\X_8^\eps. \label{mm4444}
\end{align}
\end{subequations}
\end{lemma}

\begin{proof}
By the construction of the local contour \(\X^\eps\), the cubic and quadratic
parts of the model phases have the required signs on the corresponding rays.
More precisely, there exists \(c_0>0\) such that
\begin{align*}
\mathrm {Re}\bigl(-\frac{8\ii}{3}\beta^3\bigr)\le -c_0|\beta|^3,
\quad z\in\X_3^\eps\cup\X_4^\eps \cup \X_7^\eps; \qquad
\mathrm {Re} \bigl(\frac{8\ii}{3}\beta^3\bigr)\le -c_0|\beta|^3,
\quad z\in\X_1^\eps\cup\X_2^\eps \cup \X_6^\eps.
\end{align*}
Since \(y\) is bounded in \(\cP_+\), we have
\[
|\mathrm {Re} \bigr(\pm2\ii y\beta \bigr)|
\le C|\beta|
\le \frac{c_0}{2}|\beta|^3+C.
\]
Together with
$
|\beta|^3=\frac{3t}{4q_0}|z+q_0|^3
$,
this proves \eqref{mm1} and \eqref{mm2}.

Similarly,  we have
$$
\mathrm {Re} \bigl(\ii\Lambda_o(z) \bigr)
\le -c_0t^{1/3}|\beta|^2, \quad  z \in \X_5^\eps\cup\X_7^\eps; \qquad
\mathrm {Re} \bigl(-\ii\Lambda_o(z)\bigr)
\le -c_0t^{1/3}|\beta|^2, \quad z\in \X_6^\eps\cup\X_8^\eps.
$$
Using again the boundedness of \(y\), we obtain, for sufficiently large \(t\),
\[
|\mathrm {Re}\bigl(\pm\ii y\beta \bigr)| \le C |\beta|
\le \frac{c_0}{2}t^{1/3}|\beta|^2+C.
\]
Since
\[
t^{1/3}|\beta|^2
=
\left(\frac{3}{4q_0}\right)^{2/3}t|z+q_0|^2,
\]
we obtain~\eqref{mm3}, \eqref{mm3333}, \eqref{mm4}, and \eqref{mm4444}.
\end{proof}

\begin{lemma}\label{L:Phaseest}
If \(\eps>0\) is sufficiently small, then for all sufficiently large \(t\) and
\((x,t)\in\cP_+\), we have
\begin{subequations}
\begin{align}
|\e^{t\Phi_{31}(\xi,z)}|
&\le C\e^{-ct|z+q_0|^3},
&& z\in\X_1^\eps\cup\X_2^\eps\cup\X_6^\eps, \label{rePhiIIIa}\\
|\e^{t\Phi_{13}(\xi,z)}|
&\le C\e^{-ct|z+q_0|^3},
&& z\in\X_3^\eps\cup\X_4^\eps\cup\X_7^\eps, \label{rePhiIIIzj}\\
|\e^{t\Phi_{12}(\xi,z)}|
&\le C\e^{-ct|z+q_0|^2},
&& z\in\X_5^\eps\cup\X_7^\eps, \label{rePhiIIIe}\\
|\e^{t\Phi_{32}(\xi,z)}|
&\le C\e^{-ct|z+q_0|^2},
&& z\in\X_5^\eps\cup\X_7^\eps, \label{rePhiIIIehz}\\
|\e^{t\Phi_{21}(\xi,z)}|
&\le C\e^{-ct|z+q_0|^2},
&& z\in\X_6^\eps\cup\X_8^\eps, \label{rePhiIIIb}\\
|\e^{t\Phi_{23}(\xi,z)}|
&\le C\e^{-ct|z+q_0|^2},
&& z\in\X_6^\eps\cup\X_8^\eps. \label{rePhiIIIbhb}
\end{align}
\end{subequations}
\end{lemma}

\begin{proof}
We prove \eqref{rePhiIIIzj} and \eqref{rePhiIIIehz}; the remaining estimates are
obtained in the same way.  By \eqref{E:betabs},
\[
\theta_{13}(x,t,z)
=
\Theta_{13}^{\rm mod}(z)+S_1(x,t,z).
\]
For \(z\in\X_3^\eps\cup\X_4^\eps\), Lemma~\ref{L:mmqm} gives, for some
\(c_0>0\),
\[
\left|\e^{\Theta_{13}^{\rm mod}(z)}\right|
\le C\e^{-c_0t|z+q_0|^3}.
\]
Moreover, from \eqref{E:S1} and the definition of \(\cP_+\),
\[
|S_1(x,t,z)|
\le
C\left(
t^{1/3}|z+q_0|^2
+t^{1/3}|z+q_0|^3
+t|z+q_0|^4
\right).
\]
After decreasing \(\eps>0\) if necessary and then taking \(t\) sufficiently large,
for any fixed \(\kappa>0\) we have
\[
|S_1(x,t,z)|
\le
\kappa t|z+q_0|^3+C_\kappa 
\]
for some constant \(C_\kappa>0\) depending only on \(\kappa\).
Choosing \(\kappa>0\) sufficiently small compared with \(c_0\), we obtain
\[
|\e^{\theta_{13}(x,t,z)}|
\le
C\e^{-c_0t|z+q_0|^3+\re S_1(x,t,z)}
\le
C\e^{-ct|z+q_0|^3},
\qquad z\in\X_3^\eps\cup\X_4^\eps .
\]
The same estimate on \(\X_7^\eps\) follows from the same local expansion and
the corresponding sign of \(\Re\Phi_{13}\) on that ray.  Hence
\eqref{rePhiIIIzj} follows.

Next, from the local expansion of \(\theta_{32}\),
\[
\theta_{32}(x,t,z)-\theta_{32}(x,t,-q_0)
=
\Theta_{32}^{\rm mod}(z)+S_2(x,t,z).
\]
For \(z\in\X_5^\eps\cup\X_7^\eps\), Lemma~\ref{L:mmqm} implies
\[
\left|\e^{\Theta_{32}^{\rm mod}(z)}\right|
\le C\e^{-c_1t|z+q_0|^2}, \quad \text{for some \(c_1>0\).}
\]
  From \eqref{E:S2},
\[
|S_2(x,t,z)|
\le
C\left(
t^{1/3}|z+q_0|^2
+t|z+q_0|^3
+t|z+q_0|^4
\right).
\]
Hence, after choosing \(\eps>0\) sufficiently small and then \(t\) sufficiently
large,
\[
|S_2(x,t,z)|
\le
\kappa t|z+q_0|^2.
\]
Taking $\kappa>0$ sufficiently small compared with $c_1$, we conclude that
$$
|\e^{\theta_{32}(x,t,z)}|
\le
C\e^{-c_1t|z+q_0|^2+\re S_2(x,t,z)}
\le
C\e^{-ct|z+q_0|^2},
\quad z\in\X_5^\eps\cup\X_7^\eps.
$$
This proves \eqref{rePhiIIIehz}.
The estimate \eqref{rePhiIIIe} follows in
the same way from
$$
\theta_{12}(x,t,z)-\theta_{12}(x,t,-q_0)
=
\Theta_{12}^{\rm mod}(z),
$$
together with the fact that $\theta_{12}(x,t,-q_0)$ is purely imaginary.
The estimates \eqref{rePhiIIIa}, \eqref{rePhiIIIb}, and
\eqref{rePhiIIIbhb} are analogous.
\end{proof}

\begin{lemma}\label{L:S12}
Let \((x,t)\in\cP_+\) and let \(\eps>0\) be sufficiently small.  Then, as
\(t\to+\infty\),
\begin{subequations}
\begin{align}
\left|
F_{13}(x,t,z)
\right|
&\le Ct^{-1/3},
&& z\in\X_3^\eps\cup\X_4^\eps \cup \X_7^\eps, \label{SIS2est}\\
\left|
F_{31}(x,t,z)
\right|
&\le Ct^{-1/3},
&& z\in\X_1^\eps\cup\X_2^\eps \cup \X_6^\eps, \label{SiS21}\\
|F_{32}(x,t,z)|
&\le Ct^{-1/2},
&& z\in\X_5^\eps\cup\X_7^\eps, \label{SiS22}\\
|F_{23}(x,t,z)|
&\le Ct^{-1/2},
&& z\in\X_6^\eps\cup\X_8^\eps. \label{S1S23}
\end{align}
\end{subequations}
\end{lemma}

\begin{proof}
We use the elementary inequality
\be \label{E:gine}
|\e^{w_1}-\e^{w_2}|
\le
\left(|\e^{w_1}|+|\e^{w_2}|\right)|w_1-w_2|,
\qquad w_1,w_2\in\C .
\ee
By \eqref{E:betabs}, the difference between the two exponents in
\eqref{SIS2est} is \(S_1(x,t,z)\).  Hence, by \eqref{E:gine} and
Lemmas~\ref{L:mmqm}--\ref{L:Phaseest},
\[
\left|
\e^{\theta_{13}(x,t,z)}
-
\e^{\Theta_{13}^{\rm mod}(z)}
\right|
\le
C|S_1(x,t,z)|\e^{-ct|z+q_0|^3},
\]
for \(z\in\X_3^\eps\cup\X_4^\eps \cup \X_7^\eps\).  Since
\[
|S_1(x,t,z)|
\le
C\left(
t^{1/3}|z+q_0|^2
+t^{1/3}|z+q_0|^3
+t|z+q_0|^4
\right),
\]
setting \(u=t^{1/3}|z+q_0|\) gives
\begin{align*}
|S_1(x,t,z)|\e^{-ct|z+q_0|^3}
\le
C\left(
t^{-1/3}u^2
+t^{-2/3}u^3
+t^{-1/3}u^4
\right)\e^{-cu^3} \le Ct^{-1/3}.
\end{align*}
This proves \eqref{SIS2est}.  The proof of \eqref{SiS21} is identical.

It remains to prove \eqref{SiS22}; the proof of \eqref{S1S23} is the same.
Using the definition of \(F_{32}\), the expansion
\[
\theta_{32}(x,t,z)-\theta_{32}(x,t,-q_0)
=
\Theta_{32}^{\rm mod}(z)+S_2(x,t,z),
\]
and \eqref{E:gine}, we obtain
\[
|F_{32}(x,t,z)|
\le
C|S_2(x,t,z)|\e^{-ct|z+q_0|^2},
\qquad z\in\X_5^\eps\cup\X_7^\eps .
\]
Moreover,
\[
|S_2(x,t,z)|
\le
C\left(
t^{1/3}|z+q_0|^2
+t|z+q_0|^3
+t|z+q_0|^4
\right).
\]
Let \(v=t^{1/2}|z+q_0|\).  Then
\begin{align*}
|S_2(x,t,z)|\e^{-ct|z+q_0|^2}
\le
C\left(
t^{-2/3}v^2
+t^{-1/2}v^3
+t^{-1}v^4
\right)\e^{-cv^2}\le Ct^{-1/2}.
\end{align*}
This proves \eqref{SiS22}, and hence the lemma.
\end{proof}

\begin{remark}\label{R:fj}
The corresponding $L^1$-bounds can be obtained in the same way. More
precisely, as $t\to+\infty$,
\be \label{E:Relyy}
\begin{aligned}
\int_{\X_3^\eps\cup\X_4^\eps} \left|F_{13}(x,t,z)\right| \,d|z|
&\le Ct^{-2/3},&
\int_{\X_1^\eps\cup\X_2^\eps}\left| F_{31}(x,t,z) \right| \,d|z|
&\le Ct^{-2/3},\\
\int_{\X_5^\eps\cup\X_7^\eps} \left| F_{32}(x,t,z) \right| \,d|z|
&\le Ct^{-1},&
\int_{\X_6^\eps\cup\X_8^\eps} \left| F_{23}(x,t,z) \right| \,d|z|
&\le Ct^{-1}.
\end{aligned}
\ee
In particular, all four integrals are \(\mathcal O(t^{-2/3})\).
\end{remark}

We are now ready to show that \(\N^{(2)}\) can be approximated by
\(\N^{loc}\) in \(\cD_\eps\).

\begin{lemma}\label{L:Nloc}
For each \((x,t)\), the function \(\N^{loc}(x,t,z)\) defined in~\eqref{E:Nloc}
is analytic and bounded for \(z\in\cD_\eps\setminus\X^\eps\).  Across
\(\X^\eps\), it satisfies
\[
\N^{loc}_+(x,t,z)=\N^{loc}_-(x,t,z)\V^{loc}(x,t,z).
\]
Moreover, for sufficiently large \(t\),
\begin{align}\label{E:estVloc}
\begin{cases}
\|\V^{(2)}-\V^{loc}\|_{L^\infty(\X^\eps)}
\le C t^{-1/3}\log t,\\[1mm]
\|\V^{(2)}-\V^{loc}\|_{L^1(\X^\eps)}
\le C t^{-2/3}\log t,
\end{cases}
\qquad (x,t)\in\cP_+ .
\end{align}
Furthermore, as \(t\to+\infty\),
\begin{align}
\|(\N^{loc})^{-1}-\I\|_{L^\infty(\partial\cD_\eps)}
&=\mathcal O(t^{-1/3}), \label{E:bzsy}\\
(\N^{loc})^{-1}(x,t,z)-\I
&=
-\frac{
\Y\bigl(\cE^\X_1+\M^P_1\bigr)\Y^{-1}
}{
\left(\frac{3}{4q_0}\right)^{1/3}t^{1/3}(z+q_0)
}
+\mathcal O(t^{-2/3}),
\quad z\in\partial\cD_\eps, \label{E:bzsy1}
\end{align}
where \(\cE^\X_1\) and \(\M^P_1\) are given by~\eqref{E:estcEX1} and
\eqref{m1Pexpression}, respectively.
\end{lemma}

\begin{proof}
The analyticity and boundedness of \(\N^{loc}\) follow directly from the
definition~\eqref{E:Nloc} and from the corresponding properties of the model
solution \(\N^\X\).  The jump relation is inherited from the model RH problem.
Indeed, on \(\X^\eps\),
\[
\V^{(2)}-\V^{loc}
=
\Y\bigl(\widetilde\V-\V^\X\bigr)\Y^{-1}.
\]
Since \(\Y\) is uniformly bounded, it suffices to estimate
\(\widetilde\V-\V^\X\).  We give the details on the representative contours
\(\X_1^\eps\), \(\X_7^\eps\), and \(\X_9^\eps\); the remaining contours are
treated in the same way.

\medskip
\noindent{\bf The contour \(\X_1^\eps\).}
A direct calculation gives
\[
\widetilde\V-\V^\X
=
\begin{pmatrix}
0&0&0\\
0&0&0\\
f(x,t,z)&0&0
\end{pmatrix},
\qquad z\in\X_1^\eps,
\]
where
\[
f(x,t,z)
=
r_2(z)\frac{\widetilde T_1}{\widetilde T_3}(z)\e^{\theta_{31}(x,t,z)}
-\ii\e^{\Theta_{31}^{\rm mod}(z)} .
\]
We decompose
\[
\begin{aligned}
f(x,t,z)
&=
\bigl(
r_2(z)\frac{\widetilde T_1}{\widetilde T_3}(z)-\ii
\bigr)\e^{\theta_{31}(x,t,z)}
+\ii\left(\e^{\theta_{31}(x,t,z)}-\e^{\Theta_{31}^{\rm mod}(z)}\right).
\end{aligned}
\]
By the behavior of \(r_2\) at \(-q_0\) and Lemma~\ref{L:Tsharpq0},
\[
\left|
r_2(z)\frac{\widetilde T_1}{\widetilde T_3}(z)-\ii
\right|
\le
C |z+q_0|\,|\log|z+q_0|| .
\]
Using \eqref{rePhiIIIa} and \eqref{SiS21}, we obtain
\[
|f(x,t,z)|
\le
C |z+q_0|\,|\log |z+q_0||\,\e^{-ct|z+q_0|^3}
+Ct^{-1/3}.
\]
Hence
$
\|f\|_{L^\infty(\X_1^\eps)}
\le
C t^{-1/3}\log t 
$.
Moreover, by Remark~\ref{R:fj},
$$
\begin{aligned}
\|f\|_{L^1(\X_1^\eps)}
\le
C\int_0^\infty \tau |\log \tau|\,\e^{-ct\tau^3}\,d\tau
+Ct^{-2/3}  
\le
Ct^{-2/3}\log t .
\end{aligned}
$$
Thus \eqref{E:estVloc} holds on \(\X_1^\eps\).

\medskip
\noindent{\bf The contour \(\X_7^\eps\).}
For \(z\in\X_7^\eps\), we have
\[
\widetilde\V-\V^\X
=
\begin{pmatrix}
0&g_1(x,t,z)&g_2(x,t,z)\\
0&0&0\\
0&g_3(x,t,z)&0
\end{pmatrix},
\]
where
\begin{align*}
g_1(x,t,z)
&=
\tilde g_1(z)\frac{\widetilde T_2}{\widetilde T_1}(z)
\e^{\Theta_{12}^{\rm mod}(z)}
-\ii s\,\e^{\Theta_{12}^{\rm mod}(z)},\\
\tilde g_1(x,t,z)
&=
-\frac{1}{\gamma(z)}r^*_1(z^*)
+\frac{1}{\gamma^2(z)} r_1^*(z^*)r_3(z)R_3^*(z^*),\\
g_2(x,t,z)
&=
(r_2^*(z^*)-R_2^*(z^*))
\frac{\widetilde T_3}{\widetilde T_1}(z)\e^{\theta_{13}(x,t,z)},\\
g_3(x,t,z)
&=
-\frac{1}{\gamma(z)}R_3^*(z^*)
\frac{\widetilde T_2}{\widetilde T_3}(z)
\e^{\theta_{32}(x,t,z)-\theta_{32}(x,t,-q_0)}
+s\,\e^{\Theta_{32}^{\rm mod}(z)}.
\end{align*}
Here and below the diagonal conjugation by \(\widetilde \T\) is harmless, since the
ratios of its diagonal entries are uniformly bounded on \(\X^\eps\).

By the definitions of \(s\), \(\Theta_{12}^{\rm mod}\), and the estimates in
Lemma~\ref{L:Tsharpq0}, the coefficient in \(g_1\) satisfies
\[
\biggl|
\tilde g_1(z)\frac{\widetilde T_2}{\widetilde T_1}(z)-\ii s
\biggr|
\le
C|z+q_0|\,|\log |z+q_0|| .
\]
Together with \eqref{mm3}, this gives
\[
|g_1(x,t,z)|
\le
C|z+q_0|\,|\log|z+q_0||\,\e^{-ct|z+q_0|^2}.
\]
Consequently,
\[
\|g_1\|_{L^\infty(\X_7^\eps)}
\le Ct^{-1/2}\log  t,
\qquad
\|g_1\|_{L^1(\X_7^\eps)}
\le Ct^{-1}\log t .
\]

For \(g_3\), we write
\[
\begin{aligned}
g_3(x,t,z)
&=
\biggl(
-\frac{1}{\gamma(z)}R_3^*(z^*)
\frac{\widetilde T_2}{\widetilde T_3}(z)
+s
\biggr)\e^{\Theta_{32}^{\rm mod}(z)}  
-\frac{1}{\gamma(z)}R_3^*(z^*)
\frac{\widetilde T_2}{\widetilde T_3}(z)F_{32}(x,t,z).
\end{aligned}
\]
Using Lemma~\ref{L:Tsharpq0}, \eqref{mm3}, and \eqref{SiS22}, we get
\[
|g_3(x,t,z)|
\le
C|z+q_0|\,|\log|z+q_0||\,\e^{-ct|z+q_0|^2}
+Ct^{-1/2}.
\]
Thus
\[
\|g_3\|_{L^\infty(\X_7^\eps)}
\le Ct^{-1/2}\log t,
\]
and, by Remark~\ref{R:fj},
\[
\|g_3\|_{L^1(\X_7^\eps)}
\le
C\int_0^\infty \tau|\log \tau|\e^{-ct\tau^2}\,d\tau
+Ct^{-1}
\le Ct^{-1}\log t .
\]

Finally, for \(g_2\), the local behavior of \(r_2-R_2\), Lemma~\ref{L:Tsharpq0},
and \eqref{rePhiIIIzj} imply
\[
|g_2(x,t,z)|
\le
C|z+q_0|\,|\log|z+q_0||\,\e^{-ct|z+q_0|^3}.
\]
Therefore,
\[
\|g_2\|_{L^\infty(\X_7^\eps)}
\le Ct^{-1/3}\log t,
\qquad
\|g_2\|_{L^1(\X_7^\eps)}
\le Ct^{-2/3}\log t .
\]
Combining the estimates for \(g_1,g_2,g_3\), we obtain \eqref{E:estVloc} on
\(\X_7^\eps\).

\medskip
\noindent{\bf The contour \(\X_9^\eps\).}
On the central segment \(\X_9^\eps\), a direct calculation gives
\[
\widetilde\V-\V^\X
=
\begin{pmatrix}
h_3(x,t,z)
& h_1(x,t,z)&0\\[2mm]
r_1(z)\frac{\widetilde T_{1+} }{\widetilde T_{2-}}(z)
\e^{\Theta_{21}^{\rm mod}(z)}
&
h_4(x,t,z)
&
-r_3(z)\frac{\widetilde T_{3+}}{\widetilde T_{2-}}(z)
\e^{\Theta_{23}^{\rm mod}(z)}
\\[2mm]
0&h_2(x,t,z)&\frac{\widetilde T_{3+}}{\widetilde T_{3-}}(z)
\end{pmatrix},
\]
where
\begin{align*}
h_1(x,t,z)
&=
-\frac{1}{\gamma(z)}r_1^*(z)
\frac{\widetilde T_{2+}}{\widetilde T_{1-}}(z)\e^{\Theta_{12}^{\rm mod}(z)}
-\ii s\,\e^{\Theta_{12}^{\rm mod}(z)},\\
h_2(x,t,z)
&=
-\frac{1}{\gamma(z)}r_3^*(z)
\frac{\widetilde T_{2+}}{\widetilde T_{3-}}(z)\e^{\Theta_{32}^{\rm mod}(z)}
+s\,\e^{\Theta_{32}^{\rm mod}(z)},\\
h_3(x,t,z)
&=\bigl(1-\frac{1}{\gamma(z)}|r_1(z)|^2 \bigr) \frac{\widetilde T_{1+}}{\widetilde T_{1-}}(z)-1,\\
h_4(x,t,z)
&=\bigl(1+\frac{1}{\gamma(z)}|r_3(z)|^2 \bigr) \frac{\widetilde T_{2+}}{\widetilde T_{2-}}(z)-1.
\end{align*}
By the jump relations for \(\widetilde \T\) on \(\X_9^\eps\), together with Lemma~\ref{L:Tsharpq0} and Remark~\ref{R:Tsharpboundary}, and using the boundedness of the model exponentials on \(\X_9^\eps\), we obtain
\[
 \sum_{j=1}^4 |h_j(x,t,z)|
\le
C|z+q_0|\,|\log|z+q_0|| .
\]
The remaining nonzero entries are estimated in the same way, using the fact
that the corresponding coefficients vanish at \(z=-q_0\).  Moreover, on
\(\X_9^\eps\) we have
\[
|z+q_0|\le C|\xi+q_0|\le Ct^{-2/3},
\qquad (x,t)\in\cP_+.
\]
Thus
\[
\|\widetilde\V-\V^\X\|_{L^\infty(\X_9^\eps)}
\le Ct^{-2/3}\log t.
\]
Since the length of \(\X_9^\eps\) is \(\mathcal O(t^{-2/3})\), it also follows that
\[
\|\widetilde\V-\V^\X\|_{L^1(\X_9^\eps)}
\le Ct^{-4/3}\log t.
\]
In particular, \eqref{E:estVloc} holds on \(\X_9^\eps\).  The estimates on the
remaining contours are analogous.  This proves \eqref{E:estVloc}.

It remains to prove \eqref{E:bzsy} and \eqref{E:bzsy1}.  If
\(z\in\partial\cD_\eps\), then
\[
|\beta(z)|
=
\left(\frac{3t}{4q_0}\right)^{1/3}|z+q_0|
\to\infty,
\qquad t\to+\infty.
\]
Hence the large-\(\beta\) expansion~\eqref{E:asyPhi} gives, uniformly for
\(z\in\partial\cD_\eps\) and \((x,t)\in\cP_+\),
\[
\N^\X(y,s,t;\beta(z))
=
\I+
\frac{\cE^\X_1+\M^P_1}
{\left(\frac{3}{4q_0}\right)^{1/3}t^{1/3}(z+q_0)}
+\mathcal O(t^{-2/3}).
\]
Since
$
\N^{loc}=\Y\N^\X \Y^{-1}
$
and \(\Y\), \(\cE^\X_1\), and \(\M^P_1\) are uniformly bounded, we immediately
obtain \eqref{E:bzsy}.  Taking the inverse of the last expansion yields
\[
(\N^{loc})^{-1}(x,t,z)
=
\I-
\frac{
\Y(\cE^\X_1+\M^P_1)\Y^{-1}
}{
\left(\frac{3}{4q_0}\right)^{1/3}t^{1/3}(z+q_0)
}
+\mathcal O(t^{-2/3}),
\]
uniformly for \(z\in\partial\cD_\eps\), which proves \eqref{E:bzsy1}.
\end{proof}
\subsection{Final transformation  and the small norm RH problem}
We define the final transformation to obtain a small-norm RH problem as follows:
\be \label{E:defE}
\cE(x,t,z)=\begin{cases}
\N^{(2)}(x,t,z), & z\in \C \setminus \cD_{\eps},\\
\N^{(2)}(x,t,z)(\N^{loc})^{-1}(x,t,z), &z \in \cD_{\eps}.
\end{cases}
\ee
We will show that $\cE(x,t,z)$ is close to $\I$ for large $t$ and $(x,t) \in \cP_+$.
Let
$$
\Sigma^{\xcE}=\Sigma^{(2)} \cup \partial \cD_{\eps},
$$  and define the matrix-valued function $\V^{\xcE}(x,t,z)$ for $z \in \Sigma^{\xcE}$ as follows:
\be \label{E:VEex}
\V^{\xcE}(x,t,z)=\begin{cases}
 \V^{(2)},  & z \in \Sigma^{\xcE} \setminus \overline{\cD_{\eps}},\\
(\N^{loc})^{-1}, & z\in \partial \cD_{\eps},\\
 \N^{loc}_-\V^{(7)} (\N^{loc}_+)^{-1},& z \in \X^{\eps}.
\end{cases}
\ee
The function $\cE(x,t,z)$ satisfies the following RH problem.
\begin{RHP}\label{rhp:E}
Find a $3 \times 3$ matrix-valued function $\cE(x,t,z)$ with the following properties:
\bi
\item $\cE(x,t,\cdot) : \mathbb{C}\setminus \Sigma^{\xcE}  \to \mathbb{C}^{3 \times 3}$ is analytic.
\item $\cE_+(x,t,z)=\cE_-(x,t,z) \V^{\xcE}(z), \qquad z \in \Sigma^{\xcE}.$
\item  $\cE(x,t,z)$ admits the following asymptotic behavior
$$  \cE(x,t,z)=\I+\mathcal{O}(\frac{1}{z}), \qquad z \to \infty.$$
\ei
\end{RHP}

\begin{lemma}\label{L:estWE}
Let $\w^{\xcE}=\V^{\xcE}-\I$. The following estimates hold uniformly for large $t$ and $(x,t) \in \cP_+$:
\begin{align}
&\|\w^{\xcE} \|_{  L^{\infty}(\Sigma^{\xcE} \setminus \overline{\cD_{\eps}} )} \leq C t^{-1/2}, \label{E:estwE1bc}\\
&\|\w^{\xcE} \|_{L^1 (\Sigma^{\xcE} \setminus \overline{\cD_{\eps}} )} \leq C t^{-1}, \label{E:estwE1bc1}\\
& \|\w^{\xcE} \|_{L^1 \cap L^{\infty}(\partial \cD_{\eps})} \leq C t^{-1/3},\label{E:estwE2} \\
& \|\w^{\xcE} \|_{L^1 (\X^{\epsilon } )} \leq C t^{-2/3}\log t,\label{E:estwE3}\\
& \|\w^{\xcE} \|_{L^{\infty} (\X^{\epsilon })} \leq C t^{-1/3}\log t. \label{E:estwE4}
\end{align}
\end{lemma}
\begin{proof}
It follows from~\eqref{E:gjV1} and~\eqref{E:VEex}  that~\eqref{E:estwE1bc} and~\eqref{E:estwE1bc1} hold.  From~\eqref{E:bzsy}, we obtain~\eqref{E:estwE2} directly.  For $z \in \X^{\eps}$, we have
$$
\w^{\xcE}=\N^{loc}_- \left( \V^{(7)}-\V^{loc}  \right)(\N^{loc}_+)^{-1}.
$$
Based on the boundedness of $\N^{loc}$ in $\cD_{\eps}$ and estimate~\eqref{E:estVloc},~\eqref{E:estwE4} follows immediately.
\end{proof}

The estimates in Lemma~\ref{L:estWE} show that
\begin{align*}
\begin{cases}
\|\w^{\xcE} \|_{L^1 (\Sigma^{\xcE})} \leq C t^{-1/3}\\
\|\w^{\xcE} \|_{L^{\infty} (\Sigma^{\xcE})} \leq C t^{-1/3}\log t
\end{cases} \quad (x,t) \in \cP_+, \ \ t>T_0.
\end{align*}
Thus by employing  the general inequality $\|f \|_{L^p} \leq \| f\|_{L^1}^{\frac{1}{p}} \|f \|_{L^{\infty}}^{\frac{p-1}{p}}$, we  immediately get
\begin{align}\label{E:wELp}
\|\w^{\xcE} \|_{L^p (\Sigma^{\xcE})} \leq C t^{-1/3}(\log t)^{(p-1)/p}, \quad \ \text{for} \ (x,t) \in \cP_+ \ \text{and large $t$}.
\end{align}
For the  contour $\Sigma^{\xcE}$ and a function $\h(z) \in L^2(\Sigma^{\xcE})$,  we define the Cauchy transform $\Ca(\h)(z)$ associated with $\Sigma^{\xcE}$ by
$$
\mathcal{C}(\h)(z) := \frac{1}{2\pi \ii}\int_{\Sigma^{\xcE}}\frac{\h(z')dz'}{z'-z}.
$$
Furthermore, we define the operator $\Ca_{\w^{\xcE}}$ by $\Ca_{\w^{\xcE}}(\h)=\Ca_-(\h \w^{\xcE})$. From the preceding analysis, we have known that $\|\w^{\xcE}\|_{L^2(\Sigma^{\xcE})} \to 0$ as $t \to \infty$. Consequently, there exists a $T_{*}>0$ such that the operator $I-\Ca_{\w^{\xcE}}$ is invertible whenever $t>T_{*}$ and $(x,t) \in \cP_+$. Therefore, we can  define a function $\bu^{\xcE}(x,t,z)$ for $z \in \Sigma^{\xcE}$ and $t>T_*$ by
\begin{align}\label{E:uE}
\bu^{\xcE}=\I+(I-\Ca_{\w^{\xcE}})^{-1}\Ca_{\w^{\xcE}}\I \ \in \I + L^2(\Sigma^{\xcE}).
\end{align}
We  need the estimate of $\|\bu^{\xcE} -\I\|_{L^2(\Sigma^{\xcE})}$. From \eqref{E:wELp} and \eqref{E:uE}, it follows that
\begin{equation}\label{E:yxjs11}
\begin{aligned}
\|\bu^{\xcE} - \I\|_{L^2(\Sigma^{\xcE})}&\leq \|(I-\Ca_{\w^{\xcE}})^{-1}\Ca_{\w^{\xcE}}\I \|_{L^2(\Sigma^{\xcE})} \leq \sum_{j=0}^{\infty}\| \Ca_{\w^{\xcE}}\|^j_{\mathcal{B}(L^2(\Sigma^{\xcE}))}\|\Ca_{\w^{\xcE}} \I \|_{L^2(\Sigma^{\xcE})}\\
&\leq  \frac{\|\Ca_- \|_{\mathcal{B}(L^2(\Sigma^{\xcE}))} \| \w^{\xcE}\|_{L^2(\Sigma^{\xcE})}}{1-\|\Ca_- \|_{\mathcal{B}(L^2(\Sigma^{\xcE}))} \|\w^{\xcE} \|_{L^{\infty}(\Sigma^{\xcE})}} \leq Ct^{-1/3}(\log t)^{1/2}, \qquad t>T_*.
\end{aligned}
\end{equation}
Then  $\cE(x,t,z)$ can be expressed as
\be \label{E:Ejfbs}
\cE(x,t,z)=\I+\frac{1}{2\pi \ii}\int_{\Sigma^{\xcE}}\frac{\bu^{\xcE}(z') \w^{\xcE} (z')dz'}{z'-z}, \qquad z \in \C \setminus \Sigma^{\xcE}.
\ee
From~\eqref{E:Ejfbs}, we know the following nontangential limit is well-defined:
\begin{align}\label{E:L}
\bL(x,t):=\lim^{\angle}_{z \to \infty}z (\cE(x,t,z)-\I)=-\frac{1}{2\pi \ii}\int_{\Sigma^{\xcE}}\bu^{\xcE} (x,t,\zeta)\w^{\xcE}(x,t,\zeta) \mathrm{d}\zeta.
\end{align}
\begin{lemma}\label{L:estofL}
As $t \to \infty$,
\begin{align}\label{E:estofL-1}
\bL(x,t)=-\frac{1}{2\pi \ii}\int_{\partial \mathcal{D}_{\eps}} \w^{\xcE}(x,t,\zeta) \mathrm{d}\zeta+\mathcal{O}(t^{-2/3}\log t).
\end{align}
\end{lemma}
\begin{proof}
The function $\bL(x,t)$ can be rewritten as
$$
\bL(x,t) = -\frac{1}{2\pi \ii}\int_{\partial \mathcal{D}_{\eps}} \w^{\xcE}(x,t,z) dz + \bL_1(x,t) + \bL_2(x,t),
$$
where
\begin{align*}
\bL_1(x,t) = -\frac{1}{2\pi \ii}\int_{\Sigma^{\xcE}\setminus \partial \mathcal{D}_{\eps}} \w^{\xcE} (x,t,z) \mathrm{d} z , \quad
 \bL_2(x,t) = -\frac{1}{2\pi \ii}\int_{\Sigma^{\xcE}} (\bu^{\xcE}(x,t,z)-\I) \w^{\xcE}(x,t,z) \mathrm{d} z.
\end{align*}
Then the lemma follows from Lemma~\ref{L:estWE} and Eq.~\eqref{E:yxjs11} and straightforward estimates.
\end{proof}
Recall that when $z \in \partial \cD_{\eps}$, we have
$$\w^{\xcE}(x,t,z)=(\N^{loc})^{-1}(x,t,z)-\I,$$
thus by~\eqref{E:bzsy} and Cauchy's formula,  we obtain
\begin{align*}
 -\frac{1}{2\pi \ii}\int_{\partial \mathcal{D}_{\eps}} \w^{\xcE}(x,t,z) dz=
\frac{\Y (\cE^{\X}_1  + \M^P_1) \Y^{-1}}{(\frac{3}{4q_0})^{1/3}t^{1/3}}+\mathcal{O}(t^{-2/3}).
\end{align*}
This implies
\be \label{E:Lasyt}
\bL(x,t) =\frac{\Y (\cE^{\X}_1  + \M^P_1) \Y^{-1}}{(\frac{3}{4q_0})^{1/3}t^{1/3}}+\mathcal{O}(t^{-2/3} \ln t).
\ee

Note that near $z=0$, we have  $\w^{\xcE}(x,t,z)=\V^{(1)}-\I$, and  $\V^{(1)}-\I$  is known to be $0$ at $z=0$ from~\eqref{E:estinear0}; thus  $\w^{\xcE}(x,t,z)$ is zero at the origin. This observation shows that $\cE(x,t,0)$  is well-defined and we have
\be \label{E:cExt0}
\cE(x,t,0)=\I+\frac{1}{2\pi \ii}\int_{\Sigma^{\xcE}}\frac{\bu^{\xcE}(x,t,z') \w^{\xcE} (x,t,z')}{z'} \mathrm{d} z'.
\ee
The following lemma will be crucial in the subsequent computations.
\begin{lemma}\label{L:glyl}
The following estimates hold uniformly for large $t$ and $(x,t) \in \cP_+$:
\be \label{E:cEjest}
\cE_j(x,t,0)=\ce_j+\frac{1}{2\pi \ii}\int_{\partial \mathcal{D}_{\eps}} \frac{\w^{\xcE}_j(x,t,\zeta) }{\zeta}\mathrm{d}\zeta+\mathcal{O}(t^{-2/3}\log t), \quad
j=2,3,
\ee
where $\cE_j(x,t,0)$, $\ce_j$ and   $\w^{\xcE}_j(x,t,\zeta)$
denote the j‑th columns of matrices $\cE(x,t,0)$, $\I$ and $\w^{\xcE}(x,t,\zeta)$, respectively.
\end{lemma}
\begin{proof}
By~\eqref{E:cExt0}, the function $\cE_j(x,t,0)$ can be rewritten as
$$
\cE_j(x,t,0)=\ce_j+\frac{1}{2\pi \ii}\int_{\partial \mathcal{D}_{\eps}} \frac{\w^{\xcE}_j(x,t,\zeta) }{\zeta}\mathrm{d}\zeta+ \boldsymbol{Q}_1(x,t)+\boldsymbol{Q}_2(x,t),
$$
where
$$
\boldsymbol{Q}_1(x,t)= \frac{1}{2\pi \ii}\int_{\Sigma^{\xcE}\setminus \partial \mathcal{D}_{\eps}} \frac{\w^{\xcE}_j (x,t,\zeta)}{\zeta} \mathrm{d} \zeta , \quad
 \boldsymbol{Q}_2(x,t)= \frac{1}{2\pi \ii}\int_{\Sigma^{\xcE}} (\bu^{\xcE}(x,t,\zeta)-\I) \frac{\w^{\xcE}_j(x,t,\zeta)}{\zeta} \mathrm{d} \zeta.
$$
Then the lemma follows from Remark~\ref{R:zsxy}, Lemma~\ref{L:estWE}, Eq.~\eqref{E:yxjs11} and straightforward estimates.
\end{proof}

\subsection{Proof of Theorem~\ref{Th:main} } \label{S:pthmain}
From Lemma~\ref{Th:rel}, we obtain
\be \label{E:wmxy}
\M(x,t,z)=\mathbf{\Delta}_{\infty} \left( \I + \frac{\A(x,t)}{z}\right) \N(x,t,z) \mathbf{\Delta}^{-1}(z),
\ee
where $\A(x,t)=\sig_1 \N_+^{-1}(x,t,0)= \sig_1 \N_-^{-1}(x,t,0)$.  Recalling reconstruction formula~\eqref{E:cggs}, we therefore need to study the large-$z$ behavior of $\M$. Suppose $\M$ and $ \N$ admit the following expansions as  $z \to \infty$  respectively:
\begin{align*}
\M(x,t,z)=\I+\frac{\M^{(1)}}{z}+\mathcal{O}(\frac{1}{z^2}), \quad
\N(x,t,z)=\I+\frac{\N^{(1)}}{z}+\mathcal{O}(\frac{1}{z^2}).
\end{align*}
In this section, for a matrix $\mathbf{C}$, we define the notation $(\mathbf{C})_{rc}$ as the two-dimensional column vector formed by the $(2,1)$ and $(3,1)$ entries of matrix $\mathbf{C}$. Then by~\eqref{E:wmxy}, we have
\be \label{E:newccg}
\begin{aligned}
(\M^{(1)})_{rc}&=\left(\mathbf{\Delta}_{\infty} \A \mathbf{\Delta}_{\infty}^{-1} \right)_{rc}+\left(\mathbf{\Delta}_{\infty} \N^{(1)} \mathbf{\Delta}_{\infty}^{-1} \right)_{rc}\\
&=\left(\sig_1 \mathbf{\Delta}(0) \N_+^{-1}(x,t,0) \mathbf{\Delta}_{\infty}^{-1} \right)_{rc}  + \left(\mathbf{\Delta}_{\infty} \N^{(1)} \mathbf{\Delta}_{\infty}^{-1} \right)_{rc},
\end{aligned}
\ee
where $\mathbf{\Delta}(0)=\diag\left(\delta_1(0), 1/\delta_1(0), 1  \right)$.
Recall the transformations introduced in Section~\ref{S:dza}. In the neighborhoods of  $z = 0$  and  $z = \infty$, we have
$$
\N(x,t,z)=\T_{\infty} \cE(x,t,z)  \T^{-1}(z) \G^{-1}(x,t,z).
$$
Let the region $D^c=\C \setminus \cup_{j=1}^9 \overline{D_j}$, then when $z \in D^c$, we have $\G(x,t,z)=\I(z)$. Thus, as  $z$  tends to infinity and to zero respectively from inside region  $D^c$, we obtain
\begin{align}
&\N_+(x,t,0)=\T_{\infty} \cE(x,t,0) \T^{-1}(x,t,0),  \label{E:lgds11}\\
&\N^{(1)}(x,t)=\T_{\infty} \bL(x,t) \T_{\infty}^{-1}.  \label{E:lgds}
\end{align}
For an invertible $3 \times 3$ matrix $\mathbf{C}=\left( \mathbf{C}_1,\mathbf{C}_2,\mathbf{C}_3 \right)$, we have
$$
\mathbf{C}^{-1}=\frac{1}{\det \mathbf{C}}
\bpm
 \left(\mathbf{C}_2 \times  \mathbf{C}_3 \right)^{\top}\\
\left( \mathbf{C}_3 \times  \mathbf{C}_1 \right)^{\top}\\
\left(  \mathbf{C}_1 \times    \mathbf{C}_2 \right)^{\top}
\epm,
$$
where $``\times "$ denotes the usual cross product.
Using this identity, we immediately obtain
\be \label{E:Nrc}
\begin{aligned}
\left(\sig_1 \mathbf{\Delta}(0) \N_+^{-1}(x,t,0) \mathbf{\Delta}_{\infty}^{-1} \right)_{rc}=\bpm
0\\
\ii q_0 \delta_1(0) \big( \N_{+,2}(x,t,0) \times \N_{+,3}(x,t,0) \big)_{11}
\epm,
\end{aligned}
\ee
where  $\N_{+,j}(x,t,0)$  denotes the  $j$-th column of  $\N_{+}(x,t,0)$, and $(\star)_{11}$ denotes the first entry of the corresponding vector.
Combining~\eqref{E:lgds},~\eqref{E:Nrc} with~\eqref{E:newccg}, we conclude that
\be \label{E:qgwb}
(\M^{(1)})_{rc}(x,t)=\bpm
0\\
\ii q_0 \delta_1(0) \big( \N_{+,2}(x,t,0) \times \N_{+,3}(x,t,0) \big)_{11}
\epm+
\bpm
\frac{1}{\delta_1(0) (\delta(0))^2 \rP_1(0)}  \bL_{21}(x,t)\\
\frac{\delta_1(0)  \rP_1(0)}{\delta(0)}  \bL_{31}(x,t)
\epm.
\ee
Firstly, let us consider the first term on the right-hand side of the above equation.
From~\eqref{E:lgds11}, we know that we need to estimate $\cE_j(x,t,0)$. Recall that when $z \in \partial \cD_{\eps}$, we have
$$\w^{\xcE}(x,t,z)=(\N^{loc})^{-1}(x,t,z)-\I.$$
Thus by~\eqref{E:bzsy1} and Cauchy's formula,  we obtain
\begin{align}
 \frac{1}{2\pi \ii}\int_{\partial \mathcal{D}_{\eps}} \frac{[\w^{\xcE}]_2(x,t,\zeta)}{\zeta} d \zeta=
\frac{1}{q_0} \frac{\left[ \Y \M^P_1 \Y^{-1} \right]_2}{(\frac{3t}{4q_0})^{1/3}}+\frac{1}{q_0} \frac{\left[ \Y \cE^{\X}_1 \Y^{-1} \right]_2}{(\frac{3t}{4q_0})^{1/3}}+\mathcal{O}(t^{-2/3}), \label{E:gdgs1}\\
 \frac{1}{2\pi \ii}\int_{\partial \mathcal{D}_{\eps}} \frac{[\w^{\xcE}]_3(x,t,\zeta)}{\zeta} d \zeta=
\frac{1}{q_0} \frac{\left[ \Y \M^P_1 \Y^{-1} \right]_3}{(\frac{3t}{4q_0})^{1/3}}+\frac{1}{q_0} \frac{\left[ \Y \cE^{\X}_1 \Y^{-1} \right]_3}{(\frac{3t}{4q_0})^{1/3}}+\mathcal{O}(t^{-2/3}). \label{E:gdgs2}
\end{align}
In the expressions above, $[\mathbf{A}]_j$ denotes the $j$-th column of matrix $\mathbf{A}$.
Then, from the expression for $\M^P_1$, Eq.~\eqref{E:estcEX1}, and Lemma~\ref{L:glyl}, a straightforward calculation yields
\begin{align}
&[\cE]_2(x,t,0)=\bpm 0\\ 1\\ 0 \epm - \frac{\e^{\frac{\pi }{4} \ii } s \e^{U_{HM}(y)} \e^{-\theta_2(x,t,-q_0)} }{2 q_0 \sqrt{\pi}   t^{1/2}}
\bpm
1\\
0\\
\ii
\epm +\mathcal{O}(t^{-2/3} \log t), \label{E:cE2estin1}\\
&[\cE]_3(x,t,0)=\bpm 0\\ 0\\ 1 \epm+
\frac{1}{q_0 (\frac{3 t}{4 q_0})^{1/3}} \bpm
-\frac{u_{HM}(y)}{2}\\
0\\
-\frac{\ii}{2} \int_{\infty}^y u^2_{HM}(y')\mathrm{d}y'
\epm +\mathcal{O}(t^{-2/3} \log t).\label{E:cE2estin2}
\end{align}
Thus we have
\be  \label{E:cE2estin3}
\begin{aligned}
\ii q_0 \big( \N_{+,2}(x,t,0) & \times \N_{+,3}(x,t,0) \big)_{11}\\
&=
\ii q_0 \delta_1(0) \left(  \left[  \T_{\infty} \cE(x,t,0) \T^{-1}(x,t,0) \right]_2  \times  \left[  \T_{\infty} \cE(x,t,0) \T^{-1}(x,t,0) \right]_3  \right)_{11}\\
&= \frac{\delta_1(0) \rP_1(0)}{\delta(0)} \left[
\ii q_0-
\frac{1}{2  (\frac{3 t}{4 q_0})^{1/3}} \int_{y}^{\infty} u^2_{HM}(y') \mathrm{d}y'
\right] +\mathcal{O}(t^{-2/3} \log t).
\end{aligned}
\ee
We have completed the estimation of the first term on the right‑hand side of~\eqref{E:qgwb}. Next, we proceed to analyze its second term.
Note that the first column of $\cE^{\X}_1$ satisfies the estimate
$$
[\cE^{\X}_1]_1(x,t)=\mathcal{O}(t^{-1/3}), \quad t \to \infty,
$$
and then, combining the expression for $\M^P_1$ and~\eqref{E:Lasyt}, we obtain
\begin{align}\label{E:cE2estin4}
\begin{cases}
\bL_{21}(x,t) =\mathcal{O}(t^{-2/3} \log t), \\
\bL_{31}(x,t)=-\frac{u_{HM}(y)}{2 (\frac{3 t}{4 q_0})^{1/3}} +\mathcal{O}(t^{-2/3} \log t),
\end{cases}\quad  t \to \infty.
\end{align}
Thus, from~\eqref{E:qgwb},~\eqref{E:cE2estin3}, and~\eqref{E:cE2estin4}, it follows that
\be \label{E:qgwbgma}
(\M^{(1)})_{rc}(x,t)=
\bpm
0\\
\ii q_0 \frac{\delta_1(0) \rP_1(0)}{\delta(0)}-\frac{\delta_1(0) \rP_1(0)}{2\delta(0) ( \frac{3 t}{4 q_0})^{1/3}} \left(
u_{HM}(y)+\int_{y}^{\infty} u^2_{HM}(y') \mathrm{d}y'
\right)
\epm+\mathcal{O}(t^{-2/3} \log t).
\ee
Substituting the asymptotic estimate~\eqref{E:qgwbgma} of $(\M^{(1)})_{rc}(x,t)$ into reconstruction formula~\eqref{E:cggs}, a straightforward calculation yields asymptotic formula~\eqref{E:asygs}.

\section{Concluding remarks}

In this work, we studied the Painlev\'e-type transition asymptotics for the
defocusing Manakov system under parallel NZBCs.  By applying the Deift--Zhou
steepest descent method to the associated \(3\times3\) matrix RH problem, we
derived a uniform asymptotic formula in the transition region.  The result
shows that, apart from the modulated background, the first correction is of
order \(t^{-1/3}\) and is governed by the Hastings--McLeod solution of the
Painlev\'e II equation.

A notable feature of the analysis is the appearance of a coupled
Painlev\'e--error-function local model near the branch point.  A similar model
was recently introduced by Charlier and Lenells in their steepest descent
analysis of the ``bad'' Boussinesq equation~\cite{CL2024main}.  The present
work shows that this type of coupled local construction also arises naturally
in the finite-density vector NLS setting.  We expect that such a model is not
limited to these two examples, but may also appear in the Painlev\'e-type
transition analysis of other coupled integrable systems with NZBCs, such as
coupled mKdV equations~\cite{IH1997,GZD2014}, coupled Hirota
equations~\cite{BMP2001}, and coupled Gerdjikov--Ivanov
equations~\cite{MZ-2023}.

Another natural direction is to extend the present analysis to \(N\)-component
NLS system~\cite{LSG2025}.  In that case, one may expect a higher-dimensional
version of the coupled local model to emerge.  It would also be interesting to
develop the corresponding asymptotic theory under Sobolev-type assumptions on
the initial data.  For this purpose, a \(\bar\partial\)-generalization of the
nonlinear steepest descent method~\cite{CuJe2016,DM2008,BJM2018} should be a
useful tool, although several technical difficulties remain.  Finally, it
would be worthwhile to investigate whether other transition mechanisms occur
between the Painlev\'e region and the soliton region.  We leave these questions
for future work.

\appendix
\section{Proof of Lemma~\ref{Th:rel}}\label{App:AAAA1}

It follows from~\eqref{E:jumpM1ex} and~\eqref{E:JumpN} that
\be \label{E:asess}
\N_+(x,t,0)=\N_-(x,t,0)
\bpm
1&0&0\\
\star&1&0\\
0&0&1
\epm,
\ee
where \(\star\) denotes an unspecified entry. Since the residue matrices are nilpotent and act only on one column, the possible singularities of \(\det\N\) at the discrete spectral points are removable. Moreover, \(\det\bV^{(1)}=1\) on \(\R\), and the normalization at infinity gives
$
\det\N\equiv 1
$.
In particular, \(\N_\pm^{-1}(x,t,0)\) are well defined. From~\eqref{E:asess}, one obtains
\[
\sig_1\N_+^{-1}(x,t,0)=\sig_1\N_-^{-1}(x,t,0).
\]

We now prove~\eqref{E:gfbh}. Let \(\N(x,t,z)\) be the unique solution of RH problem~\ref{RHP:re}, and define
\[
\widetilde{\M}(x,t,z)
=
\left(\I+\frac{\A(x,t)}{z}\right)\N(x,t,z),
\qquad
\A(x,t)=\sig_1\N_+^{-1}(x,t,0).
\]
It is immediate from the definition that \(\widetilde{\M}\) satisfies the jump condition~\eqref{E:JumpM1}, the normalization~\eqref{E:M1astj}, and the same residue conditions as \(\M^{(1)}\). It remains to verify the branch-point conditions~\eqref{E:M1gcc} and the two symmetries~\eqref{E:M1RHP11}.

\noindent{\bf Step 1. The \(z\mapsto\hat z\) symmetry.}
By Remark~\ref{R:syV1},
\[
\bV^{(1)}(x,t,z)
=
\bPi^{-1}(z)\left(\bV^{(1)}(x,t,\hat z)\right)^{-1}\bPi(z).
\]
Hence \(\N(x,t,z)\) and \(\N(x,t,\hat z)\bPi(z)\) have the same jump across \(\R\). Define
\[
\mathbf F(x,t,z)
=
\N(x,t,\hat z)\bPi(z)\N^{-1}(x,t,z).
\]
Then \(\mathbf F\) has no jump across \(\R\). The possible singularities of \(\mathbf F\) at the discrete spectral points are removable. Thus the only possible singularities of \(\mathbf F\) are at \(z=0\) and \(z=\infty\).

Write
\[
\bPi(z)=\sig_2+\frac{1}{z}\sig_1,
\qquad
\sig_2=
\bpm
0&0&0\\
0&1&0\\
0&0&0
\epm.
\]
As \(z\to\infty\), we have \(\hat z\to0\), and hence
\[
\mathbf F(x,t,z)=\mathcal O(1),\qquad
\mathbf F(x,t,\infty)=\N_+(x,t,0)\sig_2=\N_-(x,t,0)\sig_2.
\]
As \(z\to0\), we have \(\hat z\to\infty\), and therefore
\[
\mathbf F(x,t,z)=\frac{1}{z}\sig_1\N_+^{-1}(x,t,0)+\mathcal O(1).
\]
Liouville's theorem then gives
\be \label{E:yzbd}
\mathbf F(x,t,z)
=
\frac{1}{z}\sig_1\N_+^{-1}(x,t,0)+\N_+(x,t,0)\sig_2
=
\frac{1}{z}\sig_1\N_-^{-1}(x,t,0)+\N_-(x,t,0)\sig_2.
\ee
Set
$
A:=\N_+(x,t,0)
$
and
$B:=\N_+^{-1}(x,t,0)
$.
Then~\eqref{E:yzbd} gives
\[
\N(x,t,\hat z)\bPi(z)\N^{-1}(x,t,z)
=
\frac{1}{z}\sig_1B+A\sig_2.
\]
Multiplying this identity on the left by
$
\I+\frac{z}{q_0^2}\sig_1B
$
and on the right by \(\N(x,t,z)\), we obtain
\[
\begin{aligned}
\widetilde{\M}(x,t,\hat z)\bPi(z)
=
\left[
A\sig_2
+
\frac{1}{q_0^2}\sig_1B\sig_1B
+
\frac{1}{z}\sig_1B
\right]\N(x,t,z).
\end{aligned}
\]
Consequently, to prove
$
\widetilde{\M}(x,t,z)=\widetilde{\M}(x,t,\hat z)\bPi(z)
$,
it suffices to show that
\be \label{E:gj}
A\sig_2+\frac{1}{q_0^2}\sig_1B\sig_1B=\I.
\ee

We first derive an algebraic constraint. Let \(\bX_1\) and \(\bX_2\) be defined by
\[
\N^{-1}(x,t,z)=\I-\frac{\bX_1}{z}+\mathcal O(z^{-2}),
\qquad z\to\infty,
\]
and
\[
\N(x,t,\hat z)=A+\frac{\bX_2}{z}+\mathcal O(z^{-2}),
\qquad \C_-\ni z\to\infty.
\]
The coefficient of \(z^{-1}\) in the expansion of \(\mathbf F\) as \(\C_-\ni z\to\infty\) is
\[
A\sig_1+\bX_2\sig_2-A\sig_2\bX_1.
\]
Comparing this with~\eqref{E:yzbd}, we get
\[
A\sig_1+\bX_2\sig_2-A\sig_2\bX_1=\sig_1B.
\]
Multiplying this identity from the left by \(\sig_1B\), we obtain
\[
\sig_1^2+\sig_1B\bX_2\sig_2=\sig_1B\sig_1B.
\]
Then the above expression can be rewritten as
\be \label{E:msjy}
\sig_1B\sig_1B
=
q_0^2
\bpm
1&\star&0\\
0&0&0\\
0&\star&1
\epm,
\ee
where $\star$ denotes an unspecified entry.
It remains to determine the two unspecified entries. Write
$
A=(A_{ij})_{1\leq i,j\leq3}
$, and
$
B=(B_{ij})_{1\leq i,j\leq3}
$.
We first determine \(A_{22}\). Expanding \(\mathbf F\) as \(z\to0\), write
\[
\N(x,t,\hat z)=\I+\frac{z}{q_0^2}\bX_3+\mathcal O(z^2),
\qquad
\N^{-1}(x,t,z)=B+z\bX_4+\mathcal O(z^2).
\]
Using \(\bPi(z)=\sig_2+z^{-1}\sig_1\), comparison of the constant terms in~\eqref{E:yzbd} gives
\[
A\sig_2=\sig_2B+\sig_1\bX_4+\frac{1}{q_0^2}\bX_3\sig_1B.
\]
Multiplying this identity from the right by \(A\), we obtain
\[
A\sig_2A
=
\sig_2+\sig_1\bX_4A+\frac{1}{q_0^2}\bX_3\sig_1.
\]
Taking the \((2,2)\)-entry and using that the second row and the second column of \(\sig_1\) are zero, we get
$
A_{22}^2=1
$.

We now fix the sign of $A_{22}$. From the small-norm analysis of the auxiliary RH problem for \(\N\), in particular from~\eqref{E:lgds11} and~\eqref{E:cE2estin1}, one has, uniformly for \((x,t)\in\cP_+\),
\[
A_{22}=(\N_+(x,t,0))_{22}=1+o(1),
\qquad t\to\infty.
\]
Indeed, \eqref{E:lgds11} gives
\[
\N_+(x,t,0)=\T_\infty  \cE(x,t,0) \T^{-1}(0),
\]
and the \((2,2)\)-entries of \(\T_\infty\) and \(\T(0)\) coincide, so that
$
(\N_+(x,t,0))_{22}=\cE_{22}(x,t,0)
$.
Equation~\eqref{E:cE2estin1}, applied to the second column, yields \( \cE_{22}(x,t,0)=1+o(1)\). Hence \(A_{22}=1\) for all sufficiently large \(t\) in the region \(\cP_+\). Since \(A_{22}\) depends continuously on \((x,t)\) and takes values only in the discrete set \(\{1,-1\}\), we conclude that
$
A_{22}=1
$.

We next prove that
$
B_{13}=B_{31}
$.
Taking the \((1,3)\)- and \((3,1)\)-entries in~\eqref{E:msjy}, respectively, gives
\[
q_0^2B_{33}(B_{13}-B_{31})=0,
\qquad
q_0^2B_{11}(B_{31}-B_{13})=0.
\]
Thus \(B_{13}=B_{31}\) unless \(B_{11}=B_{33}=0\). Suppose \(B_{11}=B_{33}=0\). Then the \((1,1)\)- and \((3,3)\)-entries of~\eqref{E:msjy} yield
\[
B_{31}^2=-1,
\qquad
B_{13}^2=-1.
\]
On the other hand, since \(A=B^{-1}\) and \(\det B=1\),  we get
$
B_{13}B_{31}=-1
$.
Together with \(B_{31}^2=-1\), this implies \(B_{13}=B_{31}\). Hence \(B_{13}=B_{31}\) in all cases.

Now the \((1,2)\)- and \((3,2)\)-entries of \(\sig_1B\sig_1B\) are
\[
(\sig_1B\sig_1B)_{12}
=
q_0^2(B_{33}B_{12}-B_{13}B_{32}),
\qquad
(\sig_1B\sig_1B)_{32}
=
q_0^2(B_{11}B_{32}-B_{12}B_{31}).
\]
Using \(B_{13}=B_{31}\) and the cofactor formula for \(A=B^{-1}\), we obtain
\[
(\sig_1B\sig_1B)_{12}=-q_0^2A_{12},
\qquad
(\sig_1B\sig_1B)_{32}=-q_0^2A_{32}.
\]
Combining these identities with~\eqref{E:msjy} and using \(A_{22}=1\), we obtain
\[
A\sig_2+\frac{1}{q_0^2}\sig_1B\sig_1B
=
\bpm
0&A_{12}&0\\
0&1&0\\
0&A_{32}&0
\epm
+
\bpm
1&-A_{12}&0\\
0&0&0\\
0&-A_{32}&1
\epm
=\I.
\]
This proves~\eqref{E:gj}, and hence the first symmetry follows.

\noindent{\bf Step 2. The growth condition at the branch points.}
We verify~\eqref{E:M1gcc} at \(q_0\); the argument at \(-q_0\) is analogous. By the endpoint behavior of the reflection coefficients,
\[
\lim_{z\to q_0}\tilde r_2(z)=\ii,
\qquad
\lim_{z\to q_0}\tilde r_1(z)=\lim_{z\to q_0}\tilde r_3(z)=0.
\]
Hence
\[
\widetilde{\M}_+(x,t,q_0)
=
\widetilde{\M}_-(x,t,q_0)
\bpm
1&\star&\ii\\
0&1&0\\
\ii&\star&0
\epm.
\]
Taking the third column gives
\[
\widetilde{\M}_{+3}(x,t,q_0)=\ii\widetilde{\M}_{-1}(x,t,q_0).
\]
On the other hand, the symmetry proved in Step~1 gives, by taking the third column at \(q_0\),
\[
\widetilde{\M}_{+3}(x,t,q_0)=-\ii\widetilde{\M}_{-1}(x,t,q_0).
\]
Therefore
\[
\widetilde{\M}_{+3}(x,t,q_0)=\widetilde{\M}_{-1}(x,t,q_0)=0.
\]
Since the relevant columns are analytic in the corresponding half-neighborhoods of \(q_0\) and admit continuous extensions to the branch point, this vanishing implies the required first-order behavior in~\eqref{E:M1gcc}.

\noindent{\bf Step 3. The \(z\mapsto z^*\) symmetry.}
Let
$
\mathbf D(z)=-\gamma(z)\mathbf\Gamma^{-1}(z)
$.
By Remark~\ref{R:syV1},
\[
\mathbf D(z)\bV^{(1)}(x,t,z)\mathbf D^{-1}(z)
=
\left(\bV^{(1)}(x,t,z^*)\right)^\dagger.
\]
Thus \(\widetilde{\M}(x,t,z)\mathbf D^{-1}(z)\) and
\([\widetilde{\M}^{\dagger}(x,t,z^*)]^{-1}\) satisfy the same jump on \(\R\). Define
\[
\mathbf K(x,t,z)
=
\widetilde{\M}(x,t,z)\mathbf D^{-1}(z)
\widetilde{\M}^{\dagger}(x,t,z^*).
\]
Then \(\mathbf K\) has no jump on \(\R\), and it is analytic in
$
\C\setminus\bigl(\{0,\pm q_0\}\cup\mathcal Z\bigr)
$.
We now show that all possible singularities are removable.

First, as \(z\to0\),
\[
\mathbf D^{-1}(z)=
\diag \left(\frac{1}{\gamma(z)},-1, -\frac{1}{\gamma(z)}  \right)
= \diag \left(\mathcal O(z^2),-1, \mathcal O(z^2)  \right).
\]
Moreover,
\[
\widetilde{\M}(x,t,z)\times  \diag \left(z,1,z \right)
=\mathcal O(1),
\qquad
 \diag \left(z,1,z \right) \times
\widetilde{\M}^{\dagger}(x,t,z^*)
=\mathcal O(1).
\]
Hence \(\mathbf K(x,t,z)=\mathcal O(1)\) as \(z\to0\).
Next, let \(\widetilde{\M}_j\) denote the \(j\)-th column of \(\widetilde{\M}\). Then
\[
\begin{aligned}
\mathbf K(x,t,z)
&=
\bpm
\frac{1}{\gamma(z)}\widetilde{\M}_1(x,t,z)&
-\widetilde{\M}_2(x,t,z)&
\widetilde{\M}_3(x,t,z)
\epm
\bpm
\widetilde{\M}_1^\dagger(x,t,z^*)\\
\widetilde{\M}_2^\dagger(x,t,z^*)\\
-\frac{1}{\gamma(z)}\widetilde{\M}_3^\dagger(x,t,z^*)
\epm\\
&=
\bpm
\widetilde{\M}_1(x,t,z)&
-\widetilde{\M}_2(x,t,z)&
-\frac{1}{\gamma(z)}\widetilde{\M}_3(x,t,z)
\epm
\bpm
\frac{1}{\gamma(z)}\widetilde{\M}_1^\dagger(x,t,z^*)\\
\widetilde{\M}_2^\dagger(x,t,z^*)\\
\widetilde{\M}_3^\dagger(x,t,z^*)
\epm.
\end{aligned}
\]
Together with the branch-point behavior proved in Step~2, these two representations show that \(\mathbf K\) is bounded near \(\pm q_0\).

It remains to consider the discrete spectrum. Since the arguments for the other points are analogous, it suffices to verify the condition at \(z=z_j\).
  From the residue conditions, as \(z\to z_j\),
\[
\widetilde{\M}(x,t,z)
=
\frac{\tilde\kappa_j\e^{\theta_{21}(x,t,z_j)}}{z-z_j}
\bpm
a_1&0&0\\
a_2&0&0\\
a_3&0&0
\epm
+
\bpm
\star&a_1&\star\\
\star&a_2&\star\\
\star&a_3&\star
\epm
+\mathcal O(z-z_j),
\]
and
\[
\widetilde{\M}^{\dagger}(x,t,z^*)
=
\frac{\frac{\tilde\kappa_j}{\gamma(z_j)}\e^{\theta_{21}(x,t,z_j)}}{z-z_j}
\bpm
0&0&0\\
A_1&A_2&A_3\\
0&0&0
\epm
+
\bpm
A_1&A_2&A_3\\
\star&\star&\star\\
\star&\star&\star
\epm
+\mathcal O(z-z_j).
\]
Therefore,
\[
\widetilde{\M}(x,t,z)\mathbf D^{-1}(z)
=
\frac{\frac{\tilde\kappa_j}{\gamma(z_j)}\e^{\theta_{21}(x,t,z_j)}}{z-z_j}
\bpm
a_1&0&0\\
a_2&0&0\\
a_3&0&0
\epm
+
\bpm
\star&-a_1&\star\\
\star&-a_2&\star\\
\star&-a_3&\star
\epm
+\mathcal O(z-z_j).
\]
Multiplying the last two expansions, a straightforward calculation shows that
\(\mathbf K\) contains no negative powers as \(z\to z_j\). Hence \(z_j\) is a removable
singularity of \(\mathbf K\).

Thus \(\mathbf K\) is entire. Since
\[
\mathbf K(x,t,z)=\bJ+\mathcal O(z^{-1}),
\qquad z\to\infty,
\]
Liouville's theorem gives
\[
\mathbf K(x,t,z)
=
\widetilde{\M}(x,t,z)\mathbf D^{-1}(z)
\widetilde{\M}^{\dagger}(x,t,z^*)
=
\bJ.
\]
Equivalently,
\[
(\widetilde{\M}^{-1})^\top(x,t,z)
=
-\frac{1}{\gamma(z)}
\bJ\widetilde{\M}^{*}(x,t,z^*)\mathbf\Gamma(z),
\]
which is the second symmetry in~\eqref{E:M1RHP11}.

We have shown that \(\widetilde{\M}\) satisfies RH problem~\ref{RHP:M1}. By uniqueness,
$
\widetilde{\M}=\M^{(1)}
$.
This proves~\eqref{E:gfbh} and completes the proof of Lemma~\ref{Th:rel}.

\section{Painlev\'e II model problem}\label{App:AAA1}
\begin{figure}
\centering
\begin{tikzpicture}[
    scale=0.8,
    line cap=round,
    line join=round,
    midarrow/.style={
        line width=0.7mm,
        postaction={decorate},
        decoration={markings, mark=at position 0.55 with {\arrow{latex}}}
    },
    midarrowdashed/.style={
        dashed,
        line width=0.5mm,
        postaction={decorate},
        decoration={markings, mark=at position 0.55 with {\arrow{latex}}}
    }
]

\draw[midarrow] (3.46,-3) -- (0,-1);
\draw[midarrow] (-3.46,-3) -- (0,-1);
\draw[midarrow] (0,1) -- (3.46,3);
\draw[midarrow] (0,1) -- (-3.46,3);

\draw[midarrowdashed] (-4,0) -- (4,0);

\node at (1.7,2.35) {\small $1$};
\node at (-1.7,2.35) {\small $2$};
\node at (-1.7,-2.4) {\small $3$};
\node at (1.7,-2.4) {\small $4$};
\node at (0,0.7) {\small $\ii$};
\node at (0,-0.7) {\small $-\ii$};

\end{tikzpicture}
\caption{The jump contour $P=\cup_{j=1}^4 \X_j$ for the RH problem for $\M^P$.}
\label{Pfig}
\end{figure}

Consider the subcontour $P=\cup_{j=1}^4 \X_j$ of $\X$, see Figure~\ref{Pfig}. Define the jump matrix $\V^P(y,\beta)$ by
\begin{align}\nonumber
& \V_1^P =
\begin{pmatrix}
1 & 0 & 0 \\ 0 & 1 & 0 \\
\ii \e^{2\ii (y\beta + \frac{4\beta^3}{3})}  & 0 & 1 \end{pmatrix},
&&
 \V_2^P =
\begin{pmatrix}
1 & 0 & 0 \\ 0 & 1 & 0 \\ -\ii \e^{2\ii (y\beta + \frac{4\beta^3}{3})}  & 0 & 1 \end{pmatrix},
	\\\label{vPdef}
& \V_3^P =
\begin{pmatrix}
1 & 0 & \ii \e^{-2\ii (y\beta + \frac{4\beta^3}{3})}  \\ 0 & 1 & 0 \\ 0 & 0 & 1 \end{pmatrix},
 && \V_4^P =
\begin{pmatrix}
1 & 0 & -\ii \e^{-2\ii (y\beta + \frac{4\beta^3}{3})} \\ 0 & 1 & 0 \\
0 & 0 & 1 \end{pmatrix},
\end{align}
where $\V_j^P$ denotes the restriction of $\V^P$ to $\X_j$.

\begin{RHP}[RH problem for $\M^P$]\label{RHmP}
Find a $3 \times 3$-matrix valued function $\M^P(y, \beta)$ with the following properties:
\bi
\item $\M^P(y,\cdot) : \C \setminus \cup_{j=1}^4 \X_j \to \mathbb{C}^{3 \times 3}$ is analytic.

\item $\M^P(y,\cdot)$ has continuous boundary values on $\cup_{j=1}^4 \X_j \setminus \{\pm \ii\}$ satisfying the jump relation
\begin{align*}
  \M^P_+(y,\beta) = \M^P_-(y,\beta) \V^P(y,\beta), \qquad \beta \in \cup_{j=1}^4 \X_j \setminus \{\pm \ii\}.
\end{align*}

\item $\M^P(y,\beta) = \I + O(\beta^{-1})$ as $\beta \to \infty$ and $\M^P(y,\beta) = \mathcal{O}(1)$ as $\beta \to \pm \ii$.
\ei
\end{RHP}

\begin{lemma}\label{mPlemma}
For each $y \in \R$, RH problem \ref{RHmP} has a unique solution $\M^P(y,\beta)$ with the following properties:
\bi
\item
There are smooth functions $\{\M_j^P(y)\}_{j=1}^\infty$ of $y \in \R$ such that, for each integer $N \geq 0$,
\begin{align}\label{mPasymptotics}
\M^P(y, \beta) = \I + \sum_{j=1}^N \frac{\M_j^P(y)}{\beta^j} + O(\beta^{-N-1}), \qquad \beta \to \infty,
\end{align}
uniformly for $y$ in compact subsets of $\R$ and for $\arg \beta  \in [0,2\pi]$.

\item
The leading term is given by
\begin{align}\label{m1Pexpression}
\M_1^{P}(y) = \begin{pmatrix} \frac{\ii}{2} \int_{\infty}^y u_{HM}(y')^2 dy'
& 0 & -\frac{u_{HM}(y)}{2} \\
0 & 0 & 0 \\
-\frac{u_{HM}(y)}{2} & 0 & -\frac{\ii}{2}\int_{\infty}^y u_{HM}(y')^2 dy'  \end{pmatrix},
\end{align}
where $u_{HM}$ is the Hasting--McLeod solution of Painlev\'e II, that is, the classical solution of the equation satisfying the  boundary conditions~\eqref{E:uHMbjtj}.

\item At $\beta= 0$, we have
\begin{align}\label{mPat0expression}
\M^P(y,0) = &\; \begin{pmatrix} \cosh{U_{HM}(y)} & 0 & -\ii \sinh{U_{HM}(y)} \\
0 & 1 & 0 \\
\ii \sinh{U_{HM}(y)}  & 0 & \cosh{U_{HM}(y)}
\end{pmatrix},
\end{align}
where $U_{HM}(y) := \int_{\infty}^y u_{HM}(y') dy'$.
\ei
\end{lemma}
\begin{proof}
Since the jump contour can be arranged to intersect at the origin through a simple contour deformation, RH problem~\ref{RHmP}  is reduced to  the RH problem associated with the Hastings-McLeod solution of the Painlev\'e II equation. Thus all the assertions above are well-known results. For details of the proofs, we refer the reader to~\cite[Proposition 5.2 \&  Theorem 11.7]{FIKN2006}.
\end{proof}

\begin{figure}
\centering
\begin{tikzpicture}[
    scale=0.8,
    line cap=round,
    line join=round,
    thickline/.style={line width=0.7mm},
    thickdashed/.style={dashed, line width=0.9mm},
    midarrow/.style={
        thickline,
        postaction={decorate},
        decoration={markings, mark=at position 0.55 with {\arrow{latex}}}
    },
    midarrowdashed/.style={
        thickdashed,
        postaction={decorate},
        decoration={markings, mark=at position 0.5 with {\arrow{latex}}}
    },
    bluethick/.style={line width=0.7mm, blue},
    bluemidarrow/.style={
        line width=0.7mm,
        blue,
        postaction={decorate},
        decoration={markings, mark=at position 0.55 with {\arrow{latex}}}
    }
]

\draw[midarrow] (-4,-1.74) -- (-2.5,-0.87);
\draw[thickline] (-2.5,-0.87) -- (-1,0);
\draw[midarrow] (-1,0) -- (0,0);
\draw[thickline] (0,0) -- (1,0);
\draw[midarrow] (1,0) -- (2.5,0.87);
\draw[midarrow] (2.5,0.87) -- (4,1.74);


\draw[bluethick, rotate=45] (-4,0) -- (2,0);
\draw[bluemidarrow, rotate=45] (2,0) -- (4,0);
\draw[bluemidarrow, rotate=45] (-4,0) -- (-2,0);

\node[right] at (-2.9,-1.6) {$\Omega$};
\node[left]  at ( 2.9, 1.6) {$\Omega$};

\end{tikzpicture}
\caption{\small
The black solid contour is $E$, and the blue solid contour is 
$\widetilde E=\mathbb R \e^{\pi\ii/4}$. The region between them is denoted by $\Omega$.
}
\label{Efig}
\end{figure}

\section{Error function model problem}\label{App:AAA2}
Let the countor $E$ is shown in Fig.~\ref{Efig}. Consider the following RH problem.
\begin{RHP}[RH problem for $\M^E$]\label{RHME}
Find a $3 \times 3$-matrix valued function $\M^E(s,\alpha)$ with the following properties:
\bi
\item $\M^E(s,\cdot) : \C \setminus E \to \mathbb{C}^{3 \times 3}$ is analytic.

\item $\M^E(s,\cdot)$ has continuous boundary values on $E$ satisfying the jump relation
\begin{align*}
  \M^E_+(s,\alpha) = \M^E_-(s, \alpha) \V^E(s, \alpha), \qquad \alpha \in  E,
\end{align*}
where
$$
\V^E=\bpm
1& \ii s \e^{\ii \alpha^2} &0\\
0&1&0\\
0&-s \e^{\ii \alpha^2}&1
\epm.
$$

\item $\M^E(s, \alpha) = \I + O(\alpha^{-1})$ as $\alpha  \to \infty$.

\ei
\end{RHP}
The following lemma contains the results we need about $\M^E(s,\alpha)$.

\begin{lemma}\label{MElemma}
RH problem \ref{RHME} has a unique solution $\M^E(s,\alpha)$ . For each integer $N \geq 0$,
\begin{align}\label{MEasymptotics}
\M^E(s, \alpha) = \I + \sum_{j=0}^N \frac{\M_{2j+1}^E(s)}{\alpha^{2j+1}} + O(\alpha^{-2N-3}), \qquad \alpha \to \infty,
\end{align}
uniformly for $\mathrm{arg} \alpha \in [0,2\pi]$, where the leading coefficient is given by
\begin{align}\label{M1Eexpression}
\M_1^E(s) =- \frac{s e^{\frac{\pi i}{4}}}{2 \sqrt{\pi}}  \begin{pmatrix}
 0 & 1 & 0 \\
 0 & 0 & 0 \\
 0 & \ii & 0
\end{pmatrix}.
\end{align}
\end{lemma}
\begin{proof}
First, we transform the jump contour from $E$ to $\widetilde{E}$ by introducing the transformation
$$
\M^{\widetilde{E}}=\M^E \widetilde{\G}, \quad
\widetilde{\G}=\begin{cases}
(\V^{E})^{-1}, & \alpha \in \Omega,\\
\I, & elsewhere.
\end{cases}
$$
The contour $\widetilde{E}=\R \e^{\frac{\pi}{4} \ii}$ and the region $\Omega$ are shown in Figure~\ref{Efig}.
It is easy to verify that $\M^{\widetilde{E}}$ has a jump across $\widetilde{E}$, and its jump matrix $\V^{\widetilde{E}}$ is equal to $\V^E$. Moreover, from the expression of $\V^E$, it can be seen that $\widetilde{\G}$ exhibits the following asymptotic behavior:
\be \label{tildeGasyp}
\widetilde{\G}(s,\alpha)=
\bpm
1& \mathcal{O}(\e^{-c |\alpha| })&0\\
0&1&0\\
0&\mathcal{O}(\e^{-c |\alpha| })&1
\epm,  \quad \alpha
\to \infty.
\ee
Therefore, it suffices to study the asymptotic behavior of $\M^{\widetilde{E}}$ as $\alpha \to \infty$.

First, it is easy to see that  the first and third columns of $\M^{\widetilde{E} }$ are constant and coincide with the corresponding columns of the $3\times 3$ identity matrix, respectively. Furthermore, we find that $\M^{\widetilde{E} }$ has the following form:
\begin{align}\label{MEexplicit}
\M^E(s,\alpha)=
\bpm
1&X_{12}&0\\
0&1&0\\
0&X_{32}&1
\epm,
\end{align}
where $X_{12}(s,\alpha)$ and $X_{32}(s,\alpha)$ are functions to be determined. In fact, the jump condition $\M^{\widetilde{E}}_+=\M^{\widetilde{E}}_- \V^{\widetilde{E}}$ is equivalent to
\be
\begin{cases}
(X_{12})_+-(X_{12})_-=\ii s  \e^{\ii \alpha^2}, & \alpha \in \R \e^{\frac{\pi}{4} \ii}
,\\
X_{12} \to 0, & \alpha \to \infty,
\end{cases}\quad
\begin{cases}
(X_{32})_+-(X_{32})_-=-s  \e^{\ii \alpha^2}, & \alpha \in \R \e^{\frac{\pi}{4} \ii}
,\\
X_{32} \to 0, & \alpha \to \infty.
\end{cases}
\ee
Then, from the Plemelj's formula, it follows that
\be
\begin{aligned}
&X_{12}(s,\alpha)=\frac{1}{2 \pi \ii} \int_{\R \e^{\frac{\pi}{4} \ii}}\frac{\ii s  \e^{\ii \zeta^2}}{\zeta-\alpha} \mathrm{d}\zeta
=\frac{\ii s}{2} \left[
\frac{1}{\pi \ii} \int_{-\infty}^{\infty} \frac{\e^{- \tau^2}}{\tau - \e^{-\frac{\pi}{4} \ii}  \alpha} \mathrm{d} \tau
\right],
\\
&X_{32}(s,\alpha)=\ii X_{12}(s,\alpha).
\end{aligned}
\ee
Using the following integral representation for the error function  $\mathrm{erf}(z)$ (see~\cite[Eq.(7.7.2)]{NIST2022} ):
$$\e^{-z^2}(1 \mp \mathrm{erf}(-\ii z)) = \pm \frac{1}{\pi \ii}\int_{-\infty}^\infty \frac{\e^{-t^2} dt}{t-z}, \qquad  \Im z \gtrless 0,$$
we find
\begin{align}\label{X12explicit}
X_{12}
=
\begin{cases}
\frac{\ii s}{2} \e^{\ii \alpha^2}  \left( 1-\mathrm{erf}(\alpha \e^{-\frac{3\pi}{4}\ii}) \right), &  \mathrm{arg} \alpha \in (\frac{\pi}{4}, \frac{5 \pi}{4}),\\
\frac{\ii s}{2} \e^{\ii \alpha^2}  \left( \mathrm{erf}(\alpha \e^{\frac{\pi}{4}\ii})-1 \right), &  \mathrm{arg} \alpha \in (-\frac{3 \pi}{4}, \frac{ \pi}{4}).
 \end{cases}
\end{align}
Furthermore, for each $\varepsilon > 0$, the error function $\mathrm{erf}{z}$ satisfies the asymptotic formula
\begin{align}\label{erfasymptotics}
\mathrm{erf}(z) & \sim 1- \frac{\e^{-z^2}}{\sqrt{\pi}} \sum_{j=0}^\infty \frac{(\frac{1}{2}-1)(\frac{1}{2}-2) \cdots (\frac{1}{2}-j)}{z^{2j+1}}
= 1- \frac{\e^{-z^2}}{\sqrt{\pi}} \bigg(\frac{1}{z} - \frac{1}{2z^3} + O(z^{-5})\bigg)
\end{align}
as $z \to \infty$ uniformly for $\mathrm{arg} z \in [-\frac{3\pi}{4} + \varepsilon, \frac{3\pi}{4} - \varepsilon]$.
Using (\ref{X12explicit}), the asymptotics (\ref{erfasymptotics}) for $\mathrm{erf} (z)$, and the relation $\mathrm{erf}(-z) = -\mathrm{erf}(z)$, we find the asymptotic expansion  of $\M^{\widetilde{E}}$ as $\alpha \to \infty$.
Asymptotic expansion~\ref{MEasymptotics} can be obtained by combining the asymptotic behavior of $\M^{\widetilde{E}}$ with that of $\widetilde{\G}$(see~\ref{tildeGasyp}).
\end{proof}

\section*{Acknowledgments}
This work is supported by National Natural Science Foundation of China (Grant Nos. 12471234, 12471240, 12571268, 12171439) and Science Foundation of Henan Academy of Sciences (Grant No. 20252319002).\\

\noindent{\bf Data Availability Statement} Data sharing not applicable to this article as no datasets were generated or analyzed during the current study.

\subsection*{Declarations}

{\bf Conflict of interest statement} We declare that there is no conflict of interests.

\end{document}